\documentclass[english,onecolumn]{IEEEtran}
\usepackage{amsmath,amssymb,amsthm}
\usepackage{tikz,enumerate}
\usepackage{cite,mathtools}
\usepackage{etex}
\usepackage{dashbox}
\theoremstyle{plain}

\newtheorem{theorem}{Theorem} 
\newtheorem{proposition}{Proposition} 
\newtheorem{corollary}{Corollary}

\newtheorem{lemma}{Lemma}

\theoremstyle{remark}

\newtheorem{remark}{Remark}

\theoremstyle{definition}
\newtheorem{definition}{Definition}

\newcommand{\rmr}{\mathrm{r}}

\DeclareMathOperator\E{E}
\let\P\relax
\DeclareMathOperator\P{P}

\newcommand{\Nc}{\mathcal{N}}

\newcommand{\Tc}{\mathcal{T}}
\newcommand{\Uc}{\mathcal{U}}
\newcommand{\Vc}{\mathcal{V}}
\newcommand{\Wc}{\mathcal{W}}
\newcommand{\Xc}{\mathcal{X}}
\newcommand{\Yc}{\mathcal{Y}}

\def\eps{\epsilon}
\def\la{\lambda}

\usepackage{mathrsfs,cancel}
\usetikzlibrary{arrows, decorations.markings}
\usetikzlibrary{bayesnet}
\usepackage{tkz-graph}  
\usetikzlibrary{shapes.geometric}
\usetikzlibrary{trees}
\usetikzlibrary{positioning}
\usetikzlibrary{arrows.meta}
\renewcommand{\thesection}{\arabic{section}}
\renewcommand{\thesubsection}{\thesection.\arabic{subsection}}

\allowdisplaybreaks

\makeatletter
\def\@seccntformat#1{\@ifundefined{#1@cntformat}%
   {\csname the#1\endcsname\space}
   {\csname #1@cntformat\endcsname}}
\newcommand\section@cntformat{\thesection.\space}       
\newcommand\subsection@cntformat{\thesubsection.\space} 
\makeatother

\tikzstyle{vecArrow} = [thick, decoration={markings,mark=at position
   1 with {\arrow[semithick]{open triangle 60}}},
   double distance=1.4pt, shorten >= 5.5pt,
   preaction = {decorate},
   postaction = {draw,line width=1.4pt, white,shorten >= 4.5pt}]
\tikzstyle{innerWhite} = [semithick, white,line width=1.4pt, shorten >= 4.5pt]
\tikzstyle{arrow} = [thick,->]
\tikzstyle{mynode} = [rectangle, rounded corners, minimum width=3cm, minimum height=1cm,text centered, draw=black]
\tikzstyle{mynode1} = [rectangle, dashed, rounded corners, minimum width=2.5cm, minimum height=1cm,text centered, draw=black]
\usepackage{hyperref,cleveref}
\hypersetup{
    colorlinks=true,
    linkcolor=blue,
    filecolor=magenta,      
    urlcolor=cyan,
}

\usepackage{etex}

\usepackage[textheight=22 cm, textwidth=17cm, headheight=60pt, headsep=30pt, footskip=32pt]{geometry}

\usepackage[draft]{commenting}

\begin{document}

\title{\huge A Strengthened Cutset Upper Bound on the Capacity of the Relay Channel and Applications}
\author{Abbas El Gamal, Amin Gohari and Chandra Nair}
\maketitle

\begin{abstract}
We develop a new upper bound on the capacity of the relay channel that is tighter than previously known upper bounds. This upper bound is proved using traditional weak converse techniques involving mutual information inequalities and Gallager-type explicit identification of auxiliary random variables. We show that the new upper bound is strictly tighter than all previous bounds for the Gaussian relay channel with non-zero channel gains. When specialized to the relay channel with orthogonal receiver components, the bound resolves a conjecture by Kim on a class of deterministic relay channels. When further specialized to the class of product-form relay channels with orthogonal receiver components, 
the bound resolves a generalized version of Cover's relay channel problem, recovers the recent upper bound for the Gaussian case by Wu et al., and improves upon the recent bounds for the binary symmetric case by Wu et al. and Barnes et al., which were obtained using non-traditional geometric proof techniques. For the special class of a relay channel with orthogonal receiver components, we develop another upper bound on the capacity which utilizes an auxiliary receiver and show that it is strictly tighter than the bound by Tandon and Ulukus. Finally, we show through the Gaussian relay channel with i.i.d. relay output sequence that the bound with the auxiliary receiver can be strictly tighter than our main bound.\footnote{This paper was presented in part at the IEEE Symposium on Information Theory (ISIT) 2021.}

 \end{abstract}

\section{Introduction}

The relay channel, first introduced by van-der Meulen in 1971~\cite{van1971three}, is a canonical model of a multi-hop communication network in which a sender $X$ wishes to communicate to a receiver $Y$ with the help of a relay $(X_\mathrm{r},Y_\mathrm{r})$ over a memoryless channel of the form $p(y,y_\mathrm{r}|x,x_\mathrm{r})$. The capacity of this channel, which is the highest achievable rate from the sender $X$ to the receiver $Y$, is not known in general. In~\cite{cover1979capacity}, the cutset upper bound and several lower bounds on the capacity, later termed decode-forward, partial decode-forward, and compress-forward, were established. The cutset bound was shown to coincide with one or more of these lower bounds for several classes of relay channels, including degraded~\cite{cover1979capacity}, semi-deterministic~\cite{elgamal1982capacity}, and relay channels with orthogonal sender components~\cite{elgamal2005capacity}. In~\cite{Aleksic09}, the cutset upper bound was shown to not be tight, i.e., to be strictly larger than the capacity, for a specific relay channel with orthogonal receiver components. These results and others are detailed in Chapter 19 of \cite{elk11}. 

Motivated by Cover's problem concerning a relay channel with orthogonal receiver components~\cite{cover1987capacity}, a series of recent papers~\cite{wu2017improving, Wu19, Liu19} developed specialized upper bounds that are also tighter than the cutset bound. While the older bounds in~\cite{cover1979capacity,tandon2008new,Aleksic09} use standard weak converse techniques involving basic mutual information inequalities and identities and Gallager-type auxiliary random variable identifications, the recent bounds in~\cite{wu2017improving, Wu19, Liu19} for symmetric Gaussian and binary symmetric relay channels with orthogonal receiver components use more sophisticated arguments involving the blowing-up lemma, spherical rearrangement and hypercontractivity. More recently, Gohari-Nair~\cite{gon22} developed a new upper bound on the capacity of the general relay channel and showed that it can be strictly tighter than the cutset bound. This upper bound combines traditional converse techniques, including identification of auxiliary random variables using past and future channel variable sequences which have been used in previous converse proofs, for example, in~\cite{csiszar1978broadcast,gamal1979capacity}, along with a new idea of introducing auxiliary receivers. 

In this paper, which is a more complete and extended version of~\cite{ISITVersion}, we establish a new upper bound on the capacity of the relay channel. We show through several applications and examples that this bound is tighter than all previous upper bounds, including the cutset bound and the specialized bounds in~\cite{wu2017improving, Wu19, Liu19,gon22}. While our primary upper bound is motivated by the arguments in~\cite{gon22} and uses standard general converse techniques, it turns out that for most of the settings that are considered here improvements over previous bounds can be made without needing to introduce auxiliary receivers. Nonetheless, we provide an upper bound for the relay channel with orthogonal receiver components involving an auxiliary receiver and show through an example that it is strictly tighter than our main bound. 

Although the techniques used to establish the upper bounds in this paper have been employed in many previous works, our contributions are in (i) the judicious manner in which we minimize the discarded terms in the derivations of the constraints, consequently tightening the existing upper bounds on the capacity of the relay channel, and (ii) the rather nontrivial computation of the resulting bound to demonstrate strict improvements over previous bounds for several classes and examples of relay channels. 

\smallskip

\noindent{\bf Organization of the paper and summary of the results}.
In the following section, we introduce the relay channel capacity problem and state and discuss our results, providing proof outlines and a couple of shorter proofs of these results. The rest of the proofs are given in Section \ref{sec:proofs}. To help navigate the paper, Figure~\ref{fig:structure} depicts the classes of relay channels for which we provide new upper bounds and applications with references to the corresponding sections. Note that all our applications are for classes of relay channels \textit{without self-interference} $p(y_\mathrm{r}|x)p(y|x,x_\mathrm{r},y_\mathrm{r})$ because they include and generalize several interesting relay channel settings that have been receiving significant attention in recent years.
\begin{figure}[htpb]
\centering
\begin{tikzpicture}[node distance = 1.5cm]
\node (rc)  [mynode] {Relay channel (Section~\ref{sec:1})};
\node (rwsint)  [mynode,below of=rc,xshift=0cm] {Relay without self-interference (Section~\ref{sec:1})};
\node (gaur)  [mynode1,below of=rwsint,xshift=3.2cm] {Gaussian};
\node (primrel) [mynode,below of=rwsint,xshift=-3cm] {Relay with orthogonal receiver components (Section~\ref{sec:2})};
\node (conk) [mynode,below of=primrel,xshift=4.8cm] {{
i.i.d. output (Section \ref{sec:aux-receiver}) }
};
\node (Alek) [mynode1,below of=primrel,xshift=-4.8cm] {Kim's conjecture};
\node (prod) [mynode,below of=primrel,xshift=-0.75cm] {Product-form (Section~\ref{sec:3})};
\node (gencov) [mynode1,below of=prod,xshift=-3.1cm] {Generalized Cover's problem};
\node (aux) [mynode1,below of=conk,xshift=2cm] {Aux-receiver helps};
\node (gaupri) [mynode1,below of=prod,xshift=0.85cm] {Gaussian};
\node (bscpri) [mynode1,below of=prod,xshift=4.0cm] {Binary Symmetric};
\draw[arrow] (rc) -- (rwsint);
\draw [arrow] (rwsint) -- (gaur);
\draw [arrow] (rwsint) -- (primrel);
\draw [arrow] (primrel) -- (conk);
\draw [arrow] (primrel) -- (Alek);
\draw [arrow] (primrel) -- (prod);
\draw [arrow] (prod) -- (gencov);
\draw [arrow] (prod) -- (gaupri);
\draw [arrow] (prod) -- (bscpri);
\draw [arrow] (conk) -- (aux);
\end{tikzpicture}
\caption{Classes of relay channels for which upper bounds are established. An arrow from box A to box B indicates that class A includes class B. Dashed boxes indicate applications for which we compute the bounds.}
\label{fig:structure}
\end{figure}
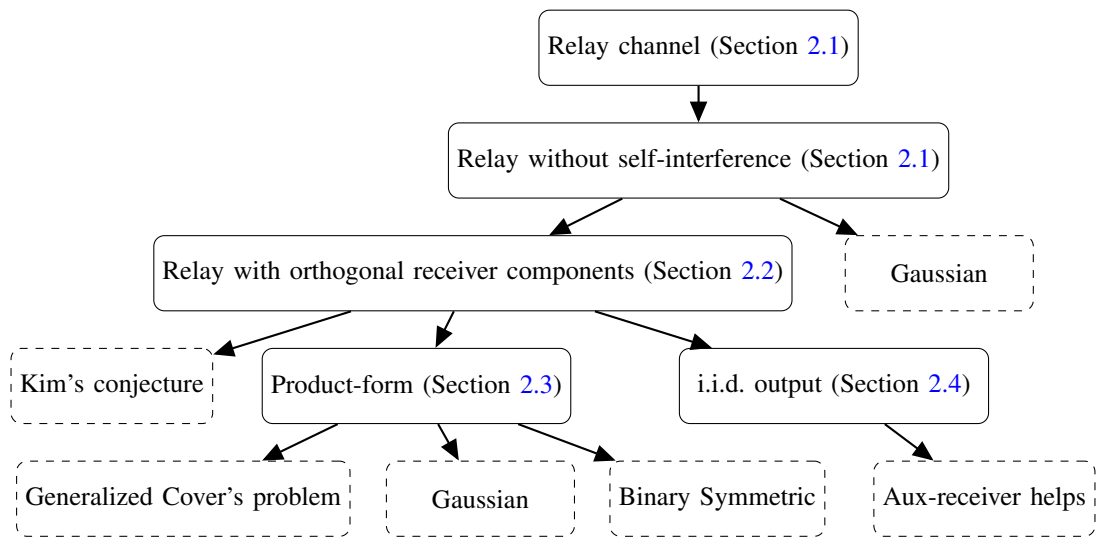

The following is a summary of our results. 
\smallskip

\noindent{\bf Section~\ref{sec:1}}. Theorem~\ref{thm:maintheorem} states our main upper bound. Although the statement and subsequent bounds  and applications of this bound are for relay channels without self-interference, in Remark~\ref{rem:1}, we point out that with a minor relaxation, the upper bound holds for the general relay channel. Corollary~\ref{cor:weakmainOB} is a weakened but easier to evaluate version  of Theorem~\ref{thm:maintheorem} which is used in Theorem~\ref{thm:cutsubopt}  to show the suboptimality of the cutset upper bound and the improved bound in~\cite{gon22} for the Gaussian relay channel; see Figure~\ref{fig1a}.
\smallskip

\noindent{\bf Section~\ref{sec:2}}. Proposition~\ref{prop:mainprirel} provides an equivalent characterization of the upper bound in Theorem~\ref{thm:maintheorem} for relay channels with orthogonal receiver components (also referred to as primitive relay). This proposition is used in Theorem~\ref{thm:KimConjecture} to prove a conjecture by Kim~\cite[Question 2]{kim2007coding}. 
\smallskip

\noindent{\bf Section~\ref{sec:3}}. Theorem~\ref{thm:covopgen} uses
Proposition~\ref{prop:mainprirel} to answer a generalized version of Cover's relay channel problem, posed in~\cite{cover1987capacity}, which extends and improves upon results in~\cite{wu2017improving, Wu19}. Proposition \ref{Proposition5} uses Proposition~\ref{prop:mainprirel} to 
derive an upper bound for the Gaussian case; see Figure~\ref{figGPrimitive}. As shown in Lemma~\ref{thm:GSCase} it indirectly recovers previous results in~\cite{Wu19,wu2017geometry} by showing that the two seemingly different optimization problems evaluate to the same expression. Theorem~\ref{thmBSCPrimitive} uses
Proposition~\ref{prop:mainprirel} to establish a tighter upper bound for a class of symmetric binary relay channels than the bounds in~\cite{wu2017improving,barnes2017solution}; see Figure~\ref{fig1b}. 

\noindent{\bf Section~\ref{sec:aux-receiver}}. Theorem \ref{prop2jm2s} provides a new upper bound on the capacity of the relay channel with orthogonal receiver components using the auxiliary receiver approach in~\cite{gon22}. Corollary~\ref{prop2} specializes Theorem \ref{prop2jm2s} to relay channels with orthogonal receiver components in which the relay output is an i.i.d. sequence, independent of the sender and relay transmissions, previously studied in~\cite{aleksic2009capacity,tandon2008new, ahlswede1983source}. 
This bound is shown to be tight for the class of channels considered in~\cite{aleksic2009capacity} and  strictly tighter than the upper bound by Tandon and Ulukus in~\cite{tandon2008new}. We then show through a Gaussian relay channel with i.i.d. relay output sequence  that the bound in Corollary~\ref{prop2} is strictly tighter than the upper bound in Proposition~\ref{prop:mainprirel}. Hence, including an auxiliary receiver into the upper bound can strictly improve upon our main upper bound in Theorem~\ref{thm:maintheorem}.

\smallskip

\section{Definitions and statement of the results}

We adopt most of our notation from \cite{elk11}. In particular, we use $Y^{i}$ to denote the sequence $(Y_1, Y_2, \cdots, Y_{i})$, and $Y_{i}^j$ to denote $(Y_i, Y_{i+1}, \cdots, Y_j)$. Unless stated otherwise, logarithms are to the base $2$. We use $p(x)$ to indicate the probability mass function of a discrete random variable $X$ and $P_Y$ to indicate the probability distribution of an arbitrary random variable $Y$. We define \(\displaystyle \mathsf{C}(x)= (1/2)\log(1+x)\) for $x \ge 0$.

The discrete memoryless relay channel depicted in Figure~\ref{setup:relay} consists of four alphabets $\mathcal{X}$, $\mathcal{X}_\mathrm{r}$, $\mathcal{Y}_\mathrm{r}$,  $\mathcal{Y}$, and a collection of conditional pmfs $p(y_{\mathrm{r}},y|x,x_{\mathrm{r}})$ on $\mathcal{Y}_{\mathrm{r}}\times \mathcal{Y}$.  A $(2^{nR},n)$ code for this channel  consists of a message set $[1:2^{nR}]$, an encoder that assigns a codeword $x^n(m)$ to each message $m
  \in [1: 2^{nR}]$, a relay encoder that assigns a symbol $x_{\mathrm{r}i}(y_\mathrm{r}^{i-1})$ to each past received sequence $y_\mathrm{r}^{i-1}$ for each time $i\in [1:n]$, and a decoder that assigns an estimate $\hat{M}$ or an error message $\varepsilon$ to each received sequence $y^n$. We assume that the message $M$ is uniformly distributed over $[1: 2^{nR}]$. The average probability of error is defined as $P_e(n) = \P\{\hat{M} \neq M\}$. A rate $R$ is said to be achievable for the discrete memoryless relay channel if there exists a sequence of $(2^{nR}, n)$ codes such that $\lim_{n\rightarrow\infty} P_e(n) = 0$. The
capacity $C$ of the discrete memoryless relay channel is the supremum of all achievable rates.

\begin{figure}[htpb]
\centering
	\begin{tikzpicture}
	\node at (-0.95,-0.95) {$M$};
	\draw [->,thick] (-0.6,-1) -- (0.2,-1);
	\draw (0.2,.-1.5) rectangle +(1.6,1); \node at (1,-1) {Encoder};
	\draw [->,thick] (1.8,-1)-- (3.1,-1); \node at (2.45,-0.7) {$X^n$};
	\draw [->,thick] (3.6,-0.5)-- (3.6,0.75); \node at (3.1,0.25) {$Y_{\mathrm{r}}^{i-1}$};
	\draw [<-,thick] (4.9,-0.5)-- (4.9,0.75); \node at (5.3,0.25) {$X_{\mathrm{r}i}$};
	\draw (3.1,-1.5) rectangle +(2.4,1); \node at (4.3,-1) {$p(y,y_{\mathrm{r}}|x,x_{\mathrm{r}})$};
	\draw (3.1,0.75) rectangle +(2.4,1); \node at (4.3,1.25) {Relay Encoder};
	\draw [->,thick] (5.5,-1) --(6.6,-1); \node at (6.05,-.7) {$Y^n$};
	\draw (6.6,-1.5) rectangle +(1.6,1); \node at (7.4,-1) {Decoder};
	\draw [->,thick] (8.2,-1) --(9.0,-1); \node at (9.25,-0.945) {$\hat M$};
	\end{tikzpicture}
	\caption{Discrete memoryless relay channel setup.}
	\label{setup:relay}
\end{figure}
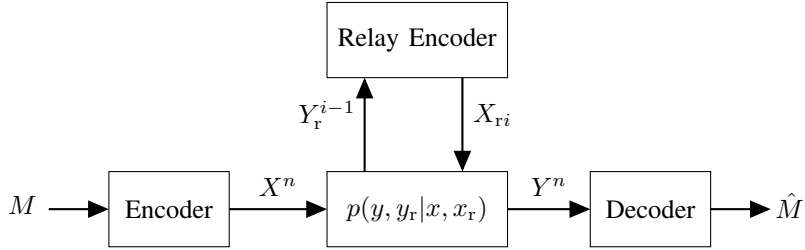

\subsection{Upper bound for the relay channel without self-interference}
\label{sec:1}

A relay channel is said to be \textit{without self-interference} if $p(y,y_\mathrm{r}|x,x_\mathrm{r})=p(y_\mathrm{r}|x)p(y|x,x_\mathrm{r},y_\mathrm{r})$. We establish the following upper bound on the capacity of this class of relay channels.
\begin{theorem} Any achievable rate $R$ for a discrete memoryless relay channel without self-interference $p(y_\mathrm{r}|x)p(y|x,x_\mathrm{r},y_\mathrm{r})$ must satisfy the following inequalities 
\begin{align}
    R&\leq I(X;Y,Y_{\mathrm{r}}|X_{\mathrm{r}})-I(U;Y|X_{\mathrm{r}},Y_{\mathrm{r}}), \label{eqnUB1}\\
    &{
    =I(X;Y|X_\rmr,U)+I(U;Y_\rmr|X_\rmr)+I(X;Y_\rmr|X_\rmr,U,Y)
    }\label{eqnUB1.1}
    \\
R&\leq I(X;Y,Y_{\mathrm{r}}|X_{\mathrm{r}})-I(V;Y|X_{\mathrm{r}},Y_{\mathrm{r}})-I(X;Y_{\mathrm{r}}|V,X_{\mathrm{r}},Y)\label{eqnUB1.5}\\
 &=  I(X;Y_{\mathrm{r}}|X_{\mathrm{r}})+I(X;Y|V,X_{\mathrm{r}})-I(X;Y_{\mathrm{r}}|V,X_{\mathrm{r}})\label{eqnUB2}\\
 &=I(X;Y,V|X_\rmr)-I(V;X|X_\rmr,Y_\rmr),\label{eqnUB2m1} \\
R&\leq I(X,X_{\mathrm{r}};Y)-I(V;Y_{\mathrm{r}}|X_{\mathrm{r}},X,Y),\label{eqnUB0}
 \end{align}
for some $p(u,x,x_{\mathrm{r}})p(y,y_{\mathrm{r}}|x,x_{\mathrm{r}})p(v|x,x_{\mathrm{r}},y_{\mathrm{r}})$ satisfying
\begin{align}
    I(V,X_{\mathrm{r}};Y_{\mathrm{r}})-I(V,X_{\mathrm{r}};Y)&=I(U;Y_{\mathrm{r}})-I(U;Y).\label{eq:UBcon}\end{align}
Moreover, we obtain an equivalent characterization of the bound if we drop the constraint on $R$ in \eqref{eqnUB0} and strengthen \eqref{eq:UBcon} to
\begin{align}
    I(V,X_{\mathrm{r}};Y_{\mathrm{r}})-I(V,X_{\mathrm{r}};Y)&=I(U;Y_{\mathrm{r}})-I(U;Y)\leq I(X_\rmr;Y_\rmr).\label{eq:UBcon2} 
    \end{align} 
In both characterizations of the bound, it suffices to consider 
 $|\Vc| \leq |\Xc| |\Xc_{\mathrm{r}}||\Yc_{\mathrm{r}}|+2$ and
$|\Uc| \leq |\Xc| |\Xc_{\mathrm{r}}|+1$.
 \label{thm:maintheorem}
\end{theorem}

We begin with a few remarks and then provide an outline of the proof. The details of the proof can be found in Section \ref{sec:proofThm1}.

\begin{remark}In Appendix \ref{AppendixNew2}, another equivalent form of the bound in Theorem \ref{thm:maintheorem} without the equality constraint \eqref{eq:UBcon} is given.
\end{remark}

\begin{remark}
\label{rem:1}
As the proof of the theorem shows, the  without self-interference assumption is used only to establish the Markov chain relationship $U\rightarrow X,X_\rmr\rightarrow Y,Y_\rmr$. Hence, by removing this condition, the above theorem readily extends to the general relay channel.
\end{remark}

\begin{remark} It is immediate that the above bound is at least as tight as the cutset upper bound in~\cite{cover1979capacity},
\begin{align}
    C\leq \max_{p(x,x_\rmr)} \min \{I(X_\rmr,X;Y),I(X;Y,Y_\rmr|X_\rmr)\}.
    \label{eqnCT1}
\end{align}
\end{remark}
\begin{remark}
\label{rem:comforrem}
From \eqref{eqnUB2m1} and  \eqref{eqnUB0}, we deduce that any achievable rate $R$ must satisfy the condition
\[
R\leq \max_{p(x,x_\rmr)p(v|x,x_\rmr,y_\rmr)}\min\big\{I(X;Y,V|X_\rmr)-I(V;X|X_\rmr,Y_\rmr), I(X,X_{\mathrm{r}};Y)-I(V;Y_{\mathrm{r}}|X_{\mathrm{r}},X,Y)\big\}.
\]
If we replace the maximum over $p(x,x_\rmr)p(v|x,x_\rmr,y_\rmr)$ with a maximum over the more restrictive joint distributions of the form  $p(x)p(x_\rmr)p(v|x_\rmr,y_\rmr)$, we obtain the equivalent form of the compress-forward lower bound without time-sharing random variable $Q$ in~\cite{elk11},
\begin{equation}
\label{eq:cflb}
C \ge \max_{p(x)p(x_\rmr)p(v|x_\rmr,y_\rmr)}\min\big\{I(X;Y,V|X_\rmr), I(X,X_{\mathrm{r}};Y)-I(V;Y_{\mathrm{r}}|X_{\mathrm{r}},X,Y)\big\}.
\end{equation}
Thus, the auxiliary random variable $V$ in the upper bound seems to have a correspondence to the auxiliary random variable in the compress-forward lower bound. 
\end{remark}

\begin{remark}
\label{rem:pdfrem}
Recall the partial decode-forward lower bound on the capacity of the relay channel in~\cite{cover1979capacity},
\begin{equation}
\label{eq:pdflb}
C \ge \max_{p(u,x,x_\rmr)}\min\big\{ I(X,X_{\mathrm{r}};Y),~I(X;Y|X_\rmr, U)+I(U;Y_\rmr|X_\rmr)\big\}.
\end{equation}
The constraint \eqref{eqnUB1.1} has the terms $I(X;Y|X_\rmr,U)+I(U;Y_\rmr|X_\rmr)$ of the partial decode-forward lower bound along with an excess term $I(X;Y_\rmr|X_\rmr,U,Y)$. Therefore.
the auxiliary random variable $U$ in the upper bound seems to have a correspondence to the auxiliary random variable in the partial decode-forward lower bound. 
\end{remark}
\begin{remark}\label{remark3}
If the cutset bound is tight for a relay channel of the form $p(y,y_\rmr |x,x_\rmr)=p(y|x,x_\rmr)p(y_\rmr|x)$ and the capacity coincides with the MAC bound in the cutset bound, i.e., $C=I(X_\rmr,X;Y)$ for the maximizing $p(x,x_\rmr)$, then the capacity is achievable by partial decode-forward, since from \eqref{eqnUB0} we obtain that $I(V;Y_\rmr|X_\rmr,X,Y)=0$. The assumption $p(y,y_\rmr |x,x_\rmr)=p(y|x,x_\rmr)p(y_\rmr|x)$ implies that $I(V;Y_\rmr|X_\rmr,X)=0$. The constraint $V\rightarrow X,X_\rmr\rightarrow Y_\rmr$, then implies that \eqref{eqnUB1.5} is equivalent to 
\[
R\leq I(X;Y,Y_\rmr|X_\rmr)-I(V;Y|X_\rmr,Y_\rmr)-I(X;Y_\rmr|V,X_\rmr,Y)= I(X;Y|V,X_\rmr)+I(V;Y_\rmr|X_\rmr),
\]
which is the partial decode-forward lower bound (with auxiliary random variable $V$). This also implies that the correspondence between the auxiliary random variables $U,V$ appearing in the outer bound and those in the two inner bounds mentioned in Remarks \ref{rem:comforrem} and \ref{rem:pdfrem} are not exact.

As a concrete application, consider the following example which has the same flavor as the one introduced by Cover in \cite{cover1987capacity}: consider a relay channel of the
form $X=(X_a, X_b)$ and $Y_\rmr=(Y_{\rmr a}, Y_{\rmr b})$,
where $p(y,y_\rmr|x,x_\rmr)=p(y_{\rmr a}|x_a) p(y|x_a,x_\rmr)p(y_{\rmr b}|x_b)$, and 
 $p(y_{\rmr b}|x_b)$ is a channel with capacity $C_1$ while $p(y_{\rmr a}|x_a)$ and $p(y|x_a,x_\rmr)$ are arbitrary. 
 In Appendix \ref{AppendixNew}, we show that replacing $p(y_{\rmr a}|x_a)$ with a noiseless link of capacity $C_1$ does not alter the capacity.
 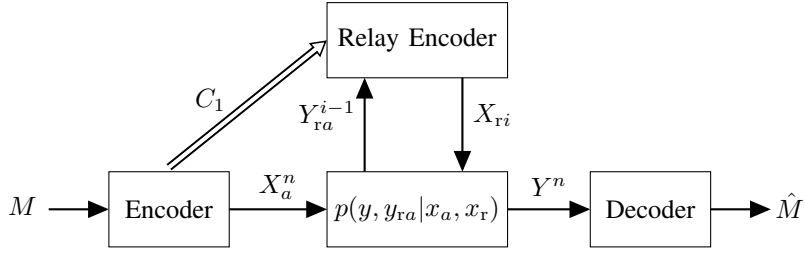
\begin{figure}[htpb]
\centering
	\begin{tikzpicture}
	\node at (-0.95,-0.95) {$M$};
	\draw [->,thick] (-0.6,-1) -- (0.2,-1);
	\draw (0.2,.-1.5) rectangle +(1.6,1); \node at (1,-1) {Encoder};
	\draw[vecArrow]  (1,-0.44) -- (3.1,1.25) ; \node at (1.55,0.45) {$C_1$};
	\draw [->,thick] (1.8,-1)-- (3.1,-1); \node at (2.45,-0.7) {$X_a^n$};
	\draw [->,thick] (3.6,-0.5)-- (3.6,0.75); \node at (3.1,0.25) {$Y_{\mathrm{r}a}^{i-1}$};
	\draw [<-,thick] (4.9,-0.5)-- (4.9,0.75); \node at (5.3,0.25) {$X_{\mathrm{r}i}$};
	\draw (3.1,-1.5) rectangle +(2.4,1); \node at (4.3,-1) {$p(y,y_{\mathrm{r}a}|x_a,x_{\mathrm{r}})$};
	\draw (3.1,0.75) rectangle +(2.4,1); \node at (4.3,1.25) {Relay Encoder};
	\draw [->,thick] (5.5,-1) --(6.6,-1); \node at (6.05,-.7) {$Y^n$};
	\draw (6.6,-1.5) rectangle +(1.6,1); \node at (7.4,-1) {Decoder};
	\draw [->,thick] (8.2,-1) --(9.0,-1); \node at (9.25,-0.945) {$\hat M$};
	\end{tikzpicture}
	\caption{Relay channel with noiseless link from sender to relay.}
	\label{setup:Cover-like}
\end{figure}
 In other words, we have a noiseless link of capacity $C_1$ from the sender to the relay in parallel to the  channel $p(y_{\rmr a}|x_a) p(y|x_a,x_\rmr)$ as depicted in Figure \ref{setup:Cover-like}.  
Let $\mathcal{C}(C_1)$ be the capacity of this relay channel in terms of $C_1$ for a fixed channel $p(y_{\rmr a}|x_a) p(y|x_a,x_\rmr)$.
For $C_1=\infty$, we have  $\mathcal{C}(\infty)=\max_{p(x, x_\rmr)}I(X,X_\rmr;Y)=\max_{p(x_a, x_\rmr)}I(X_a,X_\rmr;Y)$. One can then ask what is the critical value $C_1^*$ such that $C_1^*=\inf \{C_1:\mathcal{C}(C_1) = \mathcal{C}(\infty)\}$? It is immediate from the above discussion  that $C_1^*$ can be characterized as the minimum value $C_1$ such that $\mathcal{C}(\infty)$ is equal to the rate achieved by partial decode-forward:
\[R_{\mathsf{PDF}}(C_1)=
\max_{p(u,x_a,x_\rmr)}\min\big\{ I(X_a,X_{\mathrm{r}};Y),~I(X_a;Y|X_\rmr, U)+C_1+I(U;Y_{\rmr a}|X_\rmr)\big\}.
\]

\end{remark}

\textit{Outline of the proof for Theorem \ref{thm:maintheorem}}:
As in every weak converse, the starting point of the proof is to consider the joint distribution $(M,X^n,X_\rmr^n,Y^n,Y_\rmr^n)$ induced by any given code. Since the message $M$ is uniform on its support and can be recovered from $Y^n$ with a probability of error bounded by $\eps_n$, by Fano's inequality we have
\begin{equation}\label{eq:Fano}
nR =   H(M)\leq I(M;Y^n) + n\eps_n.\end{equation} 
Note that  the broadcast bound of the cutset bound is obtained as follows \cite[Sec 16.2]{elk11}
\begin{align*}I(M;Y^n)&\leq I(M;Y^n, Y_\rmr^n)\\
&\leq I(X^n;Y^n,Y_\rmr^n)\\
& = \sum_{i=1}^n \left(I(X_i;Y_i,Y_{\rmr,i}|X_{\rmr,i}) - I(Y^{i-1},Y_\rmr^{i-1};Y_i,Y_{\rmr,i}|X_{\rmr,i})\right)\\
&\leq \sum_{i=1}^n I(X_i;Y_i,Y_{\rmr,i}|X_{\rmr,i})\\
& = n I(X_Q;Y_Q,Y_{\rmr,Q}|X_{\rmr,Q},Q) \\
& \leq n I(X_Q;Y_Q,Y_{\rmr,Q}|X_{\rmr,Q}),
\end{align*}
which yields the single-letter term $I(X;Y,Y_{\rmr}|X_{\rmr})$ in which $Y$ and $Y_{\rmr}$ are lumped together. In the above, $Q$ is a random variable independent of other random variables and uniformly distributed in $[1:n]$. The various inequalities above can be justified using the chain-rule, the data-processing inequality, and the non-negativity of the mutual information.
Observe that while obtaining this cutset bound, the term $\frac{1}{n} \sum_{i=1}^n I(Y^{i-1},Y_\rmr^{i-1};Y_i,Y_{\rmr,i}|X_{\rmr,i})$ is discarded from the rate, hence can cause the bound to be loose.

In our proof, we instead begin from
\begin{align*}I(M;Y^n)&= I(M;Y_\rmr^n)+I(M;Y^n)-I(M;Y_\rmr^n),
\end{align*}
and single-letterize $I(M;Y_\rmr^n)$ and $I(M;Y^n)-I(M;Y_\rmr^n)$ separately. We expand the second term, which represents the information \textit{not} recovered by the relay, using the Csisz\'ar-K\"orner-Marton sum lemma,  in two ways
\begin{align*}
    I(M;Y^n)-I(M;Y_\rmr^n)&=\sum_{i}I(M;Y_i|{ Y_{\rmr i+1}^n,Y^{i-1}})-\sum_{i}I(M;Y_{\rmr i}|{Y_{\rmr i+1}^n,Y^{i-1}})
   \\&=\sum_{i}I(M;Y_i|{Y_{ i+1}^n,Y_\rmr^{i-1}})-\sum_{i}I(M;Y_{\rmr i}|{Y_{ i+1}^n,Y_\rmr^{i-1}}).
\end{align*}
This naturally suggests the identifications of the auxiliary random variables $V$ and $U$ as
\begin{align*}
V_i&=(Y_{i+1}^n, Y_{\mathrm{r}}^{i-1}),\quad
    U_i =  (Y_{\mathrm{r}i+1}^n, Y^{i-1}),
\end{align*}
The rest of the proof in Section~\ref{sec:proofThm1} involves bounding similar carefully selected  terms and using standard mutual information inequalities and identities to establish the stated bounds on the rate $R$.

The following corollary is an immediate weakening of the upper bound in Theorem \ref{thm:maintheorem}, which is easier to evaluate for the Gaussian relay channel.
\begin{corollary}
\label{cor:weakmainOB}
Any achievable rate $R$ for a discrete memoryless relay channel without interference $p(y_\mathrm{r} |x)p(y|x, x_\mathrm{r},y_\mathrm{r} )$ must satisfy the following conditions 
\begin{align}
R&\leq I(X;Y,Y_{\mathrm{r}}|X_{\mathrm{r}})-I(V;Y|X_{\mathrm{r}},Y_{\mathrm{r}})-I(X;Y_{\mathrm{r}}|V,X_{\mathrm{r}},Y)\label{eqnCoreq1}
\\&=
I(X;Y_{\mathrm{r}}|X_{\mathrm{r}})
-I(X;Y_{\mathrm{r}}|V,X_{\mathrm{r}})+I(X;Y|V,X_{\mathrm{r}})\label{eqnCoreq2}
 \end{align}
for some $p(x,x_{\mathrm{r}})p(y,y_{\mathrm{r}}|x,x_{\mathrm{r}})p(v|x,x_{\mathrm{r}},y_{\mathrm{r}})$ satisfying
\begin{align*}
    I(V,X_{\mathrm{r}};Y_{\mathrm{r}})-I(V,X_{\mathrm{r}};Y)&\leq {\min\left[I(X_\rmr;Y_\rmr),~\max_{p(u|x,x_{\mathrm{r}})}\left(I(U;Y_{\mathrm{r}})-I(U;Y)\right)\right]}.
\end{align*}
Further it suffices to consider $|\Vc| \leq |\Xc| |\Xc_{\mathrm{r}}||\Yc_{\mathrm{r}}|+2$.
\end{corollary}
The corollary is established by weakening the alternate form of the upper bound in Theorem \ref{thm:maintheorem} simply by removing  the constraint in \eqref{eqnUB1} and relaxing the condition in \eqref{eq:UBcon2}.
\begin{remark}\label{rmknw7} The above bound is a strengthening of the bound given in \cite[Theorem 1]{gon22} for the choice of $J=Y_{\mathrm{r}}$. To see this, observe that  from \eqref{eqnUB2} we have
\begin{align}
    R&\leq
    I(X;Y_{\mathrm{r}}|X_{\mathrm{r}})+I(X;Y|V,X_{\mathrm{r}})-I(X;Y_{\mathrm{r}}|V,X_{\mathrm{r}})\nonumber
    \\&=
    I(X;Y_{\mathrm{r}}|X_{\mathrm{r}})+I(V,X_{\mathrm{r}},X;Y)-I(V,X_{\mathrm{r}},X;Y_{\mathrm{r}})-I(V,X_{\mathrm{r}};Y)+I(V,X_{\mathrm{r}};Y_{\mathrm{r}})\nonumber
    \\&
    \leq 
    I(X;Y_{\mathrm{r}}|X_{\mathrm{r}})+I(V,X_{\mathrm{r}},X;Y)-I(V,X_{\mathrm{r}},X;Y_{\mathrm{r}})+\max_{p(u|x,x_{\mathrm{r}})}\left[I(U;Y_{\mathrm{r}})-I(U;Y)\right]\label{eqnAAA}
    \\&
    =
    I(X;Y_{\mathrm{r}}|X_{\mathrm{r}})+\max_{p(u|x,x_{\mathrm{r}})}\left[I(X,X_{\mathrm{r}};Y|U)-I(X,X_{\mathrm{r}};Y_{\mathrm{r}}|U)\right]-I(V;Y_{\mathrm{r}}|X,X_{\mathrm{r}},Y)\nonumber
    \\&
    \leq
    I(X;Y_{\mathrm{r}}|X_{\mathrm{r}})+\max_{p(u|x,x_{\mathrm{r}})}\left[I(X,X_{\mathrm{r}};Y|U)-I(X,X_{\mathrm{r}};Y_{\mathrm{r}}|U)\right], \label{eqnAAA2}
\end{align}
where \eqref{eqnAAA} follows from \eqref{eq:UBcon}. The expression in \eqref{eqnAAA2} is the bound given in \cite[Theorem 1]{gon22} for the choice of $J=Y_{\mathrm{r}}$.
\end{remark}

\smallskip

\noindent{\bf Gaussian relay channel.} The Gaussian relay channel is defined by
 \begin{equation} 
\begin{split}
Y_\mathrm{r} &= g_{12} X+Z_\mathrm{r},\\
Y &= g_{13} X+g_{23} X_\mathrm{r}+Z,
\end{split}\label{eqnGRC}
\end{equation}
where $g_{12}, g_{13}$, and $g_{23}$ are channel gains, and 
$Z \sim \Nc(0,1)$ and $Z_\mathrm{r} \sim \Nc(0,1)$ are
independent noise components. We assume average power constraint $P$ on each of $X$ and $X_\rmr$.
\begin{remark}
Note that as defined, the Gaussian relay channel belongs to the class of relay channels  without self-interference.
\end{remark}
In the following discussion, we use the SNRs $S_{12}=g_{12}^2P$, $S_{13}=g_{13}^2P$ and $S_{23}=g_{23}^2P$ to characterize the Gaussian relay channel. 
Recall that the cutset bound for the Gaussian relay channel reduces to~\cite{Cover79}
  \begin{align}
  C&\leq \max_{0\leq \rho\leq 1}\min \big\{\mathsf{C}(S_{13}+S_{23}+2\rho\sqrt{S_{13}S_{23}}\big),
\mathsf{C}((1-\rho^2)(S_{13}+S_{12}))\big\} \nonumber
\\&=\begin{cases}
 \mathsf{C}\Big(\frac{\left(\sqrt{S_{12}S_{23}}+\sqrt{S_{13}(S_{13}+S_{12}-S_{23})}\right)^2}{S_{13}+S_{12}}\Big) & \text{if }S_{12}\geq S_{23},\\[3pt]
\mathsf{C}(S_{13}+S_{12}) &\text{otherwise.} \label{eq:csrelay}
\end{cases}
\end{align}
Also recall that the decode-forward lower bound (evaluated using Gaussian distributions) for the Gaussian relay channel reduces to~\cite[Eq. 16.6]{elk11}
\begin{align}
\label{eq:dec-forwd-gau}
C &\geq \begin{cases}
 \mathsf{C}\Big(\frac{\left(\sqrt{S_{13}(S_{12}-S_{23})}+\sqrt{S_{23}(S_{12}-S_{13})}\right)^2}{S_{12}}\Big) & \text{if }S_{12}\geq S_{23}+S_{13},\\[3pt]
\mathsf{C}(S_{12}) &\text{otherwise,} 
\end{cases}
\end{align}
while the compress-forward lower bound (evaluated using Gaussian distributions) for the Gaussian relay channel reduces to~\cite[Eq. 16.12]{elk11}
\begin{align}
\label{eq:comp-forwd-gau}
C &\geq\mathsf{C}\left(S_{13}+\frac{S_{12}S_{23}}{S_{13}+S_{12}+S_{23}+1}\right).
\end{align}
 The compress–forward lower bound outperforms decode–forward if and only if $S_{12}(1 +
S_{12})\leq S_{13}(S_{13}+S_{23}+1)$. It is worth noting that when $S_{12}(1 +
S_{12})= S_{13}(S_{13}+S_{23}+1)$, i.e., when compress–forward and decode–forward yield the same rate, the
mixed strategy lower bound in \cite[Theorem 7]{Cover79} evaluated using Gaussian distributions strictly improves upon both the compress–forward and decode–forward strategies~\cite[Theorem 3]{Luo13}. It is also known that for small values of $S_{12}$, a simple two-letter amplify-forward strategy outperforms both compress–forward and decode–forward strategies~\cite[Example 1]{el2006bounds}. Additionally, for certain values of channel gains, the lower bound of Chong, Motani and Garg~\cite{ChongMotaniGarg} improves upon the
mixed strategy lower bound in~\cite{Cover79} evaluated using Gaussian distributions~\cite[Remark 6]{ChongMotaniGarg}.

We are now ready to show that our upper bound is strictly tighter than the cutset bound.
\begin{theorem}\label{thm:cutsubopt} 
For the Gaussian relay channel, the bound in Corollary \ref{cor:weakmainOB}  is strictly tighter than the cutset upper bound for every non-zero values of $S_{12},S_{13},S_{23}$. Furthermore, the bound reduces to the following.
Any achievable rate $R$ for the Gaussian relay channel must satisfy:
\begin{align}
R&\leq 
\frac12\log\left( (1-\rho^2)S_{12}+1\right)
-\frac12\log\left(\beta + S_{12}(1-\rho^2)\alpha + 2\sigma\sqrt{S_{12}(1-\rho^2)\alpha\beta}  \right)+\nonumber
\\&\quad
+\frac12\log\left( \beta(1-\sigma^2)\right)
+\frac12\log\left( (1-\rho^2)\alpha S_{13}+1\right)
 \end{align}
for some 
$0\leq \alpha,\beta\leq 1$, $\rho\in[-1,1]$ such that 
$(1-\alpha)(1-\beta)\geq \sigma^2\alpha \beta$, 
where
\begin{align*}
    &\sigma 
=
\frac{(1-\rho^2)\alpha S_{13}
+1}{2T\sqrt{S_{12}(1-\rho^2)\alpha \beta}}-\frac{(1-\rho^2)\alpha S_{12}
+\beta}{2\sqrt{S_{12}(1-\rho^2)\alpha \beta}},
\end{align*}
\[
T=\min\left[\frac{1+S_{13}+S_{23}+2\rho\sqrt{S_{13}S_{23}} }{(1-\rho^2) S_{12}
+1},~~
 \lambda_{\max}
\right],
\]
and $\lambda_{\max}$ is the larger root of the quadratic equation
\begin{align}
2\rho \sqrt{S_{13}S_{23}}+S_{13}+S_{23}+1-\lambda\big(S_{23}S_{12}(1-\rho^2)+S_{13}+S_{23}+S_{12}+2+2\rho\sqrt{S_{13}S_{23}}\big)+\lambda^2(S_{12}+1)=0.\end{align}

\end{theorem}
The proof of this theorem is given in Section \ref{sec:proof:thm:cutsubopt}. Since the new upper bound implies the cutset upper bound, the proof idea is to arrive at a contradiction assuming that the two upper bounds matched. As an outline, recall that the cutset bound in~\eqref{eqnCT1} is maximized by a jointly Gaussian $(X,X_\mathrm{r})$ and that at the maximizing distribution, the bound coincides with the broadcast bound $I(X;Y,Y_\rmr|X_\rmr)$ \cite[Section 16.5]{elk11}. Our bound involves the inequality 
\begin{align*}
    R&\leq I(X;Y,Y_{\mathrm{r}}|X_{\mathrm{r}})-I(V;Y|X_{\mathrm{r}},Y_{\mathrm{r}})-I(X;Y_{\mathrm{r}}|V,X_{\mathrm{r}},Y).
\end{align*}
If the cutset bound is equal to the new upper bound, then it must be that
$I(V;Y|X_{\mathrm{r}},Y_{\mathrm{r}})=I(X;Y_{\mathrm{r}}|V,X_{\mathrm{r}},Y)=0$. These two conditions are shown to imply that given $X_{\mathrm{r}}=x_{\mathrm{r}}$, $V\rightarrow Y_{\mathrm{r}}\rightarrow X$ and $Y_{\mathrm{r}}\rightarrow V\rightarrow X$.
    Given $X_{\mathrm{r}}=x_{\mathrm{r}}$, this ``double Markovity condition" implies that $X$ is independent of $(V,Y_\rmr)$. However, this is a contradiction since $X$ is not independent of $Y_\rmr$ given $X_{\mathrm{r}}=x_{\mathrm{r}}$. This allows us to conclude that the new bound strictly improves upon the cutset bound.
The rest of the proof uses the sub-additivity, doubling, and rotation techniques from \cite{gen14} to show that our upper bound is maximized by jointly Gaussian random variables.

Figure \ref{fig1a} compares the bound in Corollary \ref{cor:weakmainOB} to the bound in \cite[Proposition 1]{gon22} for the scalar Gaussian relay channel, the cutset bound, and the compress-forward lower bound (evaluated using Gaussian distributions). Note that for the example in Figure \ref{fig1a}, compress-forward outperforms  decode-forward. 

\begin{figure}
    \centering
    \includegraphics[scale=0.4]{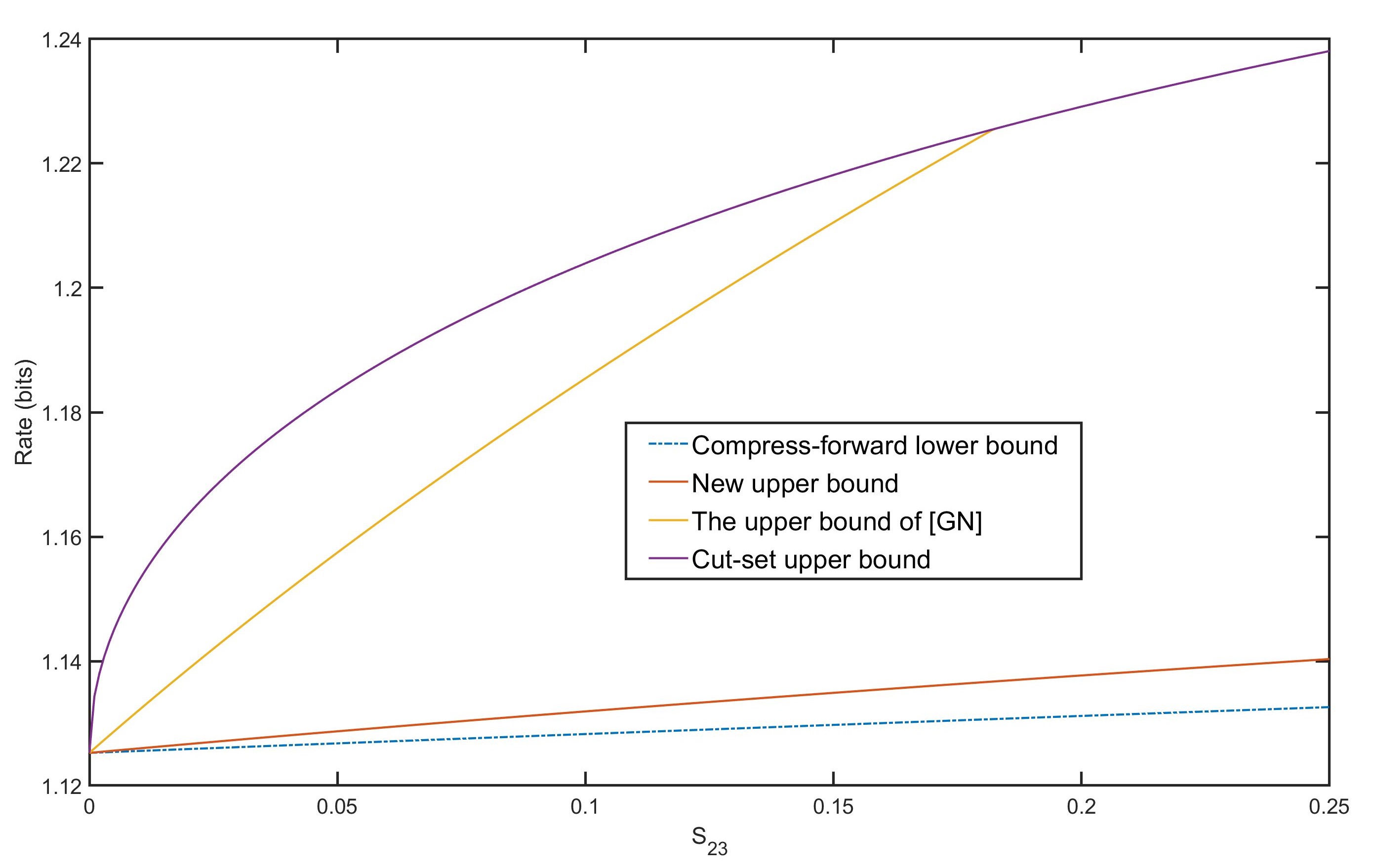}
    \caption{Plots of the bounds for the Gaussian relay channel with $S_{13}=3.7585$, $S_{12}=1.2139$. The new upper bound is the bound in Corollary \ref{cor:weakmainOB}. }
    \label{fig1a}
\end{figure}

\subsection{Relay channels with orthogonal receiver components}
\label{sec:2}

In this section we present results for the following sub-class of relay channels without self-interference.
\begin{definition} 
A relay channel is said to be \textit{with orthogonal receiver components} (also referred to as \textit{primitive}) (see Section 16.7.3 in~\cite{elk11}) if $Y=(Y_1, Y_2)$, where $p(y_1, y_2, y_{\mathrm{r}}|x,x_{\mathrm{r}})=p(y_1,y_{\mathrm{r}}|x)p(y_2|x_{\mathrm{r}})$. 
It is known, see \cite{kim2007coding}, that in this case relaxing the strictly causal relay function $x_{\rmr i}(y_\rmr^{i-1})$ to a non-causal relay function $x_{\rmr i}(y_\rmr^n)$ does not change the capacity. Furthermore,
the capacity of this relay channel depends on $p(y_2|x_{\mathrm{r}})$ only through the capacity of the point-to-point channel $p(y_2|x_{\mathrm{r}})$. Hence 
we can substitute the relay-to-receiver channel $p(y_2|x_{\mathrm{r}})$ with a noiseless link of the same capacity $C_0$~\cite{kim2007coding} as shown in Figure \ref{setup:primitive relay}. For completeness, we give an outline of the proof in  Appendix \ref{AppendixNew}.
\end{definition}
\begin{figure}\centering
	\begin{tikzpicture}
	\node at (-0.95,-0.95) {$M$};
	\draw [->,thick] (-0.6,-1) -- (0.2,-1);
	\draw (0.2,.-1.5) rectangle +(1.6,1); \node at (1,-1) {Encoder};
	\draw [->,thick] (1.8,-1)-- (3.1,-1); \node at (2.45,-0.7) {$X^n$};
	\draw [->,thick] (4.4,-0.5)-- (4.4,0.5); \node at (4.8,0) {$Y_{\mathrm{r}}^n$};
	\draw[vecArrow]  (5.53,1) -- (7.1,-0.5) ; \node at (6.55,0.45) {$C_0$};
	\draw (3.1,-1.5) rectangle +(2.4,1); \node at (4.3,-1) {$p(y_1,y_{\mathrm{r}}|x)$};
	\draw (3.1,0.5) rectangle +(2.4,1); \node at (4.3,1) {Relay Encoder};
	\draw [->,thick] (5.5,-1) --(6.4,-1); \node at (5.95,-.7) {$Y_{1}^n$};
	\draw (6.4,-1.5) rectangle +(1.6,1); \node at (7.2,-1) {Decoder};
	\draw [->,thick] (8.0,-1) --(8.8,-1); \node at (9.15,-0.945) {$\hat M$};
	\end{tikzpicture}
		\caption{Relay channel with orthogonal receiver components.}
	\label{setup:primitive relay}
\end{figure}
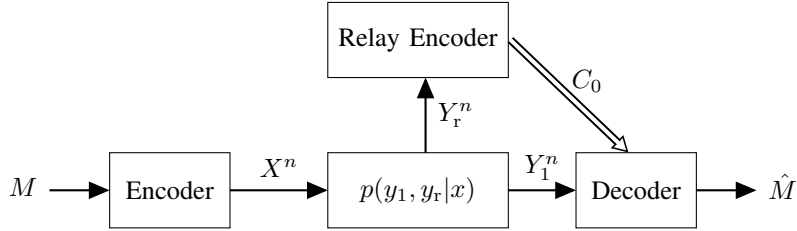

The following  provides an equivalent characterization of the upper bound in Theorem \ref{thm:maintheorem} for relay channels with orthogonal receiver components.

\begin{proposition}
\label{prop:mainprirel} Any achievable rate $R$ for the relay channel with orthogonal receiver components $p(y_1,y_{\mathrm{r}}|x)$ with a relay-to-receiver link of capacity $C_0$ must satisfy the following conditions
 \begin{align}
R&\leq I(X;Y_1,Y_{\mathrm{r}})-I( V;Y_1|Y_{\mathrm{r}})-I(X;Y_{\mathrm{r}}| V,Y_1)\label{propeq0.5o}
\\&=I(X;Y_1,V)-I(V;X|Y_\rmr),\label{propeq1o}
 \end{align}
 for some $p(x)p(y_1,y_{\mathrm{r}}|x)p(v|x,y_{\mathrm{r}})$ such that
\begin{align}
 I( V;Y_{\mathrm{r}})-I( V;Y_1)\leq C_0.\label{propconstraint}
 \end{align}
An equivalent form the bound without constraints is as follows: a rate $R$ is achievable if
 \begin{align}
R&\leq I(X;Y_1,Y_{\mathrm{r}})-I( V;Y_1|Y_{\mathrm{r}})-I(X;Y_{\mathrm{r}}| V,Y_1)=I(X;Y_1,V)-I(V;X|Y_\rmr)\label{Prop1neqn1}
\\
R&\leq I(X;Y_1)+C_0-I(V;Y_{\mathrm{r}}|X,Y_1)\label{Prop1neqn2}
 \end{align}
 for some $p(x)p(y_1,y_{\mathrm{r}}|x)p(v|x,y_{\mathrm{r}})$.
 
 Further it suffices to consider $|\Vc| \leq |\Xc| |\Yc_{\mathrm{r}}|+1$ when evaluating either form of the bound.

\end{proposition}
The proof of this proposition is given in Section \ref{sec:proofProp1}.  An alternative direct proof of the above proposition is given in Section \ref{sec:alternaproof}. 

\begin{remark}  \label{rmk8d}
The only difference between the upper bound in Proposition \ref{prop:mainprirel} and the equivalent forms of the compress-forward lower bound in 
\cite[Eq. 16.14]{elk11} and \cite[Proposition 3]{kim2007coding} is that the upper bound takes the maximum over $p(x)p(v|x,y_{\mathrm{r}})$ while the compress-forward lower bound takes maximum of the same expression over
$p(x)p(v|y_{\mathrm{r}})$.
\end{remark}

\smallskip

\noindent{\bf Kim's conjecture.} We use Proposition \ref{prop:mainprirel} to prove a conjecture posed by Kim in \cite[Question 2]{kim2007coding} for a class of deterministic relay channels with orthogonal receiver components defined by $p(y_1,y_\rmr|x)$, where $X=f(Y_1, Y_\rmr)$ for some function $f$. 
\begin{theorem}
Let $\mathcal{C}(C_0)$ be the supremum of achievable rates $R$ for a given $C_0$. Let $C_0^*$ be the minimum value of $C_0$ for which $\mathcal{C}(C_0) = \mathcal{C}(\infty) =\log|\mathcal{X}|$. Then 
$C_0^*=H_G(Y_\rmr|Y_1)$ and is achieved by a uniform distribution on $X$. Here $H_G(Y_\rmr|Y_1)$ denotes the conditional graph entropy of the characteristic
graph of $(Y_\rmr, Y_1)$ and the function $f$ (as defined in \cite{orlitsky1995coding}).
\label{thm:KimConjecture}
\end{theorem}
 \begin{proof}
 As argued in \cite{kim2007coding}, $C_0^*\leq H_G(Y_\rmr|Y_1)$. It remains to show that $C_0^*\geq H_G(Y_\rmr|Y_1)$.
 Let $C_0=C_0^*$. 
 Consider the following constraint in the upper bound of Proposition \ref{prop:mainprirel}:
\begin{align}
R&\leq I(X;Y_1,Y_{\mathrm{r}})-I( V;Y_1|Y_{\mathrm{r}})-I(X;Y_{\mathrm{r}}| V,Y_1).
\end{align}
If $\max_{p(x)}I(X;Y_1,Y_{\mathrm{r}})$ is achievable, it forces
\begin{align}
I(V;Y_1|Y_{\mathrm{r}})=I(X;Y_{\mathrm{r}}| V,Y_1)=0.
\label{eqnneqaaa}
\end{align}
 Since $X$ is a function of $(Y_1, Y_\rmr)$, $I(V;Y_1|Y_{\mathrm{r}})=I(V;X,Y_1|Y_{\mathrm{r}})\geq
 I(V;X|Y_\rmr)$. This yields $I(V;X|Y_\rmr)=0$. Consequently, \eqref{propeq1o} and \eqref{propconstraint} imply that the achievable rate is achievable by compress-forward (see \cite[Proposition 3]{kim2007coding} and Remark \ref{rmk8d}). From 
$\log|\mathcal{X}|=I(X; V,Y_1)$, we deduce that $X$ is uniformly distributed and $H(X|V,Y_1)=0$. From \cite[Theorem 2]{orlitsky1995coding}, we deduce that $C_0^*\geq H_G(Y_\rmr|Y_1)$. This confirms Kim's conjecture in \cite{kim2007coding}.
\end{proof}

\subsection{Product-form relay channels}
\label{sec:3}
Consider the following class of relay channels with orthogonal receiver components.
\begin{definition}
A relay channel with orthogonal receiver components is said to be \textit{product-form} if $p(y_1,y_{\mathrm{r}}|x) = p(y_1|x) p(y_{\mathrm{r}}|x)$.
\end{definition}
\begin{remark}
As Zhang observed in \cite{zhang1988partial}, $\max_{p(x)}I(X;Y_1)+C_0$  is an upper bound on the capacity of the product-form relay channel. Further, he identified a class of product-form relay channels for which the upper bound  is achievable. Proposition \ref{prop:mainprirel} can be used to recover and generalize his result. A similar argument as in Remark \ref{remark3} shows  that for any arbitrary product-form relay channel,
the rate $\max_{p(x)}I(X;Y_1)+C_0$ is achievable if and only if it is achievable by partial decode-forward.
To see this, note that  if the rate $\max_{p(x)}I(X;Y_1)+C_0$ is achievable, from \eqref{Prop1neqn2} we deduce that $I(V;Y_{\mathrm{r}}|X,Y_1)=0$. For a product-form relay channel we have $p(x)p(y_1|x)p(y_{\mathrm{r}}|x)p(v|x,y_{\mathrm{r}})$, which implies that $I(V;Y_1|X)=0$. Thus, $I(V;Y_\rmr,Y_1|X)=0$ or $p(v|x,y_r)=p(v|x)$. Thus, from \eqref{propeq1o} we obtain the following constraint on any achievable rate
\begin{align*}
R&\leq  I(X;Y_1,V)-I(V;X|Y_\rmr)
 \\&=I(V;Y_\rmr)-I(V;Y_\rmr|X)+I(X;Y_1|V)
 \\&= I(V;Y_\rmr)+I(X;Y_1|V).
 \end{align*}
On the other hand, the partial decode-forward rate~\cite[Eq. (5)]{kim2007coding} is
\[\max_{p(v,x)}\min\left(I(V;Y_\rmr)+I(X;Y_1|V), I(X;Y_1)+C_0\right).\]
This completes the proof.
\end{remark}

We will also consider the following special case of the above class.
\begin{definition}
\label{defn:spr}
A product-form relay channel is said to be \textit{symmetric} if $\mathcal{Y}_1=\mathcal{Y}_{\mathrm{r}}$, $p(y_{\mathrm{r}},y_1|x)=p(y_{\mathrm{r}}|x)p(y_1|x)$, and $p_{Y_1|X}(y|x)=p_{Y_{\mathrm{r}}|X}(y|x)$ for all $x,y$.
\end{definition}
In the following we consider several applications for these classes of relay channels.
\smallskip

\noindent{\bf Generalized Cover Relay Channel Problem.} We will need the following definitions.
\begin{definition}
A discrete-memoryless channel $p(y|x)$ is said to be \textit{generic} if the channel matrix $P$ with entries $P_{x,y}=p(y|x)$, $x \in \mathcal{X}$ and $y \in \mathcal{Y}$, is full row rank. 
\label{def:gencha}
\end{definition}

\begin{remark}
Observe that when $|\mathcal Y|\geq |\mathcal X|$, the class of row rank deficient matrices lie in a smaller dimensional space. Hence generic matrices are dense in the space of channel matrices when $|\mathcal Y| \geq |\mathcal X|$. When $|\mathcal Y| < |\mathcal X|$, let $\Yc=\{y_1,\ldots,y_k\}$. Append a new symbol $y_{k+1}$ to $\Yc$ and define the new channel matrix
\begin{align*}
    p_{new}(y_i|x) = \begin{cases} \begin{array}{cc} p(y_i|x) & 1\leq i \leq k-1\\
    \frac{1}{2} p(y_i|x) & i \in \{k,k+1\} \end{array} \end{cases}.
\end{align*}
Observe that the new receiver $Y_{new}$ can simulate $Y$ by collapsing the outputs $y_{k+1}$ and $y_k$ into a single symbol, and $Y$ can locally simulate $Y_{new}$. Hence the capacity for the new setting is the same as that in the previous setting. Therefore, from a capacity perspective, without loss of generality we can assume that $|\mathcal Y| \geq |\mathcal X|$. Thus the generic assumption that we make above allows us to deal with a dense class of matrices. Observe that if the expressions involved in determining the capacity satisfy a continuity property with respect to channel transition matrices, then one could deal with the capacity of a non-generic channel as the limit of the capacity of generic channels.
\end{remark}

\begin{remark}
It is immediate that if $p(y_1|x_1)$ and $p(y_2|x_2)$ are generic, then so is $p(y_1|x_1) \otimes p(y_2|x_2)$. That is, the class of generic channels is closed under a product operation.
\end{remark}

\begin{definition}
A product-form relay channel is said to be \textit{generic} if  the channel $p(y_1|x)$ is  generic.
\end{definition}

In \cite{cover1987capacity}, Cover posed a special (symmetric) case of the following problem. Consider a   product-form relay channel and let $\mathcal{C}(C_0)$ be the supremum of achievable rates $R$ for a given $C_0$. What is the critical value $C_0^*$ for which $C_0^*=\inf \{C_0:\mathcal{C}(C_0) = \mathcal{C}(\infty) = \max_{p(x)}I(X;Y_1,Y_{\mathrm{r}})\}$?

This problem has recently attracted a fair amount of attention and non-traditional methods have been used to answer the question as well as obtain new upper bounds for $\mathcal{C}(C_0)$ for symmetric Gaussian channels and binary-symmetric channels. As we will show in the next subsections our new upper bound, which uses traditional converse techniques, recovers and (in the binary-symmetric case) improves on these recent results. 

We now show that we can answer the generalized Cover relay channel problem when $p(y|x)$ is generic by evaluating the bound in Proposition~\ref{prop:mainprirel}.
\begin{theorem}
\label{thm:covopgen}
For a generic product-form relay channel, let $C_0^*$ be the minimum value of $C_0$ such that $\mathcal{C}(C_0) =   \mathcal{C}(\infty) = \max_{p(x)}I(X;Y_1,Y_{\mathrm{r}})$  and $R_{0}^*$ be the minimum value of $C_0$ such that $R_{\mathsf{CF}}(C_0) =   \mathcal{C}(\infty) = \max_{p(x)}I(X;Y_1,Y_{\mathrm{r}})$. Then $C_0^*=R_{0}^*$.
\end{theorem}
\begin{proof}
It is immediate  that $R_0^* \geq C_0^*$ since the compress-forward lower bound achieves $\mathcal{C}(\infty)$ when $C_0=R_0^*$. Therefore it suffices to show that $R_0^* \leq C_0^*$.
Let $C_0$ be such that $\mathcal{C}(C_0) = \mathcal{C}(\infty) = \max_{p(x)}I(X;Y_1,Y_{\mathrm{r}})$. 
From the constraint \eqref{propeq0.5o} of the upper bound in Proposition \ref{prop:mainprirel}, it follows that if $\mathcal{C}(\infty) = \max_{p(x)}I(X;Y_1,Y_{\mathrm{r}})$ is achievable, for a maximizing distribution $p^*(x)$, there exists a distribution  $p(v|x,y_{\mathrm{r}})$ such that $I(V;Y_1|Y_{\mathrm{r}})=0$. Since the channel $p_{Y_1|X}$ is generic, it follows from Lemma \ref{lem:genind} that $ I(V;Y_1|Y_{\mathrm{r}})=0$ implies that $I(V;X|Y_{\mathrm{r}})=0$. Therefore $V \to Y_{\mathrm{r}} \to X \to Y_1$ form a Markov chain. On the other hand, the compress-forward rate in~\cite[Proposition 3]{kim2007coding} is
\[\max_{p(x)p(v|y_\rmr)}\{(I(X;Y_1,V):I(V;Y_\rmr|Y_1)\leq C_0\}.\]
Thus, the constraints in \eqref{propeq1o} and \eqref{propconstraint} 
imply that the rate $R$ is achievable by compress-forward with the compression random variable $V$. Hence, we have that $R_{\mathsf{CF}}(C_0) = \mathcal C(C_0) = \mathcal C(\infty)$. Since this holds for any $C_0$ such that $\mathcal{C}(C_0) = \mathcal{C}(\infty) = \max_{p(x)}I(X;Y_1,Y_{\mathrm{r}})$, we have that $R_0^* \leq C_0^*$. This completes the proof.\end{proof}

\begin{remark}
In a concurrent work \cite{liu2020minorization} uses results and techniques from convex geometry to arrive at a solution for Cover's problem. In particular, Theorem 1 of \cite{liu2020minorization} states that there is a constant $c>0$ such that for any $R_0\in[H(\bar Z|\bar Y)-c^{-1}, H(\bar Z|\bar Y)]$, we have
\[H(\bar Z|\bar Y)-R_0\leq c\lambda^{\frac{1}{10}}\log^{\frac65}\frac{1}{\lambda},
\]
where 
$\lambda=\mathcal{C}(\infty)-\mathcal{C}(R_0)$, and $\bar Y$ and $\bar Z$ are random variables defined in \cite{liu2020minorization}.
In contrast, our proof follows from the upper bound established in Proposition~\ref{prop:mainprirel}. 
\end{remark}

\smallskip

\noindent{\bf  Gaussian product-form relay channel}. The Gaussian product-form relay channel is defined as
\begin{align}
Y_1&=X+W_1,\nonumber \\
Y_{\rmr}&=X+W_\rmr, \nonumber
\end{align}
where $W_1\sim \Nc(0,N_1)$ and $W_\rmr \sim \Nc(0,N_\rmr)$ are independent of each other and of $X$, and a link of capacity $C_0$ from the relay to the destination. We assume average power constraint $P$ on $X$ and define $S_{12}=P/N_\rmr$, $S_{13}=P/N_1$ and $S_{23}=2^{2C_0}-1$. 

The upper bound in Proposition \ref{prop:mainprirel} reduces to the following.
\begin{proposition}
\label{Proposition5}
Any achievable rate $R$ for the Gaussian product-form relay channel must satisfy the condition
\begin{align}\label{eqnPrimitiveGaussian}
        R\leq
        \begin{cases}
        \frac12\log\left(1+S_{13}+\frac{S_{12}(S_{13}+1)S_{23}}{(S_{13}+1)(S_{23}+1)-1}\right), & \textit{for } S_{12}\leq S_{13}+S_{23}+S_{13}S_{23},\\
        \frac{1}{2}\log((1+S_{13})(1+S_{23})), & \textit{otherwise.} 
        \end{cases}
\end{align}
\end{proposition}

The proof of this proposition is given in Section~\ref{proofthm:Prop5}. We use ideas from \cite{gen14} to show the optimality of jointly Gaussian $(V,X,Y_1,Y_\rmr)$ in computing the upper bound of Proposition \ref{prop:mainprirel}, leading to an optimization problem with only three free variables, which we then reduce to a one-dimensional optimization problem by observing two properties of any maximizer: (i) equation \eqref{propconstraint} holds with equality for the maximizing distribution and (ii) 
$K_{X,Y_\rmr}-K_{X,Y_\rmr|{V}}$ is a rank one matrix, where $K_{X,Y_\rmr}$ is the covariance matrix of $(X,Y_\rmr)$ and $K_{X,Y_\rmr|{V}}$ is the conditional covariance matrix of $(X,Y_\rmr)$ given $V$.

\begin{remark}
The decode-forward lower bound for the Gaussian product-form relay channel in \cite[Eq. 16.16]{elk11} is
\begin{align}\label{eqnDF}
        C\geq
        \begin{cases}
        \frac12\log\left(1+S_{12}\right), & \text{for } S_{12}\leq S_{13}+S_{23}+S_{13}S_{23},\\
        \frac{1}{2}\log((1+S_{13})(1+S_{23})), & \text{otherwise}. 
        \end{cases}
\end{align}
When $S_{12}\geq S_{13}+S_{23}+S_{13}S_{23}$, the upper bound in \eqref{eqnPrimitiveGaussian}, the cutset bound and the decode-forward lower bound all coincide.  The condition $S_{12}\leq S_{13}+S_{23}+S_{13}S_{23}$ is equivalent to $I(X; Y_\rmr)\leq I(X; Y_1) + C_0$ for a Gaussian input $X\sim \mathcal{N}(0,P)$. 

The compress-forward lower bound for this relay channel in \cite[Eq. 16.17]{elk11} implies that
\begin{align}
    C\geq \frac12\log\left(1+S_{13}+\frac{S_{12}(S_{13}+1)S_{23}}{S_{12}+(S_{13}+1)(S_{23}+1)}\right).
\end{align}
Furthermore, this lower bound can be improved by time-sharing at the transmitter \cite[Sec. 16.8]{elk11} or at the relay \cite[Footnote 2]{Wu19}. 
\end{remark}

Figure \ref{figGPrimitive} plots the upper bound in Proposition \ref{Proposition5} along with the cutset upper bound and the compress-forward lower bound for an example Gaussian product-form relay channel.

\begin{figure}
    \centering
    \includegraphics[scale=0.46]{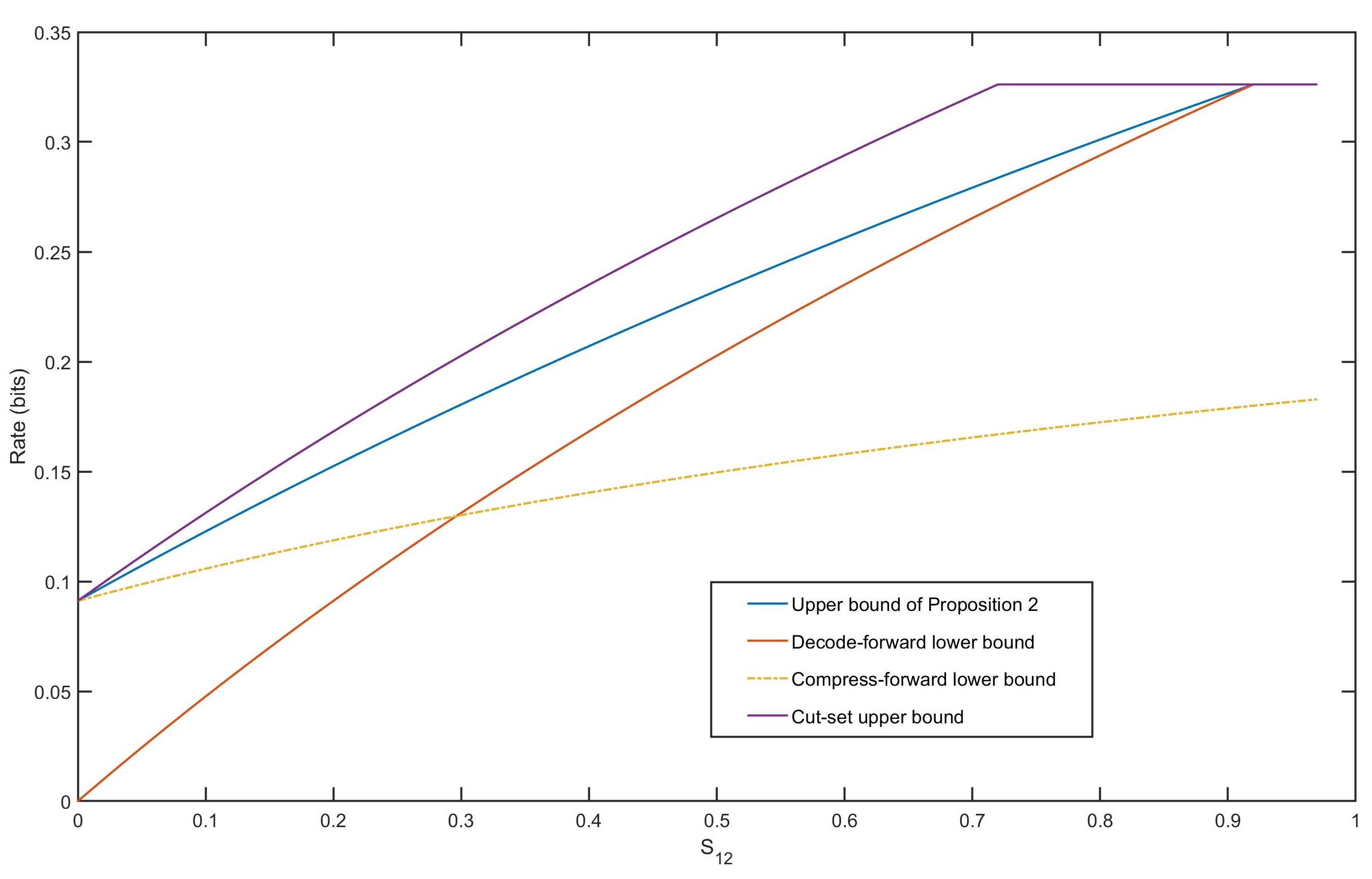}
    \caption{Plots of the bounds for the Gaussian product-form relay channel with $S_{13}=0.2$, $S_{23}=0.6$. }
    \label{figGPrimitive}
\end{figure}

The following upper bound on the capacity of the Gaussian product-form relay channel was established in~\cite{wu2017geometry}.
\begin{theorem}[\cite{wu2017geometry}]\label{thm:Wu19thm}
For the Gaussian product-form relay channel, any achievable rate $R$ must satisfy the condition
\begin{align} 
 R \leq \ &\frac{1}{2}\log \left(1+S_{13}\right)+\sup_{\theta\in\left[\arcsin (\frac{1}{1+S_{23}}),\frac{\pi}{2}\right]}\min \Bigg\{ C_0+\log \sin \theta, 
 \min_{\omega\in \left(\frac{\pi}{2}-\theta, \frac{\pi}{2}\right]}  h_{\theta}(\omega)
\Bigg\},\label{eqn:htheta}
\end{align}
where
\begin{align}&~h_{\theta}(\omega) =    \frac{1}{2}\log \left(\frac{
\left(S_{12}+S_{13}+\sin ^2(\omega)-2\cos(\omega)\sqrt{S_{12}S_{13}}
\right)\sin^2\theta}{(S_{13}+1)(\sin ^2 \theta - \cos^2 \omega)} \right).\label{eqn:htheta2}
 \end{align}
\end{theorem}
Although the techniques used to prove this theorem are completely different from those used in this paper, it turns out quite surprisingly that bound \eqref{eqn:htheta} coincides with the bound in Proposition~\ref{Proposition5}.

\begin{lemma}\label{thm:GSCase}
The bound in Theorem \ref{thm:Wu19thm} simplifies to
\begin{align}
        R\leq
        \begin{cases}
        \frac12\log\left(1+S_{13}+\frac{S_{12}(S_{13}+1)S_{23}}{(S_{13}+1)(S_{23}+1)-1}\right), & \textit{for } S_{12}\leq S_{13}+S_{23}+S_{13}S_{23},\\
        \frac{1}{2}\log((1+S_{13})(1+S_{23})), & \text{otherwise}, 
        \end{cases}
\end{align}
which coincides with the bound in Proposition~\ref{Proposition5}.
\end{lemma}
The proof of this lemma is given in Section \ref{proofthm:GSCase}. 
The proof works by explicitly computing the optimal values of $\theta$ and $\omega$ in \eqref{eqn:htheta}.

\smallskip

\noindent{\bf Symmetric binary relay channel with orthogonal receiver components}. Consider a symmetric product-form relay channel described by $p(y_1,y_{\rmr}|x)$ and a link of rate $C_0$ from relay to the destination such that
$p(y_1,y_{\rmr}|x)=p(y_1|x)p(y_{\rmr}|x)$, where $x, y_1, y_{\rmr}\in\{0,1\}$. Further, assume that  
both the channels  $p(y_1|x)$ and $p(y_{\rmr}|x)$ are binary symmetric channels with crossover probability $\rho\in[0,1/2]$. By specializing Proposition \ref{prop:mainprirel} to this channel, we obtain the following bound.

\begin{theorem}\label{thmBSCPrimitive} Given $\lambda\in [0,1]$ and $c\in[0,1]$, let $g_{\lambda}(c)$ be the maximum of $(1-\la) \left(H(Y_1) - H(Y_{\rmr})\right) + H(Y_{\rmr}|X)$
over all joint probability distributions $p(x,y_{\rmr})$ on $\{0,1\}\times\{0,1\}$ satisfying $p_{X,Y_\rmr}(0, 1) + p_{X,Y_\rmr}(1, 0)=c$. 
For any fixed $\la\in[0,1]$, let $\mathscr{C}[g_{\la}]:[0,1]\mapsto \mathbb{R}$ be the upper concave envelope of the function $g_{\la}(\cdot)$, i.e., the smallest concave function that dominates $g_{\la}(\cdot)$ from above. 
Any achievable rate $R$ for a symmetric binary relay channel with orthogonal receiver components with parameter $\rho$ must satisfy the condition
\[
R\leq 1-2H_2(\rho) +\lambda C_0+\mathscr{C}[g_\la](\rho)
\]
 for any $\la\in [0,1]$, where $H_2(x)=-x\log(x)-(1-x)\log(1-x)$ is the binary entropy function. 
 \end{theorem}

The proof of this theorem is given in Section \ref{proofthmBSCPrimitive}.

Figure \ref{fig1b} shows that our new upper bound strictly improves upon the bounds given in \cite{wu2017improving} and \cite{barnes2017solution}.

\begin{figure}[htpb]
    \centering
    \includegraphics[scale=0.15]{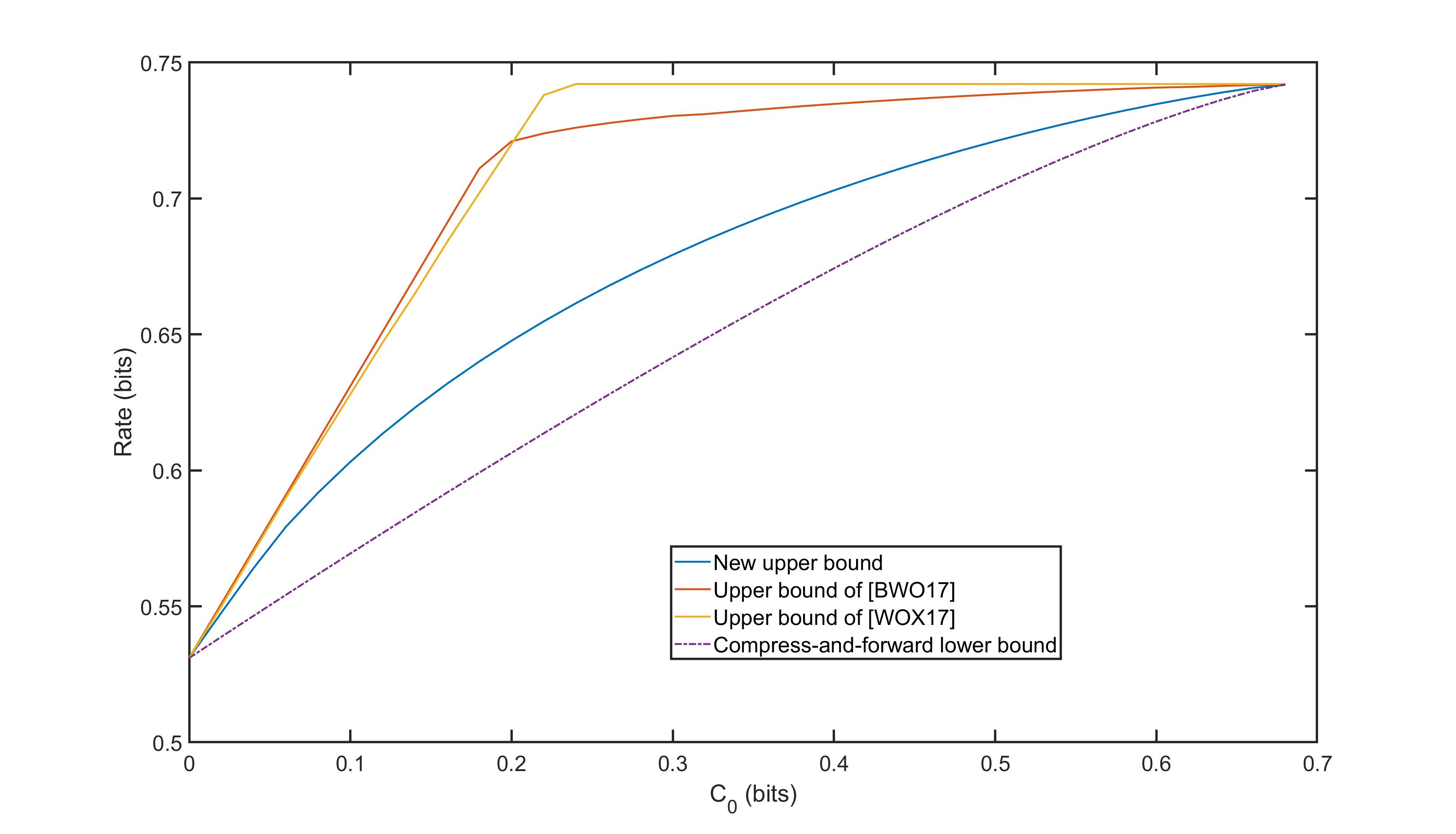}
    \caption{Plots of the minimum of the two upper bounds give in \cite{wu2017improving}, the upper bound given in \cite{barnes2017solution}, our new bound and the compress-forward lower bound for a symmetric binary relay channel with orthogonal receiver components with parameter $\rho=0.1$. }
    \label{fig1b}
\end{figure}

\smallskip

\subsection{Upper bound with auxiliary receiver and applications}
\label{sec:aux-receiver}

We establish an upper bound on the capacity of the relay channel with orthogonal receiver components which uses an auxiliary receiver~\cite{gon22}. We use this bound to to obtain a tighter upper bound than the one given in \cite{tandon2008new} and to show the suboptimality of our main bound in Proposition \ref{prop:mainprirel}.

As in \cite{gon22}, the auxiliary receiver $J$ is an enhancement of the relay's output variable $Y_\rmr$. However, the identification of the auxiliary variables in the following upper bound are different from those in \cite{gon22}.

\begin{theorem}\label{prop2jm2s} Consider a relay channel with orthogonal receiver components. Assume that for an auxiliary receiver $J$ we have $p(j,y_\rmr,y_1|x)=p(j|x)p(y_\rmr|j)p(y_1|y_\rmr,j,x)$.
Then any achievable rate $R$ must satisfy the condition 
\begin{align*}
R&\leq I(X;J)+I(X;Y_1|V,W,J,T)-I(X;Y_\rmr|V,W,J,T),
\end{align*}
 for some $p(t,x)p(j|x)p(y_\rmr|j)p(y_1|j,y_\rmr,x)p(w|t,y_{\mathrm{r}})p(v|t,x,y_{\mathrm{r}},w)$ such that
\begin{align*}
  I(V,W;Y_\rmr|J,T)-I(V,W;Y_1|J,T)&\leq I(W;Y_{\mathrm{r}}|J,T) \leq  C_0.
\end{align*}
Further it suffices to consider 
$|\Tc|\leq |\Xc|+3$, 
$|\Wc| \leq |\Yc_{\mathrm{r}}|+3$ and
$|\Vc| \leq |\Xc| |\Yc_{\mathrm{r}}|+1$.
\end{theorem}
The proof of this theorem is given in Section \ref{proofprop2jm2s}. As an outline, let  $M_\rmr\in[1:2^{2C_0}]$ be  a function of $Y_\rmr^n$ that can be viewed as the ``message" from the relay to the receiver. Proposition \ref{prop:mainprirel} is established via the  identification   $V_i=(Y_{1i+1}^{n}, Y_{\rmr }^{i-1}, M_\rmr)$ in Section \ref{sec:alternaproof}. In Theorem \ref{prop2jm2s}, we instead use the more refined identification
\begin{align*}
V_i&=Y_{1i+1}^{n},\quad
    W_i = (M_\rmr,Y_{\mathrm{r}}^{i-1}), \quad T_i=(J^{i-1}, J_{i+1}^n),
\end{align*}
Crucially, the introduction of the auxiliary receiver $J$ and the auxiliary random variable $T$ allow us to induce certain Markov chains and constraints on $V$ and $W$. Note that the bound in Theorem \ref{thm:maintheorem} applies to arbitrary relay channels with no self-interference while the bound in Theorem \ref{prop2jm2s}  applies only to relay channels with orthogonal receiver components, but is tighter for this class.

\smallskip

\noindent{\bf Relay channels with i.i.d. output sequence}.
Consider the following class of relay channels.

\begin{definition}
A relay channel with orthogonal receiver components is said to be with i.i.d. output if its family of conditional probabilities have the form $p(y,y_{\mathrm{r}}|x,x_{\mathrm{r}})=p(y_{\mathrm{r}})p(y_1|x,y_{\mathrm{r}})p(y_2|x_{\mathrm{r}})$.
\end{definition}
\begin{remark}
In~\cite{aleksic2009capacity}, a special case of this relay channel in which the channel from the transmitter to the receiver is binary symmetric and the relay observes a corrupted version of the transmitter-receiver channel noise was used to establish the suboptimality of the cutset bound.
\end{remark}
\begin{remark}
Communication over this relay channel is equivalent to communication over channels with rate-limited state information available at the receiver~\cite{aleksic2009capacity}. In \cite{ahlswede1983source},  Ahlswede and Han conjectured that the capacity should be equal to $\max I(X;Y_1|W)$,
where the maximum is over  $p(x)$ and $p(w|y_\rmr)$ such that
$I(W;Y_{\mathrm{r}}|Y_1)\leq C_0$. Note that the joint distribution here is of the form $p(x)p(y_{\mathrm{r}})p(y_1|x,y_{\mathrm{r}})p(w|y_{\mathrm{r}})$.
\end{remark}

For the relay channel with orthogonal receiver components and i.i.d. output, we obtain the following corollary to Theorem~\ref{prop2jm2s}.
\begin{corollary}
\label{prop2} Any achievable rate $R$ for the relay channel with orthogonal receiver components and i.i.d. output must satisfy the condition 
\begin{align*}
R&\leq I(X;Y_1|T,W,V)-I(X;Y_{\mathrm{r}}|T,W,V),
\end{align*}
 for some $p(t,x)p(y_1|x,y_{\mathrm{r}})p(y_{\mathrm{r}})p(w|t,y_{\mathrm{r}})p(v|t,x,y_{\mathrm{r}},w)$ such that
 \begin{align*}
  I(V,W;Y_\rmr|T)-I(V,W;Y_1|T)&\leq I(W;Y_{\mathrm{r}}|T) \leq  C_0.
\end{align*}
\end{corollary}
The proof of this corollary follows by setting the auxiliary receiver $J$ to be a constant in the statement of Theorem~\ref{prop2jm2s}. This choice is feasible since $Y_\rmr$ is independent of $X$.

In~\cite{tandon2008new}, the following upper bound on the capacity of the above relay channel is given. 
\begin{theorem}[\cite{tandon2008new}]\label{thm:TU08}
Any achievable rate $R$ for the relay with orthogonal receiver components channel and i.i.d. output must satisfy the condition 
\begin{align*}
R &\leq \min\{I(W,X;Y_1|T), I(X;Y_1|Y_{\mathrm{r}},T)\}
 \end{align*}
 for some $p(x,t)p(y_1|x,y_{\mathrm{r}})p(y_{\mathrm{r}})p(w|y_{\mathrm{r}},t)$ such that
$I(W;Y_{\mathrm{r}}|T)\leq C_0$.
\end{theorem}
\begin{remark}
The time-sharing random variable $T$ was not included in the  statement of the above theorem in \cite{tandon2008new}. We believe it is necessary, however, since the Markov chain $(W_T,T)\rightarrow Y_{\rmr, T}\rightarrow X_T$ does not hold with the identification of auxiliary random variables in \cite{tandon2008new}. The upper bound in Theorem \ref{thm:TU08} with  $T$ still provides a converse  for the example considered in~\cite{aleksic2009capacity}.
\end{remark}

Next, we show that the bound in Corollary \ref{prop2} strictly improves over the bound in Theorem \ref{thm:TU08}.

\begin{proposition}\label{lemmaTU} The bound in Corollary \ref{prop2} is tighter than the bound in Theorem \ref{thm:TU08}. Moreover, the tightness is strict for the following class of relay channels.
Consider a relay channel with i.i.d. output sequence $p(y_{\rmr})p(y_1|x,y_{\rmr})$ such that the channel $p(y_1|x, Y_{\rmr}=y_{\rmr})$ is generic for every $y_{\rmr}$, and for which there is no random variable $K$ such that $H(K|X,Y_1)=H(K|X,Y_\rmr)=0$ but $H(K|X)>0$. For such a channel, if the compress-forward lower bound  with time-sharing does not match the bound in Theorem \ref{thm:TU08}, then the bound in Corollary \ref{prop2} strictly improves upon the bound in Theorem \ref{thm:TU08}.
\end{proposition}
The proof of this proposition is given in Section \ref{prooflemmaTU}.

We can show that the class for which the improvement is strict is non-empty.
In \cite[Section VII]{tandon2008new} it is shown that for a  channel of the form $Y_1=XY_\rmr + N$, where $N$ and $Y_\rmr$ are Bernoulli random variables independent of the binary $X$, the compress-forward lower bound (without time-sharing) is significantly smaller than the upper bound in Theorem \ref{thm:TU08}. Further from Figure 4 in  \cite{tandon2008new}, it is evident that even with time-sharing the compress-forward lower bound continues to be below the bound in Theorem \ref{thm:TU08}. Now consider the perturbed version of the channel 
$Y_1=(XY_\rmr + N)Z + (X + N)(1-Z)$, where $Z$ is independent of all other random variables and $\P(Z=1)=1-\eps$. 
 We argue below that the perturbed channel satisfies the assumptions in Proposition \ref{lemmaTU}. Therefore, by the continuity of the various bounds in the channel parameters, the gap between the compress-forward rate (with time-sharing) and the upper bound continues to be non-zero for small enough values of $\eps$.

To see why the perturbed channel satisfies the assumptions in Proposition \ref{lemmaTU}, assume that there is a random variable $K$ with the properties in
Proposition \ref{lemmaTU} for this example. Since $H(K|X)>0$, there is some $X=x$ such that $H(K|X=x)>0$. For  $X=x$, $K$ is the Gacs-Korner common part of $Y_r$ and $(xY_\rmr + N)Z + (x + N)(1-Z)$. Since $Y_\rmr$ is binary, the only non-trivial Gacs-Korner common part $K$ is $K=Y_\rmr$. For this $K$ to work, we  require that $H\left(Y_\rmr|(xY_\rmr + N)Z + (x + N)(1-Z)\right)=0$. However, because of the noises in $Z$ and $N$, $Y_\rmr$ is not a function of
$(xY_\rmr + N)Z + (x + N)(1-Z)$. Therefore, this Gacs-Korner common part is empty.

\smallskip

\noindent{\bf Auxiliary receiver can help}.
Consider the Gaussian relay channel with orthogonal receiver components and i.i.d.\ output for which
\begin{align}
Y_{\rmr}&=Z_\rmr,\label{GaussianExampleEq1}\\
Y_1&=X+Y_\rmr+Z_1,\label{GaussianExampleEq2}
\end{align}
where $Z_1\sim \Nc(0,N_1)$ and $Z_\rmr \sim \Nc(0,N_\rmr)$ are independent of each other and of $X$. Assume average power constraint $P$ on $X$. In the following, we show that for this class of channels the bound with auxiliary receiver is strictly tighter than our main upper bound with no auxiliary receivers in Proposition \ref{prop:mainprirel}. 

\begin{theorem}\label{comparisonBD} The bound in Corollary \ref{prop2} strictly improves upon the bound in Proposition \ref{prop:mainprirel} for the Gaussian relay channel in \eqref{GaussianExampleEq1}-\eqref{GaussianExampleEq2} for any $N_1, N_\rmr, C_0>0$.
\end{theorem}

The proof of this theorem is given in Section \ref{proofcomparisonBD}.


\section{Proofs of the results}
\label{sec:proofs}
In the following sections we present the proofs of the results stated in the previous section in their order of appearance. 

\subsection{Proof of Theorem \ref{thm:maintheorem}}
\label{sec:proofThm1}

We will first prove the theorem with the original constraints and then show the equivalence of the region to that obtained by dropping the constraint on $R$ in \eqref{eqnUB0} and instead strengthening \eqref{eq:UBcon} to \eqref{eq:UBcon2}.

Consider a sequence of codes achieving a rate $R$. Each code in this sequence induces a joint distribution on $(M,X^n,X_\rmr^n,Y^n,Y_\rmr^n)$. The proof essentially follows from routine manipulations along with the following identifications:
\begin{align*}
V_i&=(Y_{i+1}^n, Y_{\mathrm{r}}^{i-1}),\quad
    U_i =  (Y_{\mathrm{r}i+1}^n, Y^{i-1}),
\end{align*}
and
$V=(Q,V_Q)$, $U=(Q,U_Q)$, $X=X_Q, X_{\mathrm{r}}=X_{\mathrm{r}Q}, Y=Y_Q, Y_{\mathrm{r}}=Y_{\mathrm{r}Q}$, where $Q\stackrel{(d)}{=} \textrm{Uniform}[1:n]$ is a time-sharing random variable. We employ the data-processing inequality, chain-rule, and the Csisz\'ar-K\"orner-Marton identity repeatedly. It may be useful for the reader to have the  Bayesian Network diagram in \Cref{fig:bayesrel} and the d-separation theorem for deducing the various Markov chains employed.

 The Markov chain
    $V\rightarrow (X,X_{\mathrm{r}},Y_{\mathrm{r}})\rightarrow Y$
    follows from the fact that the channel $p(y|x,x_{\mathrm{r}},y_{\mathrm{r}})$ is memoryless. 
    The  Markov chain $U\rightarrow (X,X_{\mathrm{r}})\rightarrow (Y,Y_{\mathrm{r}})$  for
    relay channels of the form $p(y_\mathrm{r} |x)p(y|x, x_\mathrm{r},y_\mathrm{r})$ 
    can be deduced from  Figure~\ref{fig:bayesrel}.

    Since the expressions in the statement of the theorem depend only on the marginal distributions of 
    $p(u,x,x_{\mathrm{r}})$ and $p(v|x,x_{\mathrm{r}},y_{\mathrm{r}})$, the union in the statement of the theorem is taken over $p(u,x,x_{\mathrm{r}})p(y,y_{\mathrm{r}}|x,x_{\mathrm{r}})p(v|x,x_{\mathrm{r}},y_{\mathrm{r}})$.

To show  \eqref{eqnUB1}, we write
    \begin{align*}
       I(M;Y^n)=I(M;Y_{\mathrm{r}}^n)+I(M;Y^n)-I(M;Y_{\mathrm{r}}^n).
    \end{align*}
Following the steps in the proof of the cutset bound, we obtain
\[
 I(M;Y_{\mathrm{r}}^n)\leq  \sum_{i=1}^n I(X_i;Y_{\mathrm{r}i}|X_{\mathrm{r}i}).
 \]
On the other hand, we have{\allowdisplaybreaks
\begin{align*}I(M;Y^n)- I(M;Y_{\mathrm{r}}^n)
&\leq I(M;Y^n|Y_{\mathrm{r}}^n)
\\    &=\sum_{i=1}^n I(M;Y_i|Y^{i-1},Y_{\mathrm{r}}^n) \\
    & = \sum_{i=1}^n I(M,X_i;Y_i|X_{\mathrm{r}i},Y^{i-1},Y_{\mathrm{r}}^n)  \\
    & \stackrel{(a)}{=}  \sum_{i=1}^n I(X_i;Y_i|X_{\mathrm{r}i},Y^{i-1},Y_{\mathrm{r}}^n) 
    \\&\stackrel{(b)}{\leq}
    \sum_{i=1}^n I(X_i;Y_i|X_{\mathrm{r}i},Y^{i-1},Y_{\mathrm{r}i}^n) 
    \\& = \sum_{i=1}^n I(X_i;Y_i|X_{\mathrm{r}i},Y_{\mathrm{r}i},U_i),
    \end{align*}}
    where in $(a)$  and $(b)$ we use the Markov chain 
    $(Y_{\mathrm{r}}^{i-1},Y_{\mathrm{r}i+1}^n,M,Y^{i-1})\rightarrow (X_i,X_{\mathrm{r}i},Y_{\mathrm{r}i})\rightarrow Y_i$. This Markov chain follows from the fact that the channel $p(y|x,x_{\mathrm{r}},y_{\mathrm{r}})$ is memoryless.
    {\allowdisplaybreaks
Thus, we obtain
       \begin{align*}
      I(M;Y^n)&=I(M;Y_{\mathrm{r}}^n)+I(M;Y^n)-I(M;Y_{\mathrm{r}}^n)
       \\&\leq 
       \sum_{i=1}^n I(X_i;Y_{\mathrm{r}i}|X_{\mathrm{r}i})+ I(X_i;Y_i|X_{\mathrm{r}i},Y_{\mathrm{r}i},U_i)
       \\&\leq n\left( 
       I(X;Y_{\mathrm{r}}|X_{\mathrm{r}})+ I(X;Y|X_{\mathrm{r}},Y_{\mathrm{r}},U)\right).
    \end{align*}}
    
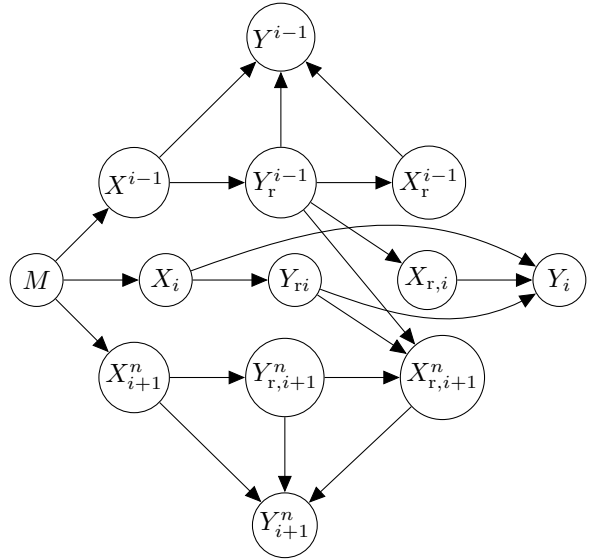
\begin{figure}[b]
\begin{center}
    \begin{tikzpicture}[->]
    \node[latent] (m) {$M$};
    \node[latent, right=of m] (xpr) {$X_i$};
    \node[latent, above right=of m] (xpa) {$X^{i-1}$};
    \node[latent, below right=of m] (xf) {$X_{i+1}^n$};
    \node[latent, right=of xpr] (yrpr) {$Y_{\rmr i}$};
    \node[latent, right=of xpa] (yrpa) {$Y_\rmr^{i-1}$};
    \node[latent, right=of xf] (yrf) {$Y_{\rmr,i+1}^n$};
    \node[latent, right=of yrpa] (xrpa) {$X_\rmr^{i-1}$};
    \node[latent, right=of yrpr] (xrpr) {$X_{\rmr,i}$};
    \node[latent, right=of yrf] (xrf) {$X_{\rmr,i+1}^n$};
    \node[latent, above=of yrpa] (ypa) {$Y^{i-1}$};
    \node[latent, right=of xrpr] (ypr) {$Y_{i}$};
    \node[latent, below=of yrf] (yf) {$Y_{i+1}^n$};
    \edge {m} {xpa,xpr,xf};
    \edge {xpa} {yrpa};
    \edge {xpr} {yrpr};
    \edge {xf} {yrf};
    \edge {yrpa} {xrpa,xrpr};
    \edge {yrpa,yrpr,yrf} {xrf};
    \edge {yrpa,xrpa,xpa} {ypa};
    \edge {xrpr} {ypr};
    \path (xpr) edge[out=20, in=150]  (ypr) ;%
    \path (yrpr) edge[out=-20, in=210]  (ypr) ;
    \edge {yrf,xrf,xf} {yf};
    \end{tikzpicture}
    \end{center}
    \caption{A Bayesian network representing the causal relationships in the joint distribution of the random variables induced by the codebooks for the relay channel without self-interference.}
    \label{fig:bayesrel}
\end{figure}

Next, we establish \eqref{eqnUB2}. Observe that by definition, $V_i$ implies that $X_{\mathrm{r}i}$. We have
\begin{align}\nonumber
\frac1n I(M;Y^n)&\leq \frac1n \left(I(M;Y_{\mathrm{r}}^n)+I(M;Y^n)-I(M;Y_{\mathrm{r}}^n)\right)
\\&\nonumber\leq I(X;Y_{\mathrm{r}}|X_{\mathrm{r}})+\frac1n\left(\sum_{i}I(M;Y_i|Y_{\mathrm{r}}^{i-1},Y_{i+1}^n)-\sum_{i}I(M;Y_{\mathrm{r}i}|Y_{\mathrm{r}}^{i-1},Y_{i+1}^n)\right)
\\\nonumber&= I(X;Y_{\mathrm{r}}|X_{\mathrm{r}})+\frac1n\left(\sum_{i}I(M,X_i;Y_i|Y_{\mathrm{r}}^{i-1},X_{\mathrm{r}i},Y_{i+1}^n)-\sum_{i}I(M,X_i;Y_{\mathrm{r}i}|Y_{\mathrm{r}}^{i-1},X_{\mathrm{r}i},Y_{i+1}^n)\right)
\\\nonumber&= I(X;Y_{\mathrm{r}}|X_{\mathrm{r}})+\frac1n\left(\sum_{i}I(X_i;Y_i|Y_{\mathrm{r}}^{i-1},X_{\mathrm{r}i},Y_{i+1}^n)+
\sum_{i}I(M;Y_i|X_i,Y_{\mathrm{r}}^{i-1},X_{\mathrm{r}i},Y_{i+1}^n)\right.
\\&\nonumber\qquad\qquad\left.
-\sum_{i}I(X_i;Y_{\mathrm{r}i}|Y_{\mathrm{r}}^{i-1},X_{\mathrm{r}i},Y_{i+1}^n)
-\sum_{i}I(M;Y_{\mathrm{r}i}|X_i,Y_{\mathrm{r}}^{i-1},X_{\mathrm{r}i},Y_{i+1}^n)\right)
\\&\nonumber\leq I(X;Y_{\mathrm{r}}|X_{\mathrm{r}})+\frac1n\left(\sum_{i}I(X_i;Y_i|Y_{\mathrm{r}}^{i-1},X_{\mathrm{r}i},Y_{i+1}^n)-\sum_{i}I(X_i;Y_{\mathrm{r}i}|Y_{\mathrm{r}}^{i-1},X_{\mathrm{r}i},Y_{i+1}^n)\right.
\\&\nonumber\qquad\qquad\left.
+\sum_{i}I(M;Y_i|X_i,Y_{\mathrm{r}}^{i-1},X_{\mathrm{r}i},Y_{i+1}^n,Y_{\mathrm{r}i})\right)\\ \label{eqnEEa}
 &= I(X;Y_\rmr|X_\rmr) + I(X;Y|V,X_\rmr) - I(X;Y_\rmr|V,X_\rmr),
\end{align}
where in equation \eqref{eqnEEa}, we used the fact that
\[
I(M;Y_{i}|X_{i}, X_{\mathrm{r}i}, Y_{i+1}^n, Y_{\mathrm{r}}^{i-1},Y_{\mathrm{r}i})=0.
\] 

To show \eqref{eqnUB0}, note that 
\begin{align}
      \frac1n I(M;Y^n)&\leq \frac1n I(X^n;Y^n)\nonumber
\\&\leq \frac1n\sum_i I(X^{n},Y_{i+1}^n;Y_{i})\nonumber
\\&= \frac1n\sum_i I(X_i;Y_i)+
\frac1n\sum_i I(X^{n\setminus i},Y_{i+1}^n;Y_{i}|X_i)\nonumber
\\&\leq \frac1n\left(\sum_i I(X_i;Y_i)+
I(V_i,X_{\mathrm{r}i};Y_i|X_i)-I(V_i;Y_{\mathrm{r}i}|X_i,X_{\mathrm{r}i})\right)\label{step2}
\\&\leq I(X;Y)+
I(V,X_{\mathrm{r}};Y|X)-I(V;Y_{\mathrm{r}}|X,X_{\mathrm{r}})\nonumber
\\&= I(X,X_{\mathrm{r}};Y)+
I(V;Y|X,X_{\mathrm{r}})-I(V;Y_{\mathrm{r}}|X,X_{\mathrm{r}})\nonumber
\\&=I(X,X_{\mathrm{r}};Y)-I(V;Y_{\mathrm{r}}|X_{\mathrm{r}},X,Y),\label{step4}
    \end{align}{\allowdisplaybreaks 
where \eqref{step4} follows from the
Markov chain relationship
    $V\rightarrow (X,X_{\mathrm{r}},Y_{\mathrm{r}})\rightarrow Y$ and \eqref{step2} follows from
    
\begin{align}
    \sum_i I(Y_{\mathrm{r}}^{i-1},Y_{i+1}^n;Y_{\mathrm{r}i}|X_i,X_{\mathrm{r}i}) &\stackrel{(a)}{=} \sum_i I(X^{n\setminus i},Y_{\mathrm{r}}^{i-1},Y_{i+1}^n;Y_{\mathrm{r}i}|X_i) - \sum_i I(X^{n\setminus i};Y_{\mathrm{r}i}|X_i,Y_{\mathrm{r}}^{i-1},Y_{i+1}^n)
    \nonumber \\&\qquad - \sum_i I(X_{\mathrm{r}i};Y_{\mathrm{r}i}|X_i)\\
    &\stackrel{(b)}{=} \sum_iI(Y_{i+1}^n;Y_{\mathrm{r}i}|X^n,Y_{\mathrm{r}}^{i-1}) - \sum_i I(X^{n\setminus i};Y_{\mathrm{r}i}|X_i,Y_{\mathrm{r}}^{i-1},Y_{i+1}^n) \nonumber\\
    & \stackrel{(c)}{=} \sum_iI(Y_{\mathrm{r}}^{i-1};Y_{i}|X^n,Y_{i+1}^n) - \sum_i I(X^{n\setminus i};Y_{\mathrm{r}i}|X_i,Y_{\mathrm{r}}^{i-1},Y_{i+1}^n)\nonumber\\
    & \stackrel{(d)}{=} \sum_iI(Y_{\mathrm{r}}^{i-1};Y_{i}|X^n,Y_{i+1}^n) - \sum_i I(X^{n\setminus i};Y_{\mathrm{r}i},Y_{i}|X_i,Y_{\mathrm{r}}^{i-1},Y_{i+1}^n) \nonumber\\
    & = \sum_i I(X^{n\setminus i},Y_{\mathrm{r}}^{i-1};Y_{i}|X_i,Y_{i+1}^n) - \sum_i I(X^{n\setminus i};Y_{i}|X_i,Y_{i+1}^n) \nonumber\\
    & \quad - \sum_i I(X^{n\setminus i};Y_{\mathrm{r}i},Y_{i}|X_i,Y_{\mathrm{r}}^{i-1},Y_{i+1}^n)\nonumber \\
    & = \sum_i I(Y_{\mathrm{r}}^{i-1},Y_{i+1}^n,X_{\mathrm{r}i};Y_{i}|X_i) -\sum_i I(X^{n\setminus i},Y_{i+1}^n;Y_{i}|X_i) \nonumber\\
    & \qquad - \sum_i I(X^{n\setminus i};Y_{\mathrm{r}i}|X_i,Y_{\mathrm{r}}^{i-1},Y_{i+1}^n,Y_{i})\nonumber\\
    & \leq \sum_i I(Y_{\mathrm{r}}^{i-1},Y_{i+1}^n,X_{\mathrm{r}i};Y_{i}|X_i)-\sum_i I(X^{n\setminus i},Y_{i+1}^n;Y_{i}|X_i),\nonumber
    \end{align}}
 where $(a)$ follows from the fact that $X_{\mathrm{r}i}$ is a function of $Y_{\mathrm{r}}^{i-1}$, $(b)$ follows because $(X^{n\setminus i},Y_{\mathrm{r}}^{i-1}) \to (X_i,X_{\rmr i}) \to Y_{\mathrm{r}i}$ form a Markov chain, $(c)$ follows from Csisz\'ar-K\"orner-Marton identity and $(d)$
follows because $(X^{n\setminus i},Y_{i+1}^n) \to (X_i,Y_{\mathrm{r}}^{i-1},Y_{\mathrm{r}i}) \to Y_{i}$ form a Markov chain.
Next, observe that
\begin{align*}I(V;Y_{\mathrm{r}}|Q)-I(V;Y|Q)&=
I(U;Y_{\mathrm{r}}|Q)-I(U;Y|Q)
=\sum_i I(Y_{\mathrm{r}}^{i-1};Y_{\mathrm{r}i})-
\sum_i I(Y^{i-1};Y_i)
\end{align*}
From the identification of $V_i=(Y_{i+1}^n, Y_{\mathrm{r}}^{i-1})$, we see that $H(X_{\rmr i}|V_i)=0$. Therefore, $I(V,X_\rmr;Y_{\mathrm{r}}|Q)-I(V,X_\rmr;Y|Q)=
I(U;Y_{\mathrm{r}}|Q)-I(U;Y|Q)$. 
Thus,
$I(V,X_\rmr;Y_{\mathrm{r}})-I(V,X_\rmr;Y)=
I(U;Y_{\mathrm{r}})-I(U;Y)$. This yields \eqref{eq:UBcon}.

\smallskip

\noindent{\bf Proof of the equivalent form}.
From \eqref{eq:UBcon2} and \eqref{eqnUB1.5}, we deduce that
\begin{align}
    R&\leq 
    I(X;Y,Y_{\mathrm{r}}|X_{\mathrm{r}})-I(V;Y|X_{\mathrm{r}},Y_{\mathrm{r}})-I(X;Y_{\mathrm{r}}|V,X_{\mathrm{r}},Y)+
    I(X_\rmr;Y_\rmr)-I(V,X_{\mathrm{r}};Y_{\mathrm{r}})+I(V,X_{\mathrm{r}};Y)\nonumber
    \\&=
    I(X,X_{\mathrm{r}};Y)-I(V;Y_{\mathrm{r}}|X_{\mathrm{r}},X,Y).
\end{align}
Therefore the bound described in the alternate form is tighter than the bound in Theorem \ref{thm:maintheorem}. It remains to  prove the other direction. Note that to show the other direction, it suffices that \eqref{eq:UBcon2} holds for a maximizing distribution $p(u,x,x_{\mathrm{r}})p(v|x,x_{\mathrm{r}},y_{\mathrm{r}})$. 
This is established by identifying new maximizing auxiliary random variables $\tilde{U}$ and $\tilde{V}$ for which \eqref{eq:UBcon2} holds.

Assume that  \eqref{eq:UBcon2} does not hold for  a maximizing distribution $p(u,x,x_{\mathrm{r}})p(v|x,x_{\mathrm{r}},y_{\mathrm{r}})$. That is, 
\begin{align}
     I(U;Y_\rmr)-I(U;Y) = I(V,X_{\mathrm{r}};Y_{\mathrm{r}})-I(V,X_{\mathrm{r}};Y)> I(X_\rmr;Y_\rmr).\end{align}
Then, the inequality in \eqref{eqnUB1.5} is strict since
\begin{align}
    R&\leq 
    I(X,X_{\mathrm{r}};Y)-I(V;Y_{\mathrm{r}}|X_{\mathrm{r}},X,Y)\nonumber
    \\&=
    I(X;Y,Y_{\mathrm{r}}|X_{\mathrm{r}})-I(V;Y|X_{\mathrm{r}},Y_{\mathrm{r}})-I(X;Y_{\mathrm{r}}|V,X_{\mathrm{r}},Y)+
    I(X_\rmr;Y_\rmr)-I(V,X_{\mathrm{r}};Y_{\mathrm{r}})+I(V,X_{\mathrm{r}};Y)\label{eqnExpansiona}
    \\&<I(X;Y,Y_{\mathrm{r}}|X_{\mathrm{r}})-I(V;Y|X_{\mathrm{r}},Y_{\mathrm{r}})-I(X;Y_{\mathrm{r}}|V,X_{\mathrm{r}},Y).\nonumber
\end{align}
Consider two independent Bernoulli time-sharing random variables $Q_1\sim \mathrm{B}(\theta)$ and
$Q_2\sim \mathrm{B}(\theta)$ which are independent of previously defined random variables. Set $\tilde{V}=(V,Q_1)$ if $Q_1=0$ and $\tilde{V}=Q_1$ if $Q_1=1$. Similarly, set $\tilde{U}=(U,Q_2)$ if $Q_2=0$ and $\tilde{U}=(X_\rmr,Q_2)$ if $Q_2=1$. Observe that for any $\theta\in[0,1]$, the following equality  
\[  I(\tilde{U};Y_\rmr)-I(\tilde{U};Y) =  I(\tilde V,X_{\mathrm{r}};Y_{\mathrm{r}})-I(\tilde V,X_{\mathrm{r}};Y)  \]
still holds. Moreover, inequalities \eqref{eqnUB0} and \eqref{eqnUB1} also continue to hold. For $\theta=0$, we have
\[
    I(\tilde V,X_{\mathrm{r}};Y_{\mathrm{r}})-I(\tilde V,X_{\mathrm{r}};Y)> I(X_\rmr;Y_\rmr)
    \]
    and for $\theta=1$, we have
\[
    I(\tilde V,X_{\mathrm{r}};Y_{\mathrm{r}})-I(\tilde V,X_{\mathrm{r}};Y)\leq I(X_\rmr;Y_\rmr).
    \]
one can find  some $\theta^*\in[0,1]$ such that $\tilde{U}$ and $\tilde{V}$  satisfy the equality
\begin{align}
I(\tilde V,X_{\mathrm{r}};Y_{\mathrm{r}})-I(\tilde V,X_{\mathrm{r}};Y)=
I(\tilde U;Y_\rmr)-I(\tilde U;Y) = I(X_\rmr;Y_\rmr).\label{eqnexus1}
\end{align}
For this choice of $\theta^*$, we claim that  $p(\tilde u,x,x_{\mathrm{r}})p(\tilde v|x,x_{\mathrm{r}},y_{\mathrm{r}})$ is a maximizing distribution for which
\eqref{eq:UBcon2} holds with equality.  To show this, it suffices to show that for the new choice of $\tilde{U}, \tilde{V}$, as defined by \eqref{eqnexus1} the constraints in \eqref{eqnUB0} and \eqref{eqnUB1} continue to hold. This is immediate as the right hand-side of two constraints is only enlarged for any choice of $\theta \in [0,1]$. This follows from the fact that inequalities \eqref{eqnUB0} and \eqref{eqnUB1} hold for any $\theta\in[0,1]$.  Using \eqref{eqnexus1} and the expansion in \eqref{eqnExpansiona}, the constraint in \eqref{eqnUB1.5} is identical to \eqref{eqnUB0} and holds for $\theta^*$. 

The cardinality bounds on the auxiliary random variables come from the standard Caratheodory-Bunt \cite{bun34} arguments and is omitted. 

\noindent{\bf Equivalence of \eqref{eqnUB1} and \eqref{eqnUB1.1}, and the expressions in \eqref{eqnUB1.5}-\eqref{eqnUB2m1}.} Note that since $I(U;Y,Y_\rmr|X,X_\rmr)=0$, we have
\begin{align*}
    &I(X;Y,Y_{\mathrm{r}}|X_{\mathrm{r}})-I(U;Y|X_{\mathrm{r}},Y_{\mathrm{r}})  \\
    &\quad= I(X;Y,Y_{\mathrm{r}}|X_{\mathrm{r}},U) + I(U;Y,Y_{\mathrm{r}}|X_{\mathrm{r}}) -I(U;Y|X_{\mathrm{r}},Y_{\mathrm{r}})\\
    &\quad = I(X;Y|X_{\mathrm{r}},U) + I(X;Y_{\mathrm{r}}|X_{\mathrm{r}},U,Y) + I(U;Y_{\mathrm{r}}|X_{\mathrm{r}}),
\end{align*}    
which establishes the equivalence of \eqref{eqnUB1} and \eqref{eqnUB1.1}.

Note that since $I(V;Y|X,X_\rmr,Y_\rmr)=0$, we have
\begin{align*}
    &I(X;Y,Y_{\mathrm{r}}|X_{\mathrm{r}})-I(V;Y|X_{\mathrm{r}},Y_{\mathrm{r}})-I(X;Y_{\mathrm{r}}|V,X_{\mathrm{r}},Y)\\
    &\quad = I(X;Y_{\mathrm{r}}|X_{\mathrm{r}}) + I(V,X;Y|X_{\mathrm{r}},Y_{\mathrm{r}}) -I(V;Y|X_{\mathrm{r}},Y_{\mathrm{r}})-I(X;Y_{\mathrm{r}}|V,X_{\mathrm{r}},Y)\\
    &\quad  = I(X;Y_{\mathrm{r}}|X_{\mathrm{r}}) + I(X;Y|V,X_{\mathrm{r}},Y_{\mathrm{r}}) -I(X;Y_{\mathrm{r}}|V,X_{\mathrm{r}},Y)\\
    &\quad  = I(X;Y_{\mathrm{r}}|X_{\mathrm{r}})+I(X;Y|V,X_{\mathrm{r}})-I(X;Y_{\mathrm{r}}|V,X_{\mathrm{r}}),
    \end{align*}
    which establishes the equivalence of \eqref{eqnUB1.5} and \eqref{eqnUB2}.

Finally note that
\begin{align*}
    & I(X;Y_{\mathrm{r}}|X_{\mathrm{r}})+I(X;Y|V,X_{\mathrm{r}})-I(X;Y_{\mathrm{r}}|V,X_{\mathrm{r}})\\
    &\quad = I(X;Y_{\mathrm{r}}|X_{\mathrm{r}})+I(X;Y,V|X_{\mathrm{r}})-I(X;Y_{\mathrm{r}},V|X_{\mathrm{r}})\\
    &\quad = I(X;Y,V|X_\rmr)-I(V;X|X_\rmr,Y_\rmr),
   \end{align*}
   which establishes the equivalence of \eqref{eqnUB2} and \eqref{eqnUB2m1}.
This completes the proof.

\subsection{Proof of Theorem \ref{thm:cutsubopt}}
\label{sec:proof:thm:cutsubopt}

First, we show that the cutset bound is not tight for any non-zero values of $S_{12},S_{13},S_{23}$. The proof is by contradiction. The cutset upper bound is given by
\begin{align}
\label{eq:csobgau}
C_\mathrm{CS} = \max_{P_{X,X_{\mathrm{r}}}:\,\E(X^2)\le P,\, \E(X_\mathrm{r}^2)\le P} \min \{I(X,X_{\mathrm{r}};Y), I(X;Y,Y_{\mathrm{r}}|X_{\mathrm{r}})\}. 
\end{align}
Further, we know (from Section 16.5 of \cite{elk11}) that the maximum is attained via the unique jointly Gaussian distribution 
\begin{align}
    (X,X_{\mathrm{r}})_{opt} \sim \begin{cases}  \mathcal{N}\left(0,\begin{bmatrix} P & \rho^* P \\ \rho^* P & P \end{bmatrix}\right) & \text{if }S_{12} > S_{23},\\[13pt]
    \mathcal{N}\left(0,\begin{bmatrix} P &0 \\ 0 & P \end{bmatrix}\right) & \text{if }S_{12} \leq  S_{23}, \end{cases}
    \label{eq:XXropt}
\end{align}
where $\rho^* \in (0,1)$ satisfies $I(X,X_{\mathrm{r}};Y)= I(X;Y,Y_{\mathrm{r}}|X_{\mathrm{r}})$. Note that $C_\mathrm{CS} = I(X;Y,Y_{\mathrm{r}}|X_{\mathrm{r}})$ holds (is tight) in both of these cases.

The upper bound in Corollary \ref{cor:weakmainOB} is
\begin{align}
\label{eq:corubgau}
    C_{\mathrm{UB}} &= \max_{P_{X,X_{\mathrm{r}}}:\,\E(X^2)\le P,\, \E(X_\mathrm{r}^2)\le P} \max_{P_{V|X,X_{\mathrm{r}},Y_{\mathrm{r}}}\in \mathcal{M}} I(X;Y,Y_{\mathrm{r}}|X_{\mathrm{r}})-I(V;Y|X_{\mathrm{r}},Y_{\mathrm{r}})-I(X;Y_{\mathrm{r}}|V,X_{\mathrm{r}},Y),
\end{align}
where $\mathcal{M}$ is the set of distributions satisfying
\begin{align}
\label{eq:corubcons}
    I(V,X_{\mathrm{r}};Y_{\mathrm{r}})-I(V,X_{\mathrm{r}};Y)&\leq {\min\left[I(X_\rmr;Y_\rmr),~\max_{p(u|x,x_{\mathrm{r}})}\left(I(U;Y_{\mathrm{r}})-I(U;Y)\right)\right]}.
\end{align}
Observe that for $P_{V|X,X_{\mathrm{r}},Y_{\mathrm{r}}}\in \mathcal{M}$
\begin{align*}
    & I(X;Y,Y_{\mathrm{r}}|X_{\mathrm{r}})-I(V;Y|X_{\mathrm{r}},Y_{\mathrm{r}})-I(X;Y_{\mathrm{r}}|V,X_{\mathrm{r}},Y)\\
    & \quad \leq I(X;Y,Y_{\mathrm{r}}|X_{\mathrm{r}})-I(V;Y|X_{\mathrm{r}},Y_{\mathrm{r}})-I(X;Y_{\mathrm{r}}|V,X_{\mathrm{r}},Y)\\
    & \qquad \quad + I(X_\rmr;Y_\rmr) + I(V,X_{\mathrm{r}};Y) - I(V,X_{\mathrm{r}};Y_{\mathrm{r}})\\
    & \quad = I(X;Y,Y_{\mathrm{r}}|X_{\mathrm{r}}) + I(X_\rmr;Y) -I(V;Y_\rmr|X_{\mathrm{r}},Y)-I(X;Y_{\mathrm{r}}|V,X_{\mathrm{r}},Y)\\
    & \quad = I(X,X_\rmr;Y) - I(V;Y_\rmr|X,X_\rmr,Y).
\end{align*}
Hence for every $P_{X,X_{\mathrm{r}}}$, the expression on the right-hand-side of \eqref{eq:csobgau} is at least as large that of \eqref{eq:corubgau}. Therefore, if $C_{\mathrm {UB}}=C_{CS}$, the expressions in \eqref{eq:csobgau} and \eqref{eq:corubgau} must match for the unique maximizer of \eqref{eq:csobgau} given in \eqref{eq:XXropt}. Consequently for the $P_{X,X_{\mathrm{r}}}$ given in \eqref{eq:XXropt}, from \eqref{eq:corubgau} and $C_\mathrm{CS} = I(X;Y,Y_{\mathrm{r}}|X_{\mathrm{r}})$,
there must exist a $P_{V|X,X_{\mathrm{r}},Y_{\mathrm{r}}} \in \mathcal{M}$ such that 
\begin{align*}
 I(V;Y|X_{\mathrm{r}},Y_{\mathrm{r}})+ I(X;Y_{\mathrm{r}}|V,X_{\mathrm{r}},Y) &= 0.
\end{align*}
Our argument below shows that no such distribution exists.

Using part $(i)$ of Lemma \ref{lmm1},  $I(V;Y|X_{\mathrm{r}},Y_{\mathrm{r}})=0$ yields $I(X;V|X_{\mathrm{r}},Y_{\mathrm{r}}) = 0$, and using part $(ii)$ of Lemma \ref{lmm1}   $I(X;Y_{\mathrm{r}}|V,X_{\mathrm{r}},Y)=0$ yields   $I(X;Y_{\mathrm{r}}|V,X_{\mathrm{r}})=0$. This gives the double Markovity conditions in Lemma \ref{le:dbma}: given $X_{\mathrm{r}}=x_{\mathrm{r}}$, $V\rightarrow Y_{\mathrm{r}}\rightarrow X$ and $Y_{\mathrm{r}}\rightarrow V\rightarrow X$. Thus there exists functions $g(X_{\mathrm{r}},V)$ and $f(X_{\mathrm{r}},Y_{\mathrm{r}})$ such that $f(X_{\mathrm{r}},Y_{\mathrm{r}})=g(X_{\mathrm{r}},V)$ with probability 1, and 
\[
I(X;Y_{\mathrm{r}},V|X_{\mathrm{r}},f(X_{\mathrm{r}},Y_{\mathrm{r}}),g(X_{\mathrm{r}},V))=0.
\]
Using an abuse of notation, let $B=f(X_{\mathrm{r}},Y_{\mathrm{r}})=g(X_{\mathrm{r}},V)$ almost surely, hence  we have $I(X;V,Y_{\mathrm{r}}|B,X_{\mathrm{r}})=0$. This implies that $I(X;Y_{\mathrm{r}}|B,X_{\mathrm{r}})=0$. We also have the Markov chain $B\rightarrow (X_{\mathrm{r}},Y_{\mathrm{r}})\rightarrow X$.  Observe that since $X,X_{\mathrm{r}},Y_{\mathrm{r}}$ are jointly Gaussian, we can express $X=a X_{\mathrm{r}} + cY_{\mathrm{r}} + \hat{Z}$, where $\hat{Z}$ is independent of $Y_{\mathrm{r}},X_{\mathrm{r}}$.  Here 
\[c= \frac{S_{12}(1-\rho^*)}{g_{12}(1+S_{12}(1-\rho^*))} \neq 0.
\]
Hence $\hat{Z}$ is also independent of $B$. Now we have $I(Y_{\mathrm{r}};cY_{\mathrm{r}} + \hat{Z}|B,X_{\mathrm{r}})=0$. This holds only if  $Y_{\mathrm{r}}$ is a function of $(B,X_{\mathrm{r}})$ (see part $(iii)$ of Lemma \ref{lmm1}). Consequently $Y_{\mathrm{r}}$ is a function of $(V,X_{\mathrm{r}})$ (since $B=g(X_{\mathrm{r}},V)$) which implies that
$I(V,X_\rmr;Y_{\mathrm{r}})=\infty$.  Since $I(X_\rmr;Y_\rmr)$ and $I(V,X_\rmr;Y)$ are finite and bounded from above, due to the presence of independent additive Gaussian noise, the constraint  $I(V,X_\rmr;Y_{\mathrm{r}}) - I(V,X_\rmr;Y)\leq I(X_\rmr;Y_\rmr)$ cannot hold. Therefore, from \eqref{eq:corubcons}, $P_{V|X,X_{\mathrm{r}},Y_{\mathrm{r}}} \notin \mathcal{M}$. 
This establishes the requisite contradiction.

The optimality of Gaussian random variables for the evaluation of Corollary \ref{cor:weakmainOB} is established in Lemma \ref{le:gauoptcor} in Appendix \ref{appndxB}.  For Gaussian random variables we parameterize the joint distribution as follows: 
\begin{equation} 
K_{X,X_{\rmr},Z_1}=\begin{bmatrix} P & \rho P & 0\\ \rho P & P & 0
\\0&0&1 \end{bmatrix}
\end{equation}
for some $\rho \in [-1,1]$ and
\begin{equation} 
K_{X,Z_1|VX_{\rmr}}=\begin{bmatrix} P(1-\rho^2)\alpha & \sigma \sqrt{P(1-\rho^2)\alpha \beta} \\ \sigma \sqrt{P(1-\rho^2)\alpha \beta} & \beta  \end{bmatrix} \preceq\begin{bmatrix} P(1-\rho^2) & 0\\ 0 & 1\end{bmatrix}.
\end{equation}
This gives us the conditions: $0\leq \alpha,\beta\leq 1$ and 
$(1-\alpha)(1-\beta)\geq \sigma^2\alpha \beta$.

Now note that
\begin{align*}
\max_{P_{U|X,X_{\rmr}}} (I(U;Y_{\rmr})-I(U;Y))&=
I(X,X_\rmr;Y_\rmr)-I(X,X_\rmr;Y)+\max_{P_{U|X,X_{\rmr}}}(I(X,X_\rmr;Y|U)-I(X,X_\rmr;Y_\rmr|U))
\\&\stackrel{(a)}{=}I(X,X_\rmr;Y_\rmr)-I(X,X_\rmr;Y)+\frac{1}{2}\log\lambda_{\max}
\\&=
    \frac12\log\left( S_{12}+1\right)-\frac12\log\left(1+S_{13}+S_{23}+2\rho\sqrt{S_{13}S_{23}} \right)
    +\frac{1}{2}\log\lambda_{\max},
\end{align*}
where $(a)$ follows from \cite[Remark 20]{gon22} and 
$\lambda_{\max}$ is the larger root of the quadratic polynomial
\begin{align}
2\rho \sqrt{S_{13}S_{23}}+S_{13}+S_{23}+1-\lambda\big(S_{23}S_{12}(1-\rho^2)+S_{13}+S_{23}+S_{12}+2+2\rho\sqrt{S_{13}S_{23}}\big)+\lambda^2(S_{12}+1)=0.\label{e1}
\end{align}

Next, we use an observation about the maximizing variables from Lemma \ref{le:auxgau} below to determine the value of $\sigma$. The constraint in \eqref{eqn:NC} is a linear equality in $\sigma$ and yields
\begin{align*}
    &\sigma 
=
\frac{(1-\rho^2)\alpha S_{13}
+1}{2T\sqrt{S_{12}(1-\rho^2)\alpha \beta}}-\frac{(1-\rho^2)\alpha S_{12}
+\beta}{2\sqrt{S_{12}(1-\rho^2)\alpha \beta}},
\end{align*}
where 
\[
\frac{1}{2}\log(T)= \min\big\{-
I(X;Y_\rmr|X_\rmr)+I(X,X_\rmr;Y),  \max_{P_{U|X,X_{\rmr}}}(I(X,X_\rmr;Y|U)-I(X,X_\rmr;Y_\rmr|U))\big\}.
\]
Hence, we obtain that any achievable rate $R$ must satisfy the conditions 
\begin{align}
R&\leq 
\frac12\log\left( (1-\rho^2)S_{12}+1\right)
-\frac12\log\left(\beta + S_{12}(1-\rho^2)\alpha + 2\sigma\sqrt{S_{12}(1-\rho^2)\alpha\beta}  \right)\nonumber
\\&\quad
+\frac12\log\left( \beta(1-\sigma^2)\right)
+\frac12\log\left( (1-\rho^2)\alpha S_{13}+1\right)
 \end{align}
for some $0\leq \alpha,\beta\leq 1$, $\sigma,\rho\in[-1,1]$ such that 
$(1-\alpha)(1-\beta)\geq \sigma^2\alpha \beta$. This completes the proof.

\begin{lemma}
\label{le:auxgau}
Any maximizing distribution for the optimization problem of computing the maximum rate given by Theorem \ref{thm:cutsubopt} must satisfy the condition
\begin{align}
    I(V,X_{\rmr};Y_{\rmr})-I(V,X_{\rmr};Y)&=\min\big\{ I(X_\rmr;Y_\rmr),~~ \max_{P_{U|X,X_{\rmr}}}(I(U;Y_{\rmr})-I(U;Y))\big\}.\label{eqn:NC}
\end{align}
\end{lemma}
\begin{proof}
Assume that the maximizer does not satisfy \eqref{eqn:NC}.
Then, from the alternate form of Theorem \ref{thm:maintheorem}
\[I(V,X_{\rmr};Y_{\rmr})-I(V,X_{\rmr};Y)<\min\big\{ I(X_\rmr;Y_\rmr),~~ \max_{P_{U|X,X_{\rmr}}}(I(U;Y_{\rmr})-I(U;Y))\big\}.\]
Consider a Bernoulli time-sharing random variable $Q\sim \mathrm{B}(\theta)$ and set $\tilde{V}=(V,Q)$ if $Q=0$ and $\tilde{V}=(Y_\rmr+\zeta W,Q)$ if $Q=1$, where $W$ is a standard Gaussian random variable, independent of previously defined random variables. When $\zeta=0$, replacing $V$ by $\tilde{V}$ strictly increases the right hand side of \eqref{eqnCoreq1} for any $\theta>0$. On the other hand, for any arbitrary $\theta>0$, we have
\[
\lim_{\zeta\rightarrow 0}(I(\tilde V,X_{\mathrm{r}};Y_{\mathrm{r}})-I(\tilde V,X_{\mathrm{r}};Y))=\infty.
\]
Also for any $\zeta>0$, 
\[
\lim_{\theta\rightarrow 0}I(\tilde V,X_{\mathrm{r}};Y_{\mathrm{r}})-I(\tilde V,X_{\mathrm{r}};Y)=I( V,X_{\mathrm{r}};Y_{\mathrm{r}})-I( V,X_{\mathrm{r}};Y).
\]
Therefore, one can find suitable $\theta,\zeta>0$ such that for $\tilde{V}$ \eqref{eqnCoreq2} has a larger value and still the constraints are satisfied for $\tilde{V}$. This is a contradiction.

\end{proof}

\subsection{Proof of Proposition \ref{prop:mainprirel}}
\label{sec:proofProp1}

Let $R_1$ be the upper bound of Theorem \ref{thm:maintheorem}, $R_2$ be the upper bound of Proposition \ref{prop:mainprirel} in its  form without any constraints, and $R_3$ be the the upper bound of Proposition \ref{prop:mainprirel} in its  form with a constraint. It suffices to show that $R_1\leq R_2$, $R_2\leq R_3$ and $R_3\leq R_1$.

To show $R_1\leq R_2$, assume that a rate $R$ satisfies the inequalities given in Theorem \ref{thm:maintheorem} for a pair auxiliary variables $(U,V)$. Theorem \ref{thm:maintheorem} implies that{\allowdisplaybreaks
\begin{align}
R&\leq I(X;Y,Y_{\mathrm{r}}|X_{\mathrm{r}})-I(V;Y|X_{\mathrm{r}},Y_{\mathrm{r}})-I(X;Y_{\mathrm{r}}|V,X_{\mathrm{r}},Y)
\nonumber\\&\leq I(X;Y_1,Y_{\mathrm{r}})-I(V,X_{\mathrm{r}};Y_1|Y_{\mathrm{r}})-I(X;Y_{\mathrm{r}}|V,X_{\mathrm{r}},Y_1),\nonumber\\
R&\leq I(X,X_{\mathrm{r}};Y)-I(V;Y_{\mathrm{r}}|X_{\mathrm{r}},X,Y)\nonumber\\
&=  I(X;Y_1) + I(X_\rmr;Y_2|Y_1) - I(Y_\rmr;X_\rmr|X) -I(V,X_{\mathrm{r}};Y_{\mathrm{r}}|X,Y_1)\label{eqnredad}
\\&\leq  I(X;Y_1)+C_0-I(V,X_{\mathrm{r}};Y_{\mathrm{r}}|X,Y_1),\nonumber
 \end{align}}
 where \eqref{eqnredad} follows from $I(Y_1,Y_\rmr;X_\rmr|X)=0$ and $I(Y_2;X,Y_1|X_\rmr)=0$.
 This corresponds to the constraints in Proposition \ref{prop:mainprirel} for $\tilde V=(V,X_{\mathrm{r}})$.

Next, we show that $R_2\leq R_3$. Let $R$ be a rate satisfying
 \begin{align}
R&\leq I(X;Y_1,Y_{\mathrm{r}})-I( V;Y_1|Y_{\mathrm{r}})-I(X;Y_{\mathrm{r}}| V,Y_1)\label{eqnSF1}
\\
R&\leq I(X;Y_1)+C_0-I(V;Y_{\mathrm{r}}|X,Y_1)\label{eqnSF2}
 \end{align}
for some $p(x)p(y_1,y_{\mathrm{r}}|x)p(v|x,y_{\mathrm{r}})$. It suffices to show that 
\begin{align}
 I( V;Y_{\mathrm{r}})-I( V;Y_1)\leq C_0\label{eqnNCCP}
 \end{align}
is satisfied by a maximizer of the minimum of \eqref{eqnSF1} and \eqref{eqnSF2}.
Assume that \eqref{eqnNCCP} is violated by a maximizing distribution and we have
\[I( V;Y_{\mathrm{r}})-I( V;Y_1)> C_0.\]
In this case, the inequality \eqref{eqnSF1} is strictly redundant since
\begin{align*}
    R&\leq I(X;Y_1) + C_0 - I(V;Y_{\mathrm{r}}|X,Y_1)
    \\&
    <I(X;Y_1) +I( V;Y_{\mathrm{r}})-I( V;Y_1)
    - I(V;Y_{\mathrm{r}}|X,Y_1)
    \\&=
    I(X;Y_1,Y_{\mathrm{r}})-I(V;Y_1|Y_{\mathrm{r}}) - I(X;Y_{\mathrm{r}}|V,Y_1).
\end{align*}
Now, let $Q\sim B(\epsilon)$ be a time-sharing random variable and let $\tilde{V}=(Q,V)$ if $Q=1$, and $\tilde{V}=Q$ if $Q=0$. Then for an appropriate choice for $\epsilon$, we have
\[I(\tilde V;Y_{\mathrm{r}})-I(\tilde V;Y_1)=C_0.\]
For  $\tilde{V}$, the constraint \eqref{eqnSF1} is still redundant. Moreover, \eqref{eqnSF1} and \eqref{eqnSF2} have not decreased, and the constraint still holds. This shows that without loss of generality we can impose the constraint in  \eqref{eqnNCCP} when evaluating the region.

Finally, we show that $R_3\leq R_1$. Assume that a rate $R$ satisfies
 \begin{align}
R&\leq I(X;Y_1,Y_{\mathrm{r}})-I( V;Y_1|Y_{\mathrm{r}})-I(X;Y_{\mathrm{r}}| V,Y_1)
 \end{align}
 for some $p(x)p(y_1,y_{\mathrm{r}}|x)p(v|x,y_{\mathrm{r}})$ such that
\begin{align}
 0\leq I( V;Y_{\mathrm{r}})-I( V;Y_1)\leq C_0.
 \end{align}
 The constraint $0\leq I( V;Y_{\mathrm{r}})-I( V;Y_1)$ is added above since it holds for any maximizer $V$. To see this, assume that 
\[
I(V;Y_{\mathrm{r}})-I(V;Y_1)< 0,
\]
then we have
\begin{align*}
    R&\leq I(X;Y_1,Y_{\mathrm{r}})-I(V;Y_1|Y_{\mathrm{r}}) - I(X;Y_{\mathrm{r}}|V,Y_1)\\
    & = I(X;Y_1) + I(V;Y_{\mathrm{r}}) - I(V;Y_1) - I(V;Y_{\mathrm{r}}|X,Y_1)\\
    & < I(X;Y_1),
\end{align*}
which is less than the direct transmission bound. 
For such a joint distribution on $(V,X,Y_1,Y_{\mathrm{r}})$, let  
 $(U,X_{\mathrm{r}},Y_2)$ be independent of $(V,X,Y_1,Y_{\mathrm{r}})$ such that (i) $I(X_{\mathrm{r}};Y_2)=C_0$, (ii) $U\rightarrow X_{\mathrm{r}}\rightarrow Y_2$ forms a Markov chain, and (iii) 
$$I(U;Y_2)=I(X_{\mathrm{r}};Y_2)-[I(V;Y_{\mathrm{r}})-I(V;Y_1)].$$
This is feasible for some $U\rightarrow X_{\mathrm{r}}\rightarrow Y_2$ since 
\[
0\leq I(V;Y_{\mathrm{r}})-I(V;Y_1)\leq C_0=I(X_{\mathrm{r}};Y_2).
\]
Now, consider the choice of $(U,V)$ in  
Theorem \ref{thm:maintheorem}. 
One can directly verify that the rate $R$ satisfies the inequalities given in the alternate form of Theorem  \ref{thm:maintheorem} for our choice of joint distribution of $( U,V,X,X_{\mathrm{r}},Y_1,Y_2,Y_{\mathrm{r}})$. Furthermore,
\begin{align*}
    I(V,X_{\mathrm{r}};Y_{\mathrm{r}})-I(V,X_{\mathrm{r}};Y)
    &=I(V;Y_{\mathrm{r}})-I(V;Y_1)-I(X_{\mathrm{r}};Y_2)
    \\&=I( U;Y_{\mathrm{r}})-I(U;Y).
\end{align*}
This completes the proof.

\subsection{Direct Proof of Proposition \ref{prop:mainprirel}}\label{sec:alternaproof}
It is also possible to establish Proposition \ref{prop:mainprirel} directly as follows. Let  $M_\rmr\in[1:2^{2C_0}]$ be the message from relay to the receiver ($M_\rmr$ is a function of $Y_\rmr^n$), and define $V_i=(Y_{1i+1}^{n}, Y_{\rmr }^{i-1}, M_\rmr)$ and
$V=(Q,V_Q)$, $X=X_Q, Y_1=Y_{1Q}, Y_{\mathrm{r}}=Y_{\mathrm{r}Q}$ for a time-sharing random variable $Q\stackrel{(d)}{=} \textrm{Uniform}[1:n]$. Since the message $M$ can be recovered from $(M_\rmr,Y_1^n)$ with a probability of error bounded by $\eps_n$, by Fano's inequality we have $nR\leq I(X^n;M_\rmr,Y_1^n)+n\epsilon_n$. Next, we have
\begin{align*}
    nC_0&\geq H(M_\rmr)\\&\geq H(M_\rmr|Y_1^n)
    \\&
    =I(M_\rmr;Y_\rmr^n|X^n,Y_1^n)+I(X^n;M_\rmr|Y_1^n)
    \\
    &
    =\sum_i I(M_\rmr;Y_{\rmr i}|Y_{\rmr }^{i-1},X^n,Y_1^n)+I(X^n;M_\rmr|Y_1^n)
    \\
    &
    =\sum_i I(M_\rmr,X^n,Y_{\rmr }^{i-1},Y_1^n;Y_{\rmr i}|X_i,Y_{1i})+I(X^n;M_\rmr|Y_1^n)
    \\&\geq \sum_i I(V_i;Y_{\rmr i}|X_i,Y_{1i})+ I(X^n;M_\rmr,Y_1^n)-I(X^n;Y_1^n)
    \\&\geq \sum_i I(V_i;Y_{\rmr i}|X_i,Y_{1i})+ nR-n\epsilon_n-\sum_i I(X_i;Y_{1i})
    \\&=n\cdot \big(I(V;Y_{\rmr }|X,Y_{1},Q)+ R-\epsilon_n- I(X;Y_{1}|Q)\big)
    \\&\geq n\cdot \big(I(V;Y_{\rmr }|X,Y_{1})+ R-\epsilon_n- I(X;Y_{1})\big).
    \end{align*}
Therefore, we obtain the upper bound
\begin{align*}
R&\leq I(X;Y_1)+C_0-I(Y_\rmr;V|X,Y_1).
\end{align*}
Next, consider
\begin{align}
\frac1n I(X^n;M_\rmr, Y_1^n)&\leq \frac1n \left(I(X^n;Y_{\mathrm{r}}^n)+I(X^n; Y_1^n|M_\rmr)-I(X^n;Y_{\mathrm{r}}^n|M_\rmr)\right)\nonumber
\\&\nonumber\leq \frac1n\sum_i I(X_i;Y_{\mathrm{r}i})+\frac1n\left(\sum_{i}I(X^n;Y_{1i}|Y_{\mathrm{r}}^{i-1},Y_{1 i+1}^n,M_\rmr)-\sum_{i}I(X^n;Y_{\mathrm{r}i}|Y_{\mathrm{r}}^{i-1},Y_{1 i+1}^n,M_\rmr)\right)
\\&\leq I(X;Y_{\mathrm{r}})+\frac1n\left(\sum_{i}I(X_i;Y_{1i}|Y_{\mathrm{r}}^{i-1},Y_{1i+1}^n,M_\rmr)-\sum_{i}I(X_i;Y_{\mathrm{r}i}|Y_{\mathrm{r}}^{i-1},Y_{1i+1}^n,M_\rmr)\right)\label{eqnToExplain}
\\&=I(X;Y_\rmr)+I(X;Y_1|V)-I(X;Y_\rmr|V)\nonumber
\\&=
I(X;V,Y_1)-I(V;X|Y_\rmr),\nonumber
\end{align}
where \eqref{eqnToExplain} follows from the fact that $\frac1n\sum_i I(X_i;Y_{\mathrm{r}i})=I(X;Y_\rmr|Q)\leq I(X;Y_\rmr)$ and
\begin{align*}
I(X^n;Y_{1i}|X_i, Y_{\mathrm{r}}^{i-1},Y_{1 i+1}^n,M_\rmr)
-I(X^n;Y_{\rmr i}|X_i, Y_{\mathrm{r}}^{i-1},Y_{1 i+1}^n,M_\rmr)
&
\leq I(X^n;Y_{1i}|X_i, Y_{\rmr i},Y_{\mathrm{r}}^{i-1},Y_{1 i+1}^n,M_\rmr)
\\&=0. 
\end{align*}
Thus,
\[R\leq I(X;V,Y_1)-I(V;X|Y_\rmr).
\]

The Markov chain $V_i\rightarrow (X_{i},Y_{\rmr i})\rightarrow Y_{1i}$ follows from the memorylessness of the channel $p(y_1,y_\rmr|x)$.

\subsection{Proof of Proposition \ref{Proposition5}}
\label{proofthm:Prop5}

The sufficiency of considering Gaussian random variables for the evaluation of Proposition \ref{prop:mainprirel} in this context is established in Lemma \ref{lem:gausuffcor} in the Appendix. Let the covariance of $X,Y_\rmr$ given ${V}$ be
\[
K_{X,Y_\rmr|{V}}=
\begin{bmatrix} K_1 & \rho \sqrt{K_1K_2}\\
\rho \sqrt{K_1K_2}& K_2
\end{bmatrix}\preceq \begin{bmatrix} P & P\\
P& P+N_{\rmr}
\end{bmatrix}
\]
for some $0 \leq K_1\leq P$ and $0 \leq K_2\leq N_{\rmr}+P$ and $\rho\in[-1,1]$ such that
\begin{align}
(P-K_1)(P+N_{\rmr}-K_2)\geq (P-\rho\sqrt{K_1K_2})^2.\label{eq:rhocon22}
\end{align}
Then, the upper bound becomes the maximum of
\begin{align}
    \frac12 \log\left(\frac{P+ N_{\rmr}}{ N_{\rmr}}\right)
+\frac12\log\left(\frac{K_1+N_1}{N_1}\right)
+\frac12\log\left(1-\rho^2\right)\label{eqnValPNNr}
\end{align}
subject to
\[
\frac12\log(P+ N_{\rmr})-\frac12\log(K_2)-\frac12\log(P+N_1)+\frac12 \log(K_1+N_1)\leq C_0.
\]
Lemma \ref{NewLemma} shows that the optimal choice of $K_1, K_2$ and $\rho$ is:
\begin{align*}
    K_1^*&=\begin{cases}-N_1\frac{2^{2C_0}(P+N_{\rmr})-(P+N_{\rmr})}{2^{2C_0}(P+N_1)-(P+N_{\rmr})}+P\frac{2^{2C_0}(N_1-N_{\rmr})}{2^{2C_0}(P+N_1)-(P+N_{\rmr})}& S_{12}\geq S_{13}+S_{23}+S_{13}S_{23},\\
    \\
    P\left(1-\frac{P(N_1+P)^2(2^{2C_0}-1)}{
(P+N_{\rmr})(2^{2C_0}-1)((N_1+P)^2-N_1^22^{-2C_0})+(N_{\rmr}-N_1)P^2}\right)&\text{otherwise}.
    \end{cases}\\
K^*_2&=\frac{(K^*_1+N_1)(P+N_{\rmr})}{(P+N_1)2^{2C_0}},\\
    \rho^*&= \frac{P-\sqrt{(P-K^*_1)(P+N_{\rmr}-K^*_2)}}{\sqrt{K^*_1K^*_2}}.
\end{align*}
 Substituting these values in \eqref{eqnValPNNr} and replacing $2^{2C_0} = (1+S_{23})$, $S_{12}=P/N_\rmr$ and $S_{13}=P/N_1$, the upper bound reduces to the form given in the statement of Proposition \ref{Proposition5}. This completes the proof.

\subsection{Proof of Lemma \ref{thm:GSCase}}
\label{proofthm:GSCase}

We seek to calculate
\begin{align} 
 \frac{1}{2}\log \left(1+S_{13}\right)+\sup_{\theta\in\left[\arcsin (\frac{1}{1+S_{23}}),\frac{\pi}{2}\right]}\min \Bigg\{ C_0+\log \sin \theta, 
 \min_{\omega\in \left(\frac{\pi}{2}-\theta, \frac{\pi}{2}\right]}  h_{\theta}(\omega)
\Bigg\},\label{eqnSup1}
\end{align}
where
\begin{align}&~h_{\theta}(\omega) =    \frac{1}{2}\log \left(\frac{
\left(S_{12}+S_{13}+\sin ^2(\omega)-2\cos(\omega)\sqrt{S_{12}S_{13}}
\right)\sin^2\theta}{(S_{13}+1)(\sin ^2 \theta - \cos^2 \omega)} \right).\label{eqn:htheta22}
 \end{align}
We make the following observations:
 \begin{itemize}
     \item Observation 1: The term $C_0 + \log(\sin(\theta))$ is increasing in $\theta \in [0, \frac{\pi}{2}]$.
     \item Observation 2: The term $\min_{\omega\in \left(\frac{\pi}{2}-\theta, \frac{\pi}{2}\right]}  h_{\theta}(\omega)$ is decreasing in $\theta$. To see this, observe that for $0 < \theta_1 < \theta_2 < \frac{\pi}{2}$ we have
 \begin{align*}
 &\min_{\omega \in (\frac{\pi}{2} - \theta_2, \frac{\pi}{2}]}\frac{1}{2}\log \left(\frac{
\left(S_{12}+S_{13}+\sin ^2(\omega)-2\cos(\omega)\sqrt{S_{12}S_{13}}
\right)\sin^2\theta_2}{(S_{13}+1)(\sin ^2 \theta_2 - \cos^2 \omega)} \right) \\& \leq \min_{\omega \in (\frac{\pi}{2} - \theta_1, \frac{\pi}{2}]}\frac{1}{2}\log \left(\frac{
\left(S_{12}+S_{13}+\sin ^2(\omega)-2\cos(\omega)\sqrt{S_{12}S_{13}}
\right)\sin^2\theta_2}{(S_{13}+1)(\sin ^2 \theta_2 - \cos^2 \omega)} \right)\\
 & \leq \min_{\omega \in (\frac{\pi}{2} - \theta_1, \frac{\pi}{2}]}\frac{1}{2}\log \left(\frac{
\left(S_{12}+S_{13}+\sin ^2(\omega)-2\cos(\omega)\sqrt{S_{12}S_{13}}
\right)\sin^2\theta_1}{(S_{13}+1)(\sin ^2 \theta_1 - \cos^2 \omega)} \right).
 \end{align*}
 Therefore $\min_{\omega\in \left(\frac{\pi}{2}-\theta, \frac{\pi}{2}\right]}  h_{\theta}(\omega)$ is decreasing in $\theta$. 
 \item Observation 3: Given any fixed $\theta$, the minimizer $\omega_*$ of $\min_{\omega\in \left(\frac{\pi}{2}-\theta, \frac{\pi}{2}\right]}  h_{\theta}(\omega)$ can be explicitly computed (using routine algebra) as 
 the smaller root  of the quadratic polynomial
 \[x^2-\frac{S_{12}+S_{13}+\cos^2\theta}{\sqrt{S_{12}S_{13}}}x+\sin^2(\theta)=0
 \]
 which is
 \begin{align}
     \cos \omega_* & = \frac{S_{13}+S_{12} + \cos^2 \theta -  \sqrt{(S_{12}-S_{13})^2+\cos^2 \theta(4S_{12}S_{13} + 2S_{12}+2S_{13} + \cos^2\theta)}}{2\sqrt{S_{12}S_{13}}}.\label{valcosomega}
 \end{align}

\end{itemize}

 Consider two cases:
 \begin{itemize}
     \item \textit{Case 1}:  $S_{12}\geq S_{13}+S_{23}+S_{13}S_{23}$: we claim that the supremum in \eqref{eqnSup1} is attained at $\theta_*=\pi/2$. Using observations 1 and 2 above, it suffices to verify that 
     \[
 C_0+ \log(\sin(\theta_*)) \leq   \min_{\omega\in \left(\frac{\pi}{2}-\theta_*, \frac{\pi}{2}\right]}  h_{\theta_*}(\omega).
 \]
 For $\theta_*=\pi/2$, from \eqref{valcosomega} the minimizer $\omega_*$ satisfies
 \begin{align*}
     \cos \omega_* & = \sqrt{\frac{S_{13}}{S_{12}}}.
 \end{align*}
 Consequently, 
 \[h_\theta^*(\omega^*) = \frac{1}{2}\log \left(\frac{
 S_{12}+1}{S_{13}+1} \right).\]
Using the assumption that $S_{12}\geq S_{13}+S_{23}+S_{13}S_{23}$ we note that
 \begin{align*}&~h_{\theta}(\omega_*) =    \frac{1}{2}\log \left(\frac{
 S_{12}+1}{S_{13}+1} \right)\geq \frac{1}{2}\log(1+S_{23})=C_0 =C_0+ \log(\sin(\theta_*)).
 \end{align*}
 Substituting the maximizer $\theta_*=\pi/2$ in \eqref{eqnSup1} shows that the bound in Theorem \ref{thm:Wu19thm} is equivalent to
\begin{align*}
        R\leq
        \frac{1}{2}\log((1+S_{13})(1+S_{23})).
\end{align*}

     \item \textit{Case} 2: $S_{12}< S_{13}+S_{23}+S_{13}S_{23}$: let $\theta_* \in [0,\frac{\pi}{2}]$ be defined according to
 \begin{align}
\sin^2\theta_*
=\left(1+
\frac{S_{23}S_{12}}{(S_{13}+1)(S_{23}+1)-1}\right)\frac{1}{1+S_{23}}.\label{eqndefs}
\end{align}
Observe that the above definition for $\sin^2\theta_*$ makes sense because the right hand side of \eqref{eqndefs} is less than 1 by the assumption that $S_{12}< S_{13}+S_{23}+S_{13}S_{23}$.
     
     In order to show that the supremum in \eqref{eqnSup1} is attained at  $\theta_*$, we again use Observations 1 and 2 above. Now it suffices to verify that 
 \[
 C_0 + \log(\sin(\theta_*)) =  \min_{\omega\in \left(\frac{\pi}{2}-\theta_*, \frac{\pi}{2}\right]}  h_{\theta_*}(\omega).
 \]
 This can be verified by plugging in the value for $\cos(\omega_*)$ from \eqref{valcosomega}. 
Substituting the maximizer $\theta_*$ in \eqref{eqnSup1} shows that the bound in Theorem \ref{thm:Wu19thm} is equivalent to
\begin{align}
        R\leq
        \frac12\log\left(1+S_{13}+\frac{S_{12}(S_{13}+1)S_{23}}{(S_{13}+1)(S_{23}+1)-1}\right).
\end{align}
 \end{itemize}
 
 This completes the proof.


\subsection{Proof of Theorem \ref{thmBSCPrimitive}}
\label{proofthmBSCPrimitive}

Using the symmetrization argument in \cite{nair2013upper}, without loss of generality, we can restrict $X$ to be uniformly distribution when evaluating Proposition \ref{prop:mainprirel}. To see this, given some arbitrary $(V,X)$, consider $Q\sim \mathrm{B}(0.5)$ independent of $(V,X)$. Then, letting $V'=(V,Q)$, $X'=(X+Q \mod 2)$, $Y'_1=(Y_1+Q \mod 2)$ and $Y'_{\rmr}=(Y_{\rmr}+Q \mod 2)$, one can verify that $I(V';Y'_{\rmr})-I(V';Y'_1)=I(V;Y_{\rmr})-I(V;Y_1)$ and
\[
I(X';Y'_1,Y'_{\rmr})-I(V';Y'_1|Y'_{\rmr})-I(X';Y'_{\rmr}|V',Y'_1)\geq I(X;Y_1,Y_{\rmr})-I(V;Y_1|Y_{\rmr})-I(X;Y_{\rmr}|V,Y_1).
\]
Moreover, $X'$ is uniform.

From Proposition \ref{prop:mainprirel} for any $\lambda\geq 0$, any achievable rate $R$ must satisfy
\begin{align*}
  R&\leq  \max_{p(v|x,y_{\rmr})}\big (I(X;Y_1,Y_{\rmr})-I( V;Y_1|Y_{\rmr})-I(X;Y_{\rmr}| V,Y_1)+\lambda [C_0-I( V;Y_{\rmr})+I( V;Y_1)]\big)
   \\
   & =\max_{p(v|x,y_{\rmr})}\big( I(X;Y_{\rmr})
   +\lambda C_0-H(Y_1|X)
   +H(Y_{\rmr}|X,V)+(1-\lambda) [H(Y_1| V)-H(Y_{\rmr}| V)]\big)
   \\
   & = 1-2H_2(\rho)
   +\lambda C_0
   +\max_{p(v|x,y_{\rmr})}\big(H(Y_{\rmr}|X,V)+(1-\lambda) [H(Y_1| V)-H(Y_{\rmr}| V)]\big).
\end{align*}
Without loss of generality we can assume that $\la\in[0,1]$. To see this, observe that for $\la=1$ the optimal choice for $V$ is a constant and the upper bound becomes $1-2H_2(\rho)
   + C_0+H(Y_{\rmr}|X)$. If $\la>1$, the upper bound is greater than or equal to $1-2H_2(\rho)
   + C_0+H(Y_{\rmr}|X)$ since $V$ equal to a constant is one possible choice for $V$. 

Let $p(x,y_{\rmr}) = (p_{00},p_{01},p_{10},p_{11}) $ and $f_\la: p(x,y_{\rmr}) \mapsto (1-\la) \left(H(Y_1) - H(Y_{\rmr})\right) +  H(Y_{\rmr}|X)$. The maximum of $(1-\lambda) [H(Y_1| V)-H(Y_{\rmr}| V)]+H(Y_{\rmr}|X,V)$ is equal  to the upper concave envelope, $\mathscr{C}(f_\la),$ of the function $f_\la$  at
\begin{align*}
    \left(\frac{1-\rho}{2},\frac{\rho}{2},\frac{\rho}{2},\frac{1-\rho}{2} \right).
\end{align*}
Observe that 
\begin{align}f_\la(p_{00},p_{01},p_{10},p_{11})=  f_\la(p_{11},p_{10},p_{01},p_{00}).
\label{eq:symm}
\end{align}  
Let
\[
g_\la(c)= \max_{\substack{p_{00},p_{01},p_{10},p_{11}\\
p_{01}+p_{10}=c}} f_\la(p_{00},p_{01},p_{10},p_{11}).
\]
Now define $\tilde{f}_\la(p_{00},p_{01},p_{10},p_{11})=g_\la(p_{01}+p_{10})$. Clearly $\tilde{f}_\la(p_{00},p_{01},p_{10},p_{11}) \geq f_\la(p_{00},p_{01},p_{10},p_{11})$ pointwise, and consequently the upper concave envelope of $\mathscr{C}(\tilde{f_\la}) \geq \mathscr{C}(f_\la).$  On the other hand, we can invoke the symmetry in \eqref{eq:symm} to immediately conclude that 
\[
g_\la(p_{01}+p_{10}) = \tilde{f}_\la\left(\frac{p_{00} + p_{11}}{2},\frac{p_{01} + p_{10}}{2},\frac{p_{01} + p_{10}}{2}, \frac{p_{00} + p_{11}}{2}  \right) \leq \mathscr{C}(f_\la)\big|_{\left(\frac{p_{00} + p_{11}}{2},\frac{p_{01} + p_{10}}{2},\frac{p_{01} + p_{10}}{2}, \frac{p_{00} + p_{11}}{2}  \right)}.
\]
This implies that $\mathscr{C}(g)$
\[
\mathscr{C}(g_\lambda)\big|_{\rho} = \mathscr{C}(f_\lambda)\big|_{\left(\frac{1-\rho}{2},\frac{\rho}{2},\frac{\rho}{2},\frac{1-\rho}{2} \right)},
\]
which completes the proof.


\subsection{Proof of Theorem \ref{prop2jm2s}}
\label{proofprop2jm2s}

Let  $M_\rmr\in[1:2^{2C_0}]$ be the message from relay to the receiver ($M_\rmr$ is a function of $Y_\rmr^n$) and define
\begin{align*}
V_i&=Y_{1i+1}^{n},\quad
    W_i = (M_\rmr,Y_{\mathrm{r}}^{i-1}), \quad T_i=(J^{i-1}, J_{i+1}^n),
\end{align*}
and
$V=(Q,V_Q)$, $W=(Q,W_Q)$, $T=(Q,T_Q)$ $X=X_Q, Y_1=Y_{1Q}, Y_{\mathrm{r}}=Y_{\mathrm{r}Q}$ for a time-sharing random variable $Q\stackrel{(d)}{=} \textrm{Uniform}[1:n]$.

 Observe that $T\rightarrow X\rightarrow (J,Y_\rmr,Y_1)$ forms a Markov chain. The  assumption $X_i\rightarrow J_i\rightarrow Y_{\rmr i}$ implies that
 \[I(Y_{\rmr}^{n};X_i,J_i,Y_{1i}|J^{i-1}, J_{i+1}^n,Y_{\rmr i})=0,\]
implying that
 \[I(M_\rmr,Y_{\rmr}^{i-1};X_i,J_i,Y_{1i}|J^{i-1}, J_{i+1}^n,Y_{\rmr i})=0.\]
This yields the Markov chain $W\rightarrow (T,Y_\rmr)\rightarrow (X,J,Y_1)$.

Now consider
\begin{align*}
 I(X^n;M_\rmr, Y_1^n)&\leq 
I(X^n;J^n)+
I(X^n;Y_{\mathrm{r}}^n|J^n)+I(X^n; Y_1^n|M_\rmr,J^n)-I(X^n;Y_{\mathrm{r}}^n|M_\rmr,J^n)
\\&\nonumber\leq \sum_iI(X_i;J_i)+\sum_{i}I(X^n;Y_{1i}|Y_{\mathrm{r}}^{i-1},Y_{1 i+1}^n,M_\rmr,J^n)-\sum_{i}I(X^n;Y_{\mathrm{r}i}|Y_{\mathrm{r}}^{i-1},Y_{1 i+1}^n,M_\rmr,J^n)
\\&\nonumber\leq \sum_iI(X_i;J_i)+\sum_{i}I(X_i;Y_{1i}|Y_{\mathrm{r}}^{i-1},Y_{1 i+1}^n,M_\rmr,J^n)-\sum_{i}I(X_i;Y_{\mathrm{r}i}|Y_{\mathrm{r}}^{i-1},Y_{1 i+1}^n,M_\rmr,J^n)
\\&\qquad+
\sum_{i}I(X^n;Y_{1i}|X_i,Y_{\rmr i},Y_{\mathrm{r}}^{i-1},Y_{1 i+1}^n,M_\rmr,J^n)
\\&\nonumber\leq \sum_iI(X_i;J_i)+\sum_{i}I(X_i;Y_{1i}|Y_{\mathrm{r}}^{i-1},Y_{1 i+1}^n,M_\rmr,J^n)-\sum_{i}I(X_i;Y_{\mathrm{r}i}|Y_{\mathrm{r}}^{i-1},Y_{1 i+1}^n,M_\rmr,J^n)
\\&\leq n\left(I(X;J)+I(X;Y_1|V,W,J,T)-I(X;Y_\rmr|V,W,J,T)\right).
\end{align*}
Next, observe that
\begin{align*}
   n\left[ I(V,W;Y_\rmr|J,T)-I(V,W;Y_1|J,T)\right]
    &=\sum_i I(M_\rmr,Y_{\mathrm{r}}^{i-1},Y_{1i+1}^n;Y_{\rmr i}|J^n)-\sum_i I(M_\rmr,Y_{\mathrm{r}}^{i-1},Y_{1i+1}^n;Y_{1i}|J^n)
    \\
    &=\sum_i I(M_\rmr;Y_{\rmr i}|Y_{\mathrm{r}}^{i-1},Y_{1i+1}^n,J^n)-\sum_i I(M_\rmr,Y_{\mathrm{r}}^{i-1},Y_{1i+1}^n;Y_{1i}|J^n)
    \\&\leq\sum_i I(M_\rmr;Y_{\rmr i}|Y_{\mathrm{r}}^{i-1},Y_{1i+1}^n,J^n)-\sum_i I(M_\rmr;Y_{1i}|Y_{\mathrm{r}}^{i-1},Y_{1i+1}^n,J^n)
    \\
    &=
    \sum_i I(M_\rmr;Y_{\rmr i}|Y_{\mathrm{r}}^{i-1},Y_{1i+1}^n,J^n)-\sum_i I(M_\rmr;Y_{1i}|Y_{\mathrm{r}}^{i-1},Y_{1i+1}^n,J^n)
    \\&=  I(M_\rmr;Y_{\rmr}^n|J^n)-I(M_\rmr;Y_{1}^n|J^n)\\
    & \leq \sum_i I(M_\rmr,Y_{\mathrm{r}}^{i-1};Y_{\mathrm{r}i}|J^n)\\
    & = n I(W;Y_{\mathrm{r}}|J,T).
\end{align*}
On the other hand,
\[n I(W;Y_{\mathrm{r}}|J,T)
= \sum_i I(M_\rmr,Y_{\mathrm{r}}^{i-1};Y_{\mathrm{r}i}|J^n)=I(M_\rmr;Y_{\mathrm{r}}^n|J^n)=  H(M_\rmr) - I(M_\rmr;J^n) \leq nC_0 - I(M_\rmr;J^n).
\]
Thus, 
\[
 I(V,W;Y_\rmr|J,T)-I(V,W;Y_1|J,T) \leq I(W;Y_\rmr|J,T)\leq C_0.
 \]

The cardinality bounds on the auxiliary random variables come from the standard Caratheodory-Bunt \cite{bun34} arguments and is omitted.

This completes the proof.

%

\subsection{Proof of Proposition \ref{lemmaTU}}
\label{prooflemmaTU}

Let rate $R$ be the bound given in Corollary \ref{prop2}.
Observe that
\begin{align}
R&\leq \nonumber I(X;Y_1|T,W,V)-I(X;Y_{\mathrm{r}}|T,W,V)
\\&\leq\nonumber I(X;Y_1|T,W,V,Y_{\mathrm{r}})
\\&\leq I(X;Y_1|Y_{\mathrm{r}},T).\label{cor2thm10eq1}
\end{align}
Next, using the constraint in Corollary \ref{prop2},
\begin{align}
R&\leq\nonumber I(X;Y_1|T,W,V)-I(X;Y_{\mathrm{r}}|T,W,V)
\\&\leq\nonumber I(X;Y_1|T,W,V)-I(X;Y_{\mathrm{r}}|T,W,V)+I(W;Y_{\mathrm{r}}|T)-
I(V,W;Y_\rmr|T)+I(V,W;Y_1|T) 
\\&= \nonumber
I(W,X;Y_1|T)+
I(V;Y_1|W,X,T)-I(X,V;Y_{\mathrm{r}}|W,T)\nonumber
\\&\leq
I(W,X;Y_1|T)+
I(V;Y_1|W,X,T)-I(V;Y_{\mathrm{r}}|W,X,T)\nonumber
\\&=
I(W,X;Y_1|T)+
I(V;Y_1|W,X,T)-I(V;Y_1,Y_{\mathrm{r}}|W,X,T)\label{cor2thm10eq1.5}
\\&\leq
I(W,X;Y_1|T),\label{cor2thm10eq2}
\end{align}
 where \eqref{cor2thm10eq1.5} follows from the Markov structure on the auxiliary random variable $V$. Therefore the bound given in Corollary \ref{prop2} is less than or equal to the bound in Theorem \ref{thm:TU08}. 

Next, let $R$ be the bound given in  Theorem \ref{thm:TU08}. We have 
\[
R=\min\{I(W,X;Y_1|T),I(X;Y_1|Y_\rmr,T)\}
\] 
for some $p(x,t)p(w|t,y_\rmr)$ satisfying $I(W,T;Y_\rmr)\leq C_0$. 
Assume that rate $R$ is achieved by the bound in Corollary \ref{prop2}. Consider the following two cases:

\noindent{\textit{Case 1}:} $R= I(X;Y_1|Y_\rmr,T)$. Then the chain of inequalities leading to \eqref{cor2thm10eq1} must be all equality, implying that 
\[I(X;Y_\rmr|W,V,Y_1,T)=I(V,W;Y_1|Y_\rmr,T)=0.\] Since the channel $p(y_1|x)$ is generic, it follows from Lemma \ref{lem:genind} that $I(V,W;Y_1|Y_{\mathrm{r}},T)=0$ implies that $I(V,W;X|Y_{\mathrm{r}},T)=0$. Therefore, $(V,W) \to (Y_{\mathrm{r}},T) \to X$ form a Markov chain. 
Since $Y_\rmr$ is independent of $(X,T)$, we deduce that
$I(V,W,Y_{\mathrm{r}};X|T)=0$. Therefore,
$I(X;Y_{\mathrm{r}}|W,V,T)=0$. Hence,
\begin{align}
    R\leq I(X;Y_1|W,V,T)-I(X;Y_{\mathrm{r}}|W,V,T)=I(X;Y_1|W,V,T).\label{eqn:A1}
\end{align}
Since $(V,W) \to (Y_{\mathrm{r}},T) \to X$ form a Markov chain, from the constraint
\[
I(V,W;Y_{\mathrm{r}}|T)- I(V,W;Y_1|T)  \leq C_0,
\]
 we obtain
\[
I(V,W;Y_{\mathrm{r}}|Y_1,T)  \leq C_0.
\]
From this equation and \eqref{eqn:A1}, it follows that compress-forward with time-sharing is optimal (using the characterization in Proposition 3 of \cite{kim2007coding}). 

\noindent{\textit{Case 2}:} $R=I(W,X;Y_1|T)<I(X;Y_1|Y_\rmr,T)$.  In this case
we must have $I(W,T;Y_\rmr)=C_0$ since if $I(W,T;Y_\rmr)<C_0$ in Theorem \ref{thm:TU08},  time-sharing between $W$ and $Y_\rmr$ would strictly increase $I(W,X;Y_1|T)$.

Then, from the assumption that the rate $R$ is achieved by the bound in Corollary \ref{prop2} we obtain that the chain of inequalities leading to \eqref{cor2thm10eq2} must be all equality. Thus, $I(X;Y_\rmr|T,W)=I(V;Y_\rmr|X,T,W,Y_1)=0$ and
\begin{equation}
    I(V,W;Y_\rmr|T)-I(V,W;Y_1|T)= I(W;Y_{\mathrm{r}}|T). \label{constrequality}
\end{equation}

We have $I(V;Y_\rmr|X,T,W,Y_1)=0$ and $I(V;Y_1|X,T,W,Y_\rmr)=0$. From the assumption of Proposition \ref{lemmaTU} and Lemma \ref{le:dbma}, we have
$I(V;Y_\rmr,Y_1|X,T,W)=0$, and since $I(X;Y_\rmr|W,T)=0$, we have
\[I(V,X;Y_\rmr|T,W)=0
\]
Therefore,
$I(V;Y_\rmr|W,T)=0$. 
Hence, from \eqref{constrequality},  it follows that 
$I(V,W;Y_1|T)=0$, which yields
\[
R\leq I(X;Y_1|W,V,T)-I(X;Y_{\mathrm{r}}|W,V,T)\leq I(X;Y_1|W,T).
\]
From the Markov structure on $W$, $I(W;Y_1|Y_\rmr,T)=0$, therefore
$I(W;Y_\rmr|Y_1,T)\leq I(W;Y_\rmr|T)\leq C_0$. This implies that the rate $R$ is achieved by compress-forward  with time-sharing (using the characterization in Proposition 3 of \cite{kim2007coding}). This completes the proof.


\subsection{Proof of Theorem \ref{comparisonBD}}
\label{proofcomparisonBD}
We first compute the bound given in Proposition \ref{prop:mainprirel} which is 
\[
R\leq I(X;Y_1|V)-I(X;Y_\rmr|V)
\]
for some $p(x)p(v|x,y_\rmr)$ such that
\begin{align}
I(V;Y_\rmr)-I(V;Y_1)\leq C_0.\label{c0constVYr}
\end{align} 
Lemma \ref{le:Gausuffiorth} in Appendix \ref{appndxEval} shows that it suffices to consider jointly Gaussian random variables in evaluating this bound. 
Let the covariance of $X,Y_\rmr$ given ${V}$ be
\[
K_{X,Y_\rmr|{V}}=
\begin{bmatrix} K_1 & \rho \sqrt{K_1K_2}\\
\rho \sqrt{K_1K_2}& K_2
\end{bmatrix}\preceq \begin{bmatrix} P & 0\\
0& N_\rmr
\end{bmatrix}
\]
for some $0 \leq K_1\leq P$ and $0 \leq K_2\leq N_{\rmr}$ and $\rho\in[-1,1]$ such that
\begin{align}
(P-K_1)(N_{\rmr}-K_2)\geq \rho^2K_1K_2.
\end{align}
Then, the upper bound becomes the maximum of
\begin{align}
    \frac12 \log\left(1+\frac{(\sqrt{K_1}+\sqrt{K_2}\rho)^2}{ K_2(1-\rho^2)+N_1}\right)+\frac12\log(1-\rho^2)
\label{eq:boundold}
\end{align}
subject to
\begin{align}
\frac12\log( N_{\rmr})-\frac12\log(K_2)-\frac12\log(P+N_\rmr+N_1)+\frac12 \log(K_1+K_2+2\rho\sqrt{K_1K_2}+N_1)\leq C_0,
\label{C0const}
\end{align}
$0 \leq K_1\leq P$ and $0 \leq K_2\leq N_{\rmr}$, $\rho\in[-1,1]$, and
\begin{align}
(P-K_1)(N_{\rmr}-K_2)\geq \rho^2K_1K_2.\label{consttoshoweq}
\end{align}

We argue that for any maximizing $K_1, K_2, \rho$, equations \eqref{C0const} and \eqref{consttoshoweq} must hold with equality. Regarding equality in \eqref{C0const}, any maximizing auxiliary random variable $V$ in \eqref{c0constVYr} must satisfy $I(V;Y_\rmr)-I(V;Y_1)=C_0$, otherwise, time-sharing between $V$ and $Y_\rmr$ would strictly improve the upper bound on the rate $R$. We next show that \eqref{consttoshoweq} must hold with equality. 

First consider the case of $\rho\geq 0$. Since constraint \eqref{C0const} holds with equality, we seek to maximize
\begin{align*}
    C_0 - \frac12\log( N_{\rmr})+\frac12\log\left(\frac{K_2(1-\rho^2)}{K_2(1-\rho^2)+N_1}\right)+\frac12\log(P+N_\rmr+N_1) 
\end{align*}
subject to $0 \leq K_1\leq P$ and $0 \leq K_2\leq N_{\rmr}$, $\rho\in[0,1]$ and
\begin{align*}
  &\frac12\log( N_{\rmr})-\frac12\log(K_2)-\frac12\log(P+N_\rmr+N_1)+\frac12 \log(K_1+K_2+2\rho\sqrt{K_1K_2}+N_1)= C_0,\\
  & (P-K_1)(N_{\rmr}-K_2)\geq \rho^2K_1K_2.
\end{align*}
Setting $K_3=K_2(1-\rho^2)$ and $K_4=K_2 \rho^2$, this can be rephrased as maximizing
\begin{align*}
    C_0 - \frac12\log( N_{\rmr})+\frac12\log\left(\frac{K_3}{K_3+N_1}\right)+\frac12\log(P+N_\rmr+N_1) 
\end{align*}
subject to $0 \leq K_1\leq P$ and $K_3,K_4\geq 0$, $0 \leq K_3+K_4\leq N_{\rmr}$ and
\begin{align}
  &\frac12\log( N_{\rmr})-\frac12\log(K_3+K_4)-\frac12\log(P+N_\rmr+N_1)+\frac12 \log((\sqrt{K_1}+\sqrt{K_4})^2+K_3+N_1)= C_0,\label{C0constraintc1}\\
  & (P-K_1)(N_{\rmr}-K_3-K_4)\geq K_1K_4.\label{constraintc2}
\end{align}
The constraint 
$(P-K_1)(N_{\rmr}-K_3-K_4)\geq K_1K_4$
gives the upper bound on $K_1$,
\[
K_1\leq P\left(1-\frac{K_4}{N_\rmr-K_3}\right).
\]
Assume that the above bound is strict. 
Let us fix $K_4$, and increase $K_3$ and $K_1$ simultaneously such that 
\[
\frac{(\sqrt{K_1}+\sqrt{K_4})^2+K_3+N_1}{K_3+K_4}
\]
is preserved. This strictly increases the overall expression because $K_3/(K_3+N_1)$ is strictly increasing in $K_3$, and \eqref{C0constraintc1} is preserved. This shows that  \eqref{constraintc2} must hold with equality for any maximizer and completes the argument in the case of $\rho\geq 0$. 

Next, consider the case of $\rho < 0$. We show that no maximizer can have $\rho<0$. Assume otherwise and consider an optimizer $(K_1, K_2, \rho)$, where $\rho<0$.  If $\sqrt{K_1} < 2|\rho|\sqrt{K_2}$, the expression in \eqref{eq:boundold} is non-positive and cannot be the maximizer. To see this, observe that if $\sqrt{K_1} < 2|\rho|\sqrt{K_2}$ then \eqref{eq:boundold} is less than or equal to
\begin{align}
    \frac12 \log\left(1+\frac{K_2\rho^2}{ K_2(1-\rho^2)+N_1}\right)+\frac12\log(1-\rho^2)=\frac12 \log\left(\frac{K_2(1-\rho^2)+N_1(1-\rho^2)}{ K_2(1-\rho^2)+N_1}\right)\leq 0.
\end{align}
 Therefore we can assume that $\sqrt{K_1} \geq 2|\rho|\sqrt{K_2}$.
Replacing $\rho$ by $-\rho$ and $\sqrt{K_1}$ by the smaller $\sqrt{K_1} + 2\rho \sqrt{K_2}$ does not change \eqref{eq:boundold} while still keeps all the constraints satisfied. However, since we have strictly decreased $K_1$,  \eqref{consttoshoweq} will not hold with equality by this transformation. However, the argument for the case of $\rho\geq 0$ shows that \eqref{consttoshoweq} must hold with equality. Therefore, $\rho<0$ cannot hold for any maximizer.

To sum this up, for any maximizing $K_1, K_2, \rho$, equations \eqref{C0const} and \eqref{consttoshoweq} must hold with equality.

Next, consider the bound given in Corollary \ref{prop2},
\begin{align*}
R&\leq I(X;Y_1|T,W,V)-I(X;Y_{\mathrm{r}}|T,W,V)
\end{align*}
 for some $p(t,x)p(y_1|x,y_{\mathrm{r}})p(y_{\mathrm{r}})p(w|t,y_{\mathrm{r}})p(v|t,x,y_{\mathrm{r}},w)$ such that
\begin{align}
  I(V,W,T;Y_{\mathrm{r}})-I(V,W,T;Y_1)+I(T;Y_1)&\leq 
  I(W,T;Y_{\mathrm{r}})\leq C_0.\label{eqC0relax2}
\end{align}

Lemma \ref{le:Gausuffiorth} in Appendix \ref{appndxEval} shows that it suffices to consider jointly Gaussian random variables in evaluating the above bound. We use a proof by contradiction. Setting $\tilde{V}=(V,W,T)$, it follows that the above bound is less than or equal to the bound in Proposition \ref{prop:mainprirel}.
Assume that the bound in Corollary \ref{prop2} matches the one in Proposition \ref{prop:mainprirel}. Since \eqref{c0constVYr} holds with equality for any maximizer of Proposition \ref{prop:mainprirel},  we must have 
$I(T;Y_1)=0$. This implies $I(T;X)=0$. Thus, $I(W,T;X)=0$.
Let 
\[
K_{X,Y_\rmr|V,W,T}=
\begin{bmatrix} K_1 & \rho \sqrt{K_1K_2}\\
\rho \sqrt{K_1K_2}& K_2
\end{bmatrix}\preceq K_{X,Y_\rmr|W,T}=
\begin{bmatrix} P & 0\\
0& N'_\rmr
\end{bmatrix}
\preceq K_{X,Y_\rmr}=
\begin{bmatrix} P & 0\\
0& N_\rmr
\end{bmatrix}
\]
for some $0 \leq K_1\leq P$ and $0 \leq K_2\leq N'_{\rmr}\leq N_\rmr$ and $\rho\in[-1,1]$ such that
\begin{align}
(P-K_1)(N'_{\rmr}-K_2)\geq \rho^2K_1K_2.
\end{align}

Since \eqref{consttoshoweq} holds with equality, we must have $N'_\rmr=N_\rmr$. Therefore, $(W,T)$ is independent of $Y_\rmr$. This contradicts \eqref{eqC0relax2}. This completes the proof.

\section{Conclusion and Final remarks}
We presented new upper bounds on the capacity of several classes of relay channels and showed through several applications that our bounds can strictly improve upon earlier bounds.

One key insight that leads to strict improvements over previous bounds is the appearance of the same auxiliary variables in multiple constraints  with incompatible choice of optimizing auxiliary for each constraint. In particular, in Theorem~\ref{thm:maintheorem}, bounds \eqref{eqnUB1} and \eqref{eqnUB1.5} that strengthen the term $
I(X;Y,Y_{\mathrm{r}}|X_{\mathrm{r}})$ in the cutset bound are maximized when $I(U;Y|X_{\mathrm{r}},Y_{\mathrm{r}})=I(V;Y|X_{\mathrm{r}},Y_{\mathrm{r}})=I(X;Y_{\mathrm{r}}|V,X_{\mathrm{r}},Y)=0$. On the other hand, bound \eqref{eqnUB0} that strengthens the term $
I(X,X_{\mathrm{r}};Y)$ is maximized when $I(V;Y_\rmr|X_\rmr,Y,X)=0$. For generic channels, these constraints cannot hold simultaneously, leading to strict improvement over the cutset bound. This leads to a trade-off in the realization of the auxiliary which is clearly demonstrated in the statement of Proposition \ref{prop:mainprirel}. Similarly, a maximizer for bound \eqref{propeq0.5o} in Proposition~\ref{prop:mainprirel} is $V=Y_\rmr$ while constraint \eqref{propconstraint} prevents $V$ from becoming close to $Y_\rmr$. Observe that  \eqref{propconstraint} is deduced from constraint \eqref{eq:UBcon} on the auxiliary random variables and can tighten the bounds. 

Other key insights that led to the new results are the techniques used in the evaluation of the bounds: either by carefully relaxing the constraints as in Corollary \ref{cor:weakmainOB} or by identifying the optimal auxiliaries in the associated non-convex optimization problems as in Proposition \ref{Proposition5} and Theorem \ref{thm:GSCase}. There are also instances in which properties of the optimal auxiliaries have been identified that make the optimization problem amenable to numerical evaluation by reducing the search dimension as in Theorem \ref{thm:cutsubopt} and Theorem \ref{thmBSCPrimitive}. 

Our results have focused on the class of relay channels without self-interference and subclasses and applications therein. The same techniques can be used to establish upper bounds on the capacity of the general relay channel (e.g., see Remark~\ref{rmknw7} after Theorem \ref{thm:maintheorem}). We do not have concrete examples, however, that motivate such extensions. 

We gave a specific instance for which the addition of an auxiliary receiver helped improve our main bound. However, there are many other ways that auxiliary receivers could be introduced in these bounds. The study of using auxiliary receivers to strictly improve existing bounds is still rather nascent. In particular, we highlight the following challenge in evaluating bounds involving auxiliary receivers.
A bound with an auxiliary receiver is computed as an infimum over all auxiliary receiver realizations; hence every fixed choice of auxiliary receiver yields a valid and computable bound. However, to obtain the best possible bound, the resultant optimization problem takes a max-min-max-min formulation in which the innermost minimum is over the various rate constraints, the next maximum is over the choice of auxiliary random variables, the subsequent minimum is over the choice of auxiliary receivers, and the outermost maximum is over the choice of the input distributions. Making optimization problems of the above form tractable would be an interesting problem to investigate. 

\section*{Acknowledgements}
Chandra Nair research has been funded by the Research Grants Council of Hong Kong via  GRF grants 14206518 and 14210120.

\bibliographystyle{amsalpha}
\bibliography{mybiblio}

\appendix

\subsection{Some Mathematical Preliminaries}

\begin{lemma}[Double-Markovity, Exercise 16.25 in \cite{csk11},  also see Lemma 6 and Remark 19 in \cite{gon22}]
\label{le:dbma}
Let $(X,Y,Z)$ be random variables such that $X\to Y \to Z$ and $X \to Z \to Y$ are Markov chains. Then there exists functions $f(Y)$ and $g(Z)$ such that $\P(f(Y)=g(Z))=1$ and $X$ is conditionally independent of $(Y,Z)$ given $(f(X),g(Y))$.
\end{lemma}

\begin{lemma}
\label{lem:genind}
Let $p(y|x)$ be a \textit{generic} channel (see Definition \ref{def:gencha})  and assume that $W\to X \to Y$ form a Markov chain.  Then,
$I(W;Y)=0$ implies $I(W;X)=0$. 
\end{lemma}
\begin{proof}
Given $I(W;Y)=0$. Let $w_1, w_2$ be such that $p(w_i) > 0$ for $i=1,2$. Since $I(W;Y)=0$ we have $P(Y=y|W=w_1)= P(Y=y|W=w_2)$. This implies that $\sum_{x} P(X=x|W=w_1) \vec{v}_x = \sum_{x} P(X=x|W=w_2) \vec{v}_x$, where $\vec{v}_x := p(y|x)$. From full row-rank, i.e., the linear independence of $\{\vec{v}_x\}$,  it follows that $P(X=x|W=w_1) = P(X=x|W=w_2)$ for every $x$. Since this holds for any pair $w_1, w_2$ such that $p(w_i) > 0$, we obtain $I(W;X)=0$, i.e., $W$ is independent of $X$. 
\end{proof}

\begin{lemma}\label{lmm1} Let $Y=X+Z$ be an AWGN channel. Assume that $W\rightarrow X\rightarrow Y$ form a Markov chain. Then, we have
\begin{enumerate}[$(i)$]
\item $I(W;Y)=0$ implies that $I(W;X)=0$.
\item $I(W;X|Y)=0$ (or equivalently $I(W;Y)=I(W;X)$) implies that $I(W;X)=0$.
\item $I(X;Y|W)=0$ implies that $X$ is a function of $W$.
\end{enumerate}
\end{lemma}
\begin{proof}
\noindent{$(i)$:} This is a well-known result and follows from the non-vanishing property of the characteristic function  of a Gaussian distribution.

\noindent{$(ii)$:} This follows from Lemma \ref{le:dbma} along with the observation that if $f(X)=g(X+Z)$ with probability one, then both functions $f(X),g(X+Z)$ have to be constant with probability one.

\noindent{$(iii)$:} Note that $I(X;X+Z)=0$ implies that the characteristic function of $X$ satisfies the equation $\Phi_X(t_1+t_2) = \Phi_X(t_1)\Phi_X(t_2) $, whose only solution (in the space of characteristic functions) is $\Phi_X(t)=e^{itc}$, implying that $X=c$ with probability one. Returning to $I(X;X+Z|W)=0$, we have that conditioned on $W$, $X$ is a constant; implying that $X$ is a function of $W$ as required.
\end{proof}


\subsection{Equivalence of capacity regions for relay setting with orthogonal components}\label{AppendixNew}
In the orthogonal component setting, we assume that $X=(X_a,X_b)$, $Y_\rmr=(Y_{\rmr a}, Y_{\rmr b})$, $X_\rmr=(X_{\rmr a},X_{\rmr b})$, and $Y=(Y_a,Y_b)$. The joint probability distribution decomposes as 
\[p(y_a,y_b,y_{\rmr a},y_{\rmr b}|x_a,x_b,x_{\rmr a},x_{\rmr b})= p(y_a,y_{\rmr a}|x_a,x_{\rmr a})p(y_{\rmr b}|x_b)p(y_b|x_{\rmr b}).\] Here the channels $p(y_{\rmr b}|x_b)$ and $p(y_b|x_{\rmr b})$ are considered orthogonal components from the sender to the relay and the relay to the destination respectively.

We consider the following three scenarios:

\noindent \textit{Scenario A -- strictly causal}: This is the standard scenario where the relay's encoded symbols satisfy: $x_{\rmr ai}(y_{\rmr a}^{i-1}, y_{\rmr b}^{i-1})$ and $x_{\rmr bi}(y_{\rmr a}^{i-1}, y_{\rmr b}^{i-1})$, i.e., it can only depend on the past received symbols. 

\noindent \textit{Scenario B -- non-causal}: In this scenario the relay's encoded symbols satisfy: $x_{\rmr ai}(y_{\rmr a}^{i-1}, y_{\rmr b}^{n})$ and $x_{\rmr bi}(y_{\rmr a}^{n}, y_{\rmr b}^{n})$.

\noindent \textit{Scenario C -- bit-pipe}: In this scenario the channel $p(y_{\rmr b}|x_b)$ is replaced by a bit-pipe with the same capacity, say $C_1$. Similarly the channel $p(y_b|x_{\rmr b})$  is replaced by a bit-pipe with the same capacity, say $C_0$. Further the transmitted sequence $X_b^n$ is replaced by $J_s(M) \in [1:\lfloor 2^{nC_1} \rfloor]$, and $X_{\rmr b}^n$ is replaced by $J_\rmr(Y_{\rmr a}^n,Y_{\rmr b}^n) \in [1:\lfloor 2^{nC_0} \rfloor]$.

The equivalence of capacity for the above three settings are well known, e.g.\ see \cite[Sec. 16.7.3]{elk11}, \cite{kim2007coding}. An outline of proof is  provided just for completeness here. 

\begin{lemma}\label{newlemma1}
The capacities of the relay settings in the above three scenarios are equal.
\end{lemma}
\begin{proof}
We will first outline the equivalence between Scenario A and B. Clearly, we only need to show the non-trivial direction, i.e. the capacity of Scenario A is at least as large as that in Scenario B. Let $C_B$ be the capacity in Scenario B, and let $R < C_B$ be any non-negative rate. Therefore there exists a sequence of $(nR,\eps_n)$-codes for Scenario B, with $\eps_n \to 0$ as $n$ tends to infinity. 

Let $X_a^n(m),X_b^n(m)$, $m \in [1:2{nR}]$ be the encoding mapping for the $(nR,\eps_n)$-code. Let $K > 1$ be a positive integer. Construct a $(n(K-1)R,(K-1)\eps_n)$-code with block length $n(K+1)$ for Scenario A as follows:  we divide block of length $n(K+1)$ into $K+1$ sub-blocks of length $n$ each. In sub-block $i$, $1 \leq j \leq K$, the sender transmits $(X_a^n(m_{j-1}), X_b^n(m_j))$ according to the encoding mapping for Scenario B, where $m_j \in [1:2^{nR}]$, $1\leq j \leq K-1$, and $m_0=m_K=1$. Now we can exactly simulate the encoding function $x_{\rmr ai}(y_{\rmr a}^{i-1},y_{\rmr b}^n)$ in Scenario B in sub-blocks $2 \leq j \leq K$, by using $y_{\rmr b}^n$ from the previous sub-block. Further we can exactly simulate the encoding function $x_{\rmr bi}(y_{\rmr a}^{n},y_{\rmr b}^n)$ in Scenario B in sub-blocks $3 \leq j \leq K+1$, b using $y_{\rmr a}^n$ from the previous sub-block and $y_{\rmr b}^n$ from the second previous sub-block. Thus for sub-blocks $2 \leq j \leq K$, the received sequence $Y_a^n$ in Scenario A is exactly same as that in Scenario B for the preceding sub-block, while  for sub-blocks $3 \leq j \leq K+1$, the received sequence $Y_b^n$ in Scenario A is exactly same as that in Scenario B for the second previous sub-block. Using this observation, a decoding error in Scenario $A$ occurs only if a decoding error occurred in Scenario B in one of the sub-blocks $1 \leq j \leq K-1$, and hence can be upper bounded (using union bound) by $(K-1)\eps_n$. Note that this construction implies that a rate of $\frac{(K-1)R}{K+1}$ is achievable for Scenario A for any $K > 1$. Letting $K \to \infty$, we see that rate $C_A \geq R$, where $C_A$ is the capacity region for scenario $A$. Taking the supremum over $R$ we see that $C_B \leq C_A$, establishing the non-trivial direction.

Equivalence between Scenario B and C follows from a similar block coding idea. From Shannon's channel coding theorem, we can simulate a bit-pipe from an arbitrary noisy channel of the same capacity. More generally, any two channels of the same capacity can be simulated from one another in the presence of infinite shared randomness \cite{bennett2002entanglement}. Since the addition of independent shared randomness between the nodes in a communication network does not change its capacity, we get the desired equivalence.

\end{proof}


\subsection{An equivalent form of the bound in Theorem \ref{thm:maintheorem} without an equality constraint }\label{AppendixNew2}

An equivalent form of the bound in Theorem \ref{thm:maintheorem} without the equality constraint is given in the following proposition.

\begin{proposition}\label{prop-no-equality} Any achievable rate $R$ for a discrete memoryless relay channel without self-interference $p(y_\mathrm{r}|x)p(y|x,x_\mathrm{r},y_\mathrm{r})$ must satisfy the following inequalities 
\begin{align}
    R&\leq I(X;Y,Y_{\mathrm{r}}|X_{\mathrm{r}})-I(U;Y|X_{\mathrm{r}},Y_{\mathrm{r}}), \label{eqnUB1moda}\\
    &{
    =I(X;Y|X_\rmr,U)+I(U;Y_\rmr|X_\rmr)+I(X;Y_\rmr|X_\rmr,U,Y)
    }\label{eqnUB1.1moda}
    \\
R&\leq I(X;Y,Y_{\mathrm{r}}|X_{\mathrm{r}})-I(V;Y|X_{\mathrm{r}},Y_{\mathrm{r}})-I(X;Y_{\mathrm{r}}|V,X_{\mathrm{r}},Y)\label{eqnUB1.5moda}\\
 &=  I(X;Y_{\mathrm{r}}|X_{\mathrm{r}})+I(X;Y|V,X_{\mathrm{r}})-I(X;Y_{\mathrm{r}}|V,X_{\mathrm{r}})\label{eqnUB2moda}\\
 &=I(X;Y,V|X_\rmr)-I(V;X|X_\rmr,Y_\rmr),\label{eqnUB2m1moda} \\
R&\leq I(X,X_{\mathrm{r}};Y)-I(V;Y_{\mathrm{r}}|X_{\mathrm{r}},X,Y),\label{eqnUB0moda}\\
R & \leq I(X;Y_{\mathrm{r}}|X_{\mathrm{r}})+I(X,X_{\mathrm{r}};Y|U)-I(X,X_{\mathrm{r}};Y_{\mathrm{r}}|U)-I(V;Y_{\mathrm{r}}|X,X_{\mathrm{r}},Y)\label{eqnUB2adda}\\
R & \leq I(U;Y)+I(X,X_\rmr;Y_{\mathrm{r}}|U)
     +
     I(X;Y|V,X_{\mathrm{r}})-I(X;Y_{\mathrm{r}}|V,X_{\mathrm{r}}) \label{eqnUB0add}
 \end{align}
for some $p(u,x,x_{\mathrm{r}})p(y,y_{\mathrm{r}}|x,x_{\mathrm{r}})p(v|x,x_{\mathrm{r}},y_{\mathrm{r}})$. Moreover, it suffices to consider 
 $|\Vc| \leq |\Xc| |\Xc_{\mathrm{r}}||\Yc_{\mathrm{r}}|+2$ and
$|\Uc| \leq |\Xc| |\Xc_{\mathrm{r}}|+2$.
\end{proposition}

\begin{proof}
We first show that the upper bound in Proposition \ref{prop-no-equality} is greater than or equal to the upper bound in Theorem \ref{thm:maintheorem}. To show this, consider a rate $R$ satisfying the inequalities in Theorem \ref{thm:maintheorem} for some  $p(u,x,x_{\mathrm{r}})p(y,y_{\mathrm{r}}|x,x_{\mathrm{r}})p(v|x,x_{\mathrm{r}},y_{\mathrm{r}})$. Then, using the equality constraint \eqref{eq:UBcon}, we have
\begin{align*}
    R&\leq
    I(X;Y_{\mathrm{r}}|X_{\mathrm{r}})+I(X;Y|V,X_{\mathrm{r}})-I(X;Y_{\mathrm{r}}|V,X_{\mathrm{r}})
    \\&=
    I(X;Y_{\mathrm{r}}|X_{\mathrm{r}})+I(V,X_{\mathrm{r}},X;Y)-I(V,X_{\mathrm{r}},X;Y_{\mathrm{r}})-I(V,X_{\mathrm{r}};Y)+I(V,X_{\mathrm{r}};Y_{\mathrm{r}})
    \\&
    =
    I(X;Y_{\mathrm{r}}|X_{\mathrm{r}})+I(V,X_{\mathrm{r}},X;Y)-I(V,X_{\mathrm{r}},X;Y_{\mathrm{r}})+I(U;Y_{\mathrm{r}})-I(U;Y)
    \\&
    =
    I(X;Y_{\mathrm{r}}|X_{\mathrm{r}})+I(X,X_{\mathrm{r}};Y|U)-I(X,X_{\mathrm{r}};Y_{\mathrm{r}}|U)-I(V;Y_{\mathrm{r}}|X,X_{\mathrm{r}},Y),
\end{align*}
which yields \eqref{eqnUB2adda}. We also have
\begin{align*}
    R&\leq
     I(X,X_{\mathrm{r}};Y)-I(V;Y_{\mathrm{r}}|X_{\mathrm{r}},X,Y)
    \\&=
     I(X,X_{\mathrm{r}};Y)+I(V;Y|X_{\mathrm{r}},X)-I(V;Y_{\mathrm{r}}|X_{\mathrm{r}},X)
     \\&=
     I(V,X_{\mathrm{r}},X;Y)-I(V,X_{\mathrm{r}},X;Y_{\mathrm{r}})+I(X,X_\rmr;Y_\rmr)
      \\&=
     I(V,X_{\mathrm{r}};Y)-I(V,X_{\mathrm{r}};Y_{\mathrm{r}})
     +
     I(X;Y|V,X_{\mathrm{r}})-I(X;Y_{\mathrm{r}}|V,X_{\mathrm{r}})
     +I(X,X_\rmr;Y_\rmr)
      \\&=
     I(U;Y)-I(U;Y_{\mathrm{r}})
     +
     I(X;Y|V,X_{\mathrm{r}})-I(X;Y_{\mathrm{r}}|V,X_{\mathrm{r}})
     +I(X,X_\rmr;Y_\rmr)
     \\&=
     I(U;Y)+I(X,X_\rmr;Y_{\mathrm{r}}|U)
     +
     I(X;Y|V,X_{\mathrm{r}})-I(X;Y_{\mathrm{r}}|V,X_{\mathrm{r}}),
     \end{align*}
which yields \eqref{eqnUB0add}.

Next, we show that the upper bound in Proposition \ref{prop-no-equality} is less than or equal to the upper bound in Theorem \ref{thm:maintheorem}. Assume that a rate $R$ satisfies the bound in Proposition \ref{prop-no-equality} for some $p(u,x,x_{\mathrm{r}})p(y,y_{\mathrm{r}}|x,x_{\mathrm{r}})p(v|x,x_{\mathrm{r}},y_{\mathrm{r}})$. If for this choice of auxiliary random variables $U$ and $V$, we have 
\[ I(V,X_{\mathrm{r}};Y_{\mathrm{r}})-I(V,X_{\mathrm{r}};Y) = I(U;Y_{\mathrm{r}})-I(U;Y),
\]
then we are done. Otherwise, consider two cases:

\noindent{\textbf{Case 1}:} Suppose that
\begin{align} I(V,X_{\mathrm{r}};Y_{\mathrm{r}})-I(V,X_{\mathrm{r}};Y) < I(U;Y_{\mathrm{r}})-I(U;Y). \label{s--1}
\end{align}

The assumption 
\begin{align}
I(V,X_{\mathrm{r}};Y_{\mathrm{r}})-I(V,X_{\mathrm{r}};Y) < I(U;Y_{\mathrm{r}})-I(U;Y) \label{s0}
\end{align}
implies that \eqref{eqnUB2moda}  strictly  implies \eqref{eqnUB2adda}. Also, \eqref{eqnUB0add} strictly implies \eqref{eqnUB0moda}. Thus after removing the redundant inequalities in the statement of Proposition \ref{prop-no-equality}, we obtain
\begin{align}
    R&\leq I(X;Y,Y_{\mathrm{r}}|X_{\mathrm{r}})-I(U;Y|X_{\mathrm{r}},Y_{\mathrm{r}}), \label{snn1}
    \\
R&\leq I(X;Y_{\mathrm{r}}|X_{\mathrm{r}})+I(X;Y|V,X_{\mathrm{r}})-I(X;Y_{\mathrm{r}}|V,X_{\mathrm{r}}),\label{s1}\\
R & \leq I(U;Y)+I(X,X_\rmr;Y_{\mathrm{r}}|U)
     +
     I(X;Y|V,X_{\mathrm{r}})-I(X;Y_{\mathrm{r}}|V,X_{\mathrm{r}}). \label{s2}
 \end{align}

Observe that replacing $V$ with $Y_\rmr$ increases the right hand sides of \eqref{s1} and \eqref{s2}, since
\[
I(X;Y|V,X_{\mathrm{r}})-I(X;Y_{\mathrm{r}}|V,X_{\mathrm{r}}) \leq 
I(X;Y|V,Y_\rmr,X_{\mathrm{r}})
\leq I(X;Y|Y_\rmr,X_{\mathrm{r}}).
\]
Consider two cases based on whether the inequality \eqref{s--1} is still satisfied after replacing $V$ with $Y_\rmr$:
\begin{itemize}
    \item Suppose that
\begin{align} I(Y_\rmr,X_{\mathrm{r}};Y_{\mathrm{r}})-I(Y_\rmr,X_{\mathrm{r}};Y) \geq I(U;Y_{\mathrm{r}})-I(U;Y).\label{eqn:ssss}\end{align}
Consider a Bernoulli time-sharing random variable $Q\sim \mathrm{B}(\theta)$ independent of previously defined random variables and set $\tilde{V}=(V,Q)$ if $Q=0$ and $\tilde{V}=(Y_\rmr,Q)$ if $Q=1$. From \eqref{eqn:ssss} and \eqref{s0} one can find  $\theta\in[0,1]$ such that
\[
I(\tilde{V},X_{\mathrm{r}};Y_{\mathrm{r}})-I(\tilde{V},X_{\mathrm{r}};Y)= I(U;Y_{\mathrm{r}})-I(U;Y).
\]
One can directly verify that replacing $V$ with $\tilde{V}$ increases the right hand sides of \eqref{s1} and \eqref{s2}. The choices of $U$ and $\tilde{V}$ show that the rate $R$ is  less than or equal to the upper bound in Theorem \ref{thm:maintheorem}, and we would be done in this case.
    \item Suppose that
\begin{align} 
I(Y_\rmr,X_{\mathrm{r}};Y_{\mathrm{r}})-I(Y_\rmr,X_{\mathrm{r}};Y) < I(U;Y_{\mathrm{r}})-I(U;Y).\label{eqn:sssss2} 
\end{align}
Then if we replace $V$ with $Y_\rmr$, the rate $R$ would satisfy the inequalities in Proposition \ref{prop-no-equality}. Thus, without loss of generality we can assume that $V=Y_\rmr$.

Consider a Bernoulli time-sharing random variable $Q\sim \mathrm{B}(\theta)$ independent of previously defined random variables and set $\tilde{U}=(V,Q)$ if $Q=0$ and $\tilde{U}=(X_\rmr,Q)$ if $Q=1$. Observe that when $\theta=1$, $\tilde{U}=X_\rmr$. One can also directly verify that when $\theta=1$,
$$ I(Y_\rmr,X_\rmr;Y_\rmr) - I(Y_\rmr,X_\rmr;Y) \geq   I(\tilde{U};Y_\rmr) - I(\tilde{U};Y). $$
On the other hand, when $\theta=0$ we have $\tilde U=U$ and the above inequality holds in the other direction by \eqref{eqn:sssss2}. We deduce existence of  $\theta\in[0,1]$ such that 
$$ I(Y_\rmr,X_\rmr;Y_\rmr) - I(Y_\rmr,X_\rmr;Y) = I(\tilde{U};Y_\rmr) - I(\tilde{U};Y). $$
Observe that replacing $U$ with $\tilde{U}$ increases the right hand side of \eqref{snn1}. The inequality \eqref{s2} still holds for the rate $R$ after this replacement because right hand side of \eqref{s2} changes to
\begin{align*}&\quad
\theta\left[
    I(X_\rmr;Y) + I(X;Y_{\mathrm{r}}|X_\rmr)
     +
     I(X;Y|V,X_{\mathrm{r}})-I(X;Y_{\mathrm{r}}|V,X_{\mathrm{r}}) \right]
\\&\quad+(1-\theta)\left[
I(U;Y)+I(X,X_\rmr;Y_{\mathrm{r}}|U)
     +
     I(X;Y|V,X_{\mathrm{r}})-I(X;Y_{\mathrm{r}}|V,X_{\mathrm{r}}) 
\right]
\\&\geq 
\theta\left[
    I(X;Y_{\mathrm{r}}|X_{\mathrm{r}})+I(X;Y|V,X_{\mathrm{r}})-I(X;Y_{\mathrm{r}}|V,X_{\mathrm{r}}) \right]
\\&\quad+(1-\theta)\left[
I(U;Y)+I(X,X_\rmr;Y_{\mathrm{r}}|U)
     +
     I(X;Y|V,X_{\mathrm{r}})-I(X;Y_{\mathrm{r}}|V,X_{\mathrm{r}}) 
\right]
\\&\geq R
\end{align*}
where the last step follows from \eqref{s1}. The choices of $\tilde U$ and $V=Y_\rmr$ show that the rate $R$ is  less than or equal to the upper bound in Theorem \ref{thm:maintheorem}, and We would be done in this case.
\end{itemize}

\noindent{\textbf{Case 2}:} Suppose that
\begin{align} I(V,X_{\mathrm{r}};Y_{\mathrm{r}})-I(V,X_{\mathrm{r}};Y) > I(U;Y_{\mathrm{r}})-I(U;Y).\label{s01} \end{align}

In this case \eqref{eqnUB2adda}
  strictly  implies \eqref{eqnUB2moda}. Also, \eqref{eqnUB0moda} strictly implies \eqref{eqnUB0add}. Thus after removing the redundant inequalities in the statement of Proposition \ref{prop-no-equality}, we obtain
\begin{align}
    R&\leq I(X;Y,Y_{\mathrm{r}}|X_{\mathrm{r}})-I(U;Y|X_{\mathrm{r}},Y_{\mathrm{r}}), \label{sn0}
    \\
R&\leq I(X,X_{\mathrm{r}};Y)-I(V;Y_{\mathrm{r}}|X_{\mathrm{r}},X,Y),\label{s7}\\
R & \leq I(X;Y_{\mathrm{r}}|X_{\mathrm{r}})+I(X,X_{\mathrm{r}};Y|U)-I(X,X_{\mathrm{r}};Y_{\mathrm{r}}|U)-I(V;Y_{\mathrm{r}}|X,X_{\mathrm{r}},Y).\label{s8}
 \end{align}
 Consider two cases:
 \begin{itemize}
     \item Suppose that 
     \begin{align} I(U;Y_{\mathrm{r}})-I(U;Y) \geq I(X_{\mathrm{r}};Y_{\mathrm{r}})-I(X_{\mathrm{r}};Y).\label{eqns012}\end{align}
     Consider a Bernoulli time-sharing random variable $Q\sim \mathrm{B}(\theta)$ independent of previously defined random variables and set $\tilde{V}=(V,Q)$ if $Q=0$ and $\tilde{V}=Q$ if $Q=1$. From \eqref{eqns012} and \eqref{s01} one can find  $\theta\in[0,1]$ such that
     \[
I(\tilde{V},X_{\mathrm{r}};Y_{\mathrm{r}})-I(\tilde{V},X_{\mathrm{r}};Y)= I(U;Y_{\mathrm{r}})-I(U;Y).
\]
One can directly verify that replacing $V$ with $\tilde{V}$ increases the right hand sides of 
     \eqref{s7} and \eqref{s8}. The choices of $U$ and $\tilde{V}$ show that the rate $R$ is  less than or equal to the upper bound in Theorem \ref{thm:maintheorem}, and We would be done in this case.
     
     \item Suppose that 
     \begin{align} I(U;Y_{\mathrm{r}})-I(U;Y) < I(X_{\mathrm{r}};Y_{\mathrm{r}})-I(X_{\mathrm{r}};Y).\label{eqns013}\end{align}
     Replacing $V$ with a constant preserves \eqref{s01}. Moreover, this replacement does not decrease the right hand sides of 
     \eqref{s7} and \eqref{s8}. Therefore, without loss of generality we may assume that $V$ is a constant. Next, observe that \eqref{eqns013} implies that
      \[ I(X;Y_{\mathrm{r}}|X_{\mathrm{r}})+I(X,X_{\mathrm{r}};Y|U)-I(X,X_{\mathrm{r}};Y_{\mathrm{r}}|U) < I(X;Y_{\mathrm{r}}|X_{\mathrm{r}})+I(X;Y|X_{\mathrm{r}})-I(X;Y_{\mathrm{r}}|X_{\mathrm{r}}).  \]
Thus, replacing $U$ with $\tilde{U}=X_\rmr$ does not decrease the right hand side of \eqref{s8}. It also does not decrease the right hand side of \eqref{sn0}. Moreover, for the choice of $V=\emptyset$ and $\tilde{U}=X_\rmr$, we obtain the equality
\[ I(V,X_{\mathrm{r}};Y_{\mathrm{r}})-I(V,X_{\mathrm{r}};Y)= I(\tilde U;Y_{\mathrm{r}})-I(\tilde U;Y).
\]
This choice shows that the rate $R$ is  less than or equal to the upper bound in Theorem \ref{thm:maintheorem}, and we would be done in this case.\end{itemize}
 
 This completes the proof.
\end{proof}
\subsection{Evaluation of Corollary \ref{cor:weakmainOB} for the Gaussian relay channel}
\label{appndxB}
\begin{lemma}
\label{le:gauoptcor}
For the evaluation of Corollary \ref{cor:weakmainOB} for the Gaussian relay channel, we can assume that the random variables are jointly Gaussian satisfying the requisite Markov Chains.
\end{lemma}
\begin{proof}
We follow the ideas in \cite{gen14} and further focus only on the key steps that are unique here.
Let us consider a class of problems in which we replace $Y_\rmr$ with $\begin{bmatrix} Y_\rmr \\ \eps Z_\rmr + W  \end{bmatrix}$
and $Y$ by $\begin{bmatrix} Y \\ \eps Z_\rmr + W \end{bmatrix}$, where $W \sim \Nc(0,1)$ is independent of previously defined random variables. Let us call these random variables $Y_{\rmr,\eps}$ and $Y_\eps$, respectively. The main proof step is to show that for the optimization problem described below, Gaussian distributions are the maximizers for any $\eps > 0$. Following this step, one can move to the limit $\eps=0$ (the original problem) by using the power constraints, the additive Gaussian noise model, and other arguments found in the Appendix of \cite{gen14}. These are rather standard in analysis and hence the details are omitted.

Consider the following optimization problem: for $\eps>0$, we wish to compute the supremum of $R_\eps$ satisfying the following:
\begin{align*}
R_\eps&\leq
I(X;Y_{\mathrm{r}}|X_{\mathrm{r}})
-I(X;Y_{\mathrm{r},\eps}|V,X_{\mathrm{r}})+I(X;Y_\eps|V,X_{\mathrm{r}}) - \eps h(Y_\eps|V,X_\rmr),\\
R_\eps&\leq I(X,X_{\mathrm{r}};Y)-h(Y_{\mathrm{r}}|X_{\mathrm{r}},X,Y) + h(Y_{\mathrm{r},\eps}|X_{\mathrm{r}},X,Y,V) \\
& \qquad + \eps \big( I(X;Y_{\mathrm{r}}|X_{\mathrm{r}})
-I(X;Y_{\mathrm{r},\eps}|V,X_{\mathrm{r}})+I(X;Y_\eps|V,X_{\mathrm{r}}) - \eps h(Y_\eps|V,X_\rmr)  \big)
 \end{align*}
for some $P_{X,X_{\mathrm{r}}}P_{Y,Y_{\mathrm{r}}|X,X_{\mathrm{r}}}P_{V|X,X_{\mathrm{r}},Z_{\mathrm{r}}}$ such that $\E(X^2)\le P$, $\E(X_\mathrm{r}^2)\le P$,
\begin{align*}
    h(Y_\eps|V,X_\rmr)-h(Y_{\rmr,\eps}|V,X_\rmr)&\leq \max_{P_{U|X,X_{\mathrm{r}}}}\left[h(Y_\eps|U) - h(Y_{\rmr,\eps}|U) \right].
\end{align*}

Suppose $P^\eps$ achieves the supremum of $R_\eps$ (the existence of maximizer follows from reasonably standard arguments - please refer to the Appendix of \cite{gen14} for an illustration of the ideas involved), and under $P^\eps$ let the values of the right-hand-sides of the constraints be $A,B$. Then by taking the usual doubling followed by rotation we obtain
\begin{align*}
    A &= \frac 12 \big( I(X_+;Y_{\mathrm{r},+}|X_{\mathrm{r},+})
-I(X_+;Y_{\mathrm{r},\eps,+}|V_+,X_{\mathrm{r},+})+I(X_+;Y_{\eps,+}|V_+,X_{\mathrm{r},+}) - \eps h(Y_{\eps,+}|V_+,X_{\mathrm{r},+}),\\
& \quad + I(X_-;Y_{\mathrm{r},-}|X_{\mathrm{r},-})
-I(X_-;Y_{\mathrm{r},\eps,-}|V_-,X_{\mathrm{r},-})+I(X_-;Y_{\eps,-}|V_-,X_{\mathrm{r},-}) - \eps h(Y_{\eps,-}|V_-,X_{\mathrm{r},-})\big)\\
& \quad - \frac 12 \big( I(X_{\mathrm{r},-};Y_{\mathrm{r},+}|X_{\mathrm{r},+}) + {\color{olive}I(X_-;Y_{\mathrm{r},\eps,+}|V_+,X_{\mathrm{r},+},X_+)} - {\color{blue}I(X_-;Y_{\eps,+}|V_+,X_{\rmr +},X_+) } +  I(X_{\mathrm{r},+},Y_{\mathrm{r},+};Y_{\mathrm{r},-}|X_{\mathrm{r},-})    \\
& \qquad + {\color{olive}I(X_+;Y_{\mathrm{r},\eps,-}|V_-,X_{\mathrm{r},-},X_-)} - {\color{blue} I(X_+;Y_{\eps,-}|V_-,X_{\rmr -}, X_-) } + \eps I( Y_{\mathrm{r},\eps,-}; Y_{\eps,+} |V_1,V_2,X_{\mathrm{r},-},X_{\mathrm{r},+})\big),
\end{align*}
where $V_+=(V_1,V_2,X_{\mathrm{r},-},Y_{\mathrm{r},\eps,-}  ), V_-=(V_1,V_2,X_{\mathrm{r},+},Y_{\eps,+}  ) $. 
Observe that
\begin{align*}
& {\color{olive}I(X_-;Y_{\mathrm{r},\eps,+}|V_+,X_{\mathrm{r},+},X_+)} - {\color{blue}I(X_-;Y_{\eps,+}|V_+,X_{\rmr +},X_+) }   
 = I(X_-; Y_{r,+}|V_+,X_{\mathrm{r},+},X_+,\eps Z_{\rmr_+} + W_+).  
\end{align*}
A similar equality also holds for the second pair of terms colored in olive and blue.

Now let $Q$ be a uniformly distributed binary random variable taking values $+$ and $-$ each with probability $0.5$, and define $V_\dagger = (V_Q,Q)$, $X_\dagger= X_Q$ and other  variables similar to that of $X$. Then we obtain 
\begin{align*}
    A &\leq   I(X_\dagger;Y_{\mathrm{r},\dagger}|X_{\mathrm{r},\dagger})
-I(X_\dagger;Y_{\mathrm{r},\eps,\dagger}|V_\dagger,X_{\mathrm{r},\dagger})+I(X_\dagger;Y_{\eps,\dagger}|V_+,X_{\mathrm{r},\dagger}) - \eps h(Y_{\eps,\dagger}|V_\dagger,X_{\mathrm{r},\dagger}),\\
& \quad - \frac 12 \big( I(X_{\mathrm{r},-};Y_{\mathrm{r},+}|X_{\mathrm{r},+})  +  I(X_{\mathrm{r},+},Y_{\mathrm{r},+};Y_{\mathrm{r},-}|X_{\mathrm{r},-})     + \eps I( Y_{\mathrm{r},\eps,-}; Y_{\eps,+} |V_1,V_2,X_{\mathrm{r},-},X_{\mathrm{r},+})\big)
\end{align*}

A similar inequality for B can also be written. Further observe that the dagger variables satisfy the constraint as well. This implies that for $P^\eps$ to be optimal, it is necessary (from the Skitovic-Darmois characterization, see \cite{gen14} for details) that
\[
I(X_{\mathrm{r},+},Y_{\mathrm{r},+};Y_{\mathrm{r},-}|X_{\mathrm{r},-})=0, \quad I( Y_{\mathrm{r},\eps,-}; Y_{\eps,+} |V_1,V_2,X_{\mathrm{r},-},X_{\mathrm{r},+})=0,
\]
implying Gaussianity of the conditional distributions $P_{X|X_\rmr}$ and $P_{X,Y_\rmr|X_\rmr,V}$; and further that the covariance of $P_{X,Y_\rmr|X_\rmr,V}$ does not depend on the conditioned variables. 
\end{proof}

\subsection{Evaluation of Proposition \ref{prop:mainprirel} for the Gaussian product-form relay channel}
\label{appndxC}

\begin{lemma}
\label{le:gauoptcor2}
To evaluate Proposition \ref{prop:mainprirel}  for the Gaussian product-form relay channel, we can assume that the random variables are jointly Gaussian satisfying the requisite Markov Chains.
\label{lem:gausuffcor}
\end{lemma}
\begin{proof}
As in the proof of Lemma \ref{le:gauoptcor} we follow the ideas in \cite{gen14} and the steps are very similar to those of the proof of the previous lemma. Consider a class of problems in which we replace $Y_\rmr$ with $\begin{bmatrix} Y_\rmr \\ \eps Z_\rmr + W  \end{bmatrix}$
and $Y_1$ by $\begin{bmatrix} Y_1 \\ \eps Z_\rmr + W \end{bmatrix}$, where $W \sim \Nc(0,1)$ is independent of previously defined random variables. Let us call these random variables $Y_{\rmr,\eps}$ and $Y_{1,\eps}$, respectively. For $\eps>0$, we wish to compute the supremum of $R_\eps$ satisfying the following:
\begin{align*}
R_\eps&\leq I(X;Y_{\mathrm{r}})+I( X;Y_{1,\eps}|V)-I(X;Y_{\mathrm{r},\eps}| V) - \eps h(Y_{1,\eps}|V)
 \end{align*}
for some $P_{X}P_{Y_1,Y_{\mathrm{r}}|X}P_{V|X,Z_{\mathrm{r}}}$ 
such that
\begin{align*}
 I(V;Y_{\mathrm{r},\eps})-I(V;Y_{1,\eps})&\leq C_0,
\end{align*}
and identical power constraints.
Suppose $P^\eps$ achieves the supremeum of $R_\eps$ (existence justified using routine arguments), and under $P^\eps$ let the value of the right-hand-side be  $A$. Then by taking the usual doubling followed by rotation we obtain
    \begin{align*}
    A &= \frac 12 \Big( I(X_+;Y_{\mathrm{r},+})
-I(X_+;Y_{\mathrm{r},\eps,+}|V_+)+I(X_+;Y_{1,\eps,+}|V_+) - \eps h(Y_{1,\eps,+}|V_+),\\
& \quad + I(X_-;Y_{\mathrm{r},-})
-I(X_-;Y_{\mathrm{r},\eps,-}|V_-)+I(X_-;Y_{1,\eps,-}|V_-) - \eps h(Y_{1,\eps,-}|V_-)\Big)\\
& \quad - \frac 12 \Big(  {\color{olive}I(X_-;Y_{\mathrm{r},\eps,+}|V_+,X_+)} - {\color{blue}I(X_-;Y_{1,\eps,+}|V_+,X_+) } +  I(Y_{\mathrm{r},+};Y_{\mathrm{r},-})    \\
& \qquad + {\color{olive}I(X_+;Y_{\mathrm{r},\eps,-}|V_-,X_-)} - {\color{blue} I(X_+;Y_{1,\eps,-}|V_-, X_-) } + \eps I( Y_{\mathrm{r},\eps,-}; Y_{1,\eps,+} |V_1,V_2)\Big)
\end{align*}
where $V_+=(V_1,V_2,Y_{\mathrm{r},\eps,-}  ), V_-=(V_1,V_2,Y_{1,\eps,+}  ) $. 
Observe that
\begin{align*}
& {\color{olive}I(X_-;Y_{\mathrm{r},\eps,+}|V_+,X_+)} - {\color{blue}I(X_-;Y_{1,\eps,+}|V_+,X_+) }   
 = I(X_-; Y_{r,+}|V_+,X_+,\eps Z_{\rmr_+} + W_+).  
\end{align*}
As before, let $Q$ be a uniformly distributed binary random variable taking values $+$ and $-$ each with probability $0.5$, and define $V_\dagger = (V_Q,Q)$, $X_\dagger= X_Q$ and other  variables similar to that of $X$. Using this we obtain that
\begin{align*}
    A &\leq   I(X_\dagger;Y_{\mathrm{r},\dagger})
-I(X_\dagger;Y_{\mathrm{r},\eps,\dagger}|V_\dagger)+I(X_\dagger;Y_{1,\eps,\dagger}|V_+) - \eps h(Y_{\eps,\dagger}|V_\dagger),\\
& \quad - \frac 12 \Big(  I(Y_{\mathrm{r},+};Y_{\mathrm{r},-})    
+ \eps I( Y_{\mathrm{r},\eps,-}; Y_{\eps,+} |V_1,V_2)\Big).
\end{align*}
Further observe that the dagger variables satisfy 
$$ I(V_\dagger;Y_{\mathrm{r},\eps,\dagger})-I(V_\dagger;Y_{1,\eps,\dagger})\leq C_0$$
as well. Since the dagger variables are also a feasible choice, for $P^\eps$ to be optimal, it is necessary (from the Skitovic-Darmois characterization, see \cite{gen14} for details) that $ I(Y_{\mathrm{r},+};Y_{\mathrm{r},-}) =0 $ and $I( Y_{\mathrm{r},\eps,-}; Y_{\eps,+} |V_1,V_2)=0 $, implying that $P_X$ and $P_{X,Y_\rmr|V}$ are both Gaussians and  that the covariance of the latter does not depend on $V$.
\end{proof}

\begin{lemma}\label{NewLemma}
In evaluating  Proposition 
\ref{prop:mainprirel}
in the space of jointly Gaussian distributions, the covariance of $K_{X,Y_\rmr|V}$ is unique for the maximizing distribution and equals
$$
K_{X,Y_\rmr|{V}}=
\begin{bmatrix} K^*_1 & \rho^* \sqrt{K^*_1K^*_2}\\
\rho^* \sqrt{K^*_1K^*_2}& K^*_2
\end{bmatrix}
$$
where{
\begin{align}K_1^*&=\begin{cases}-N_1\frac{2^{2C_0}(P+N_{\rmr})-(P+N_{\rmr})}{2^{2C_0}(P+N_1)-(P+N_{\rmr})}+P\frac{2^{2C_0}(N_1-N_{\rmr})}{2^{2C_0}(P+N_1)-(P+N_{\rmr})}& S_{12}\geq S_{13}+S_{23}+S_{13}S_{23},\\
    \\
    P\left(1-\frac{P(N_1+P)^2(2^{2C_0}-1)}{
(P+N_{\rmr})(2^{2C_0}-1)((N_1+P)^2-N_1^22^{-2C_0})+(N_{\rmr}-N_1)P^2}\right)&\text{otherwise}.
    \end{cases}
\\
K^*_2&=\frac{(K^*_1+N_1)(P+N_{\rmr})}{(P+N_1)2^{2C_0}}\\
    \rho^*&= \frac{P-\sqrt{(P-K^*_1)(P+N_{\rmr}-K^*_2)}}{\sqrt{K^*_1K^*_2}}.
\end{align}}
\end{lemma}
\begin{proof}
Let
$$
K_{X,Y_\rmr|{V}}=
\begin{bmatrix} K_1 & \rho \sqrt{K_1K_2}\\
\rho \sqrt{K_1K_2}& K_2
\end{bmatrix}\preceq \begin{bmatrix} P & P\\
P& P+N_{\rmr}
\end{bmatrix}
$$
for some $0 \leq K_1\leq P$ and $0 \leq K_2\leq N_{\rmr}+P$ and $\rho\in[-1,1]$ satisfying
\begin{align}(P-K_1)(P+N_{\rmr}-K_2)\geq (P-\rho\sqrt{K_1K_2})^2.\label{eq:rhocon}\end{align}
Then, the bound becomes 
\begin{align*}
    R\leq \frac12 \log\left(\frac{P+ N_{\rmr}}{ N_{\rmr}}\right)
+\frac12\log\left(\frac{K_1+N_1}{N_1}\right)
+\frac12\log\left(1-\rho^2\right)
\end{align*}
subject to
    $$\frac12\log(P+ N_{\rmr})-\frac12\log(K_2)-\frac12\log(P+N_1)+\frac12 \log(K_1+N_1)\leq C_0.$$
The optimizer $\rho$ is non-negative; otherwise if the optimizer $\rho$ is negative, moving to $\rho=0$ strictly increases the expression. This is because if we decrease $\rho^2$ while fixing $K_1$ and $K_2$, the  constraint 
$$(P-K_1)(P+N_{\rmr}-K_2)\geq (P-\rho\sqrt{K_1K_2})^2$$
will still hold. 
The expression we wish to maximize
\begin{align*}
    (K_1+N_1) (1-\rho^2)
\end{align*}
will also increase as we decrease $\rho^2$.
Increasing $K_2$ while decreasing $\rho$ such that $\rho\sqrt{K_2}$ is preserved, shows that either $K_2=P+N_{\rmr}$ or else 
$$(P-K_1)(P+N_{\rmr}-K_2)= (P-\rho\sqrt{K_1K_2})^2.$$
Even in the case of $K_2=P+N_{\rmr}$, the  inequality
\begin{align}
    (P-K_1)(P+N_{\rmr}-K_2)&\geq (P-\rho\sqrt{K_1K_2})^2\label{eqnrho2}
\end{align}
must hold with equality. Thus, any optimal choice for $\rho$ must satisfy $\rho\geq 0$ and
\begin{align}
    \rho= \frac{P-\sqrt{(P-K_1)(P+N_{\rmr}-K_2)}}{\sqrt{K_1K_2}}.\label{eqnrhot1}
\end{align}
Note that the other solution for $\rho$ in \eqref{eqnrho2} is
\begin{equation}
    \rho= \frac{P+\sqrt{(P-K_1)(P+N_{\rmr}-K_2)}}{\sqrt{K_1K_2}}.\label{eqnrhot2}
\end{equation}
However, the optimizer $\rho$ must satisfy \eqref{eqnrhot1}. If instead \eqref{eqnrhot2} holds, reducing $\rho$ to $P/\sqrt{K_1K_2}$ would strictly increase the objective function while continuing to satisfy \eqref{eq:rhocon}.

Next, we can infer that any the maximizing auxiliary random variable $V$ must satisfy $I(V;Y_\rmr)-I(V;Y_1)=C_0$, otherwise time-sharing between $V$ and $Y_\rmr$ would strictly improve the upper bound on the rate $R$. In other words, the maximizing distribution is such that \eqref{propconstraint}
holds with equality. This yields
$$K_2=\frac{(K_1+N_1)(P+N_{\rmr})}{(P+N_1)2^{2C_0}}$$
and
\begin{align*}\rho&= \frac{P-\sqrt{(P-K_1)(P+N_{\rmr}-K_2)}}{\sqrt{K_1K_2}}
\\&
= \frac{P - \sqrt{(P-K_1)(P+N_{\rmr}-(K_1+N_1)\frac{P+N_{\rmr}}{P+N_1}2^{-2C_0})  }}{\sqrt{K_1(K_1+N_1)\frac{P+N_{\rmr}}{P+N_1}2^{-2C_0}}}.
\end{align*}
We wish to maximize $(K_1+N_1)(1-\rho^2)$. 
Letting $\zeta=\frac{P+N_{\rmr}}{P+N_1}\geq 1$, we seek to maximize, subject to $K_1 \in [0, P]$, the following expression:
\begin{align*}
   &\frac{K_1(K_1+N)\zeta2^{-2C_0} - \left(P - \sqrt{(P-K_1)((P+N_1)\zeta-(K_1+N_1)\zeta2^{-2C_0})  }\right)^2}{K_1\zeta 2^{-2C_0} }\\
   & = \frac{-P^2 - P(P+N_1)\zeta + K_1(P+N_1)\zeta + P(K_1+N_1)\zeta2^{-2C_0} + 2P\sqrt{(P-K_1)((P+N_1)\zeta-(K_1+N_1)\zeta2^{-2C_0})  }}{K_1 \zeta2^{-2C_0}}.
\end{align*}
Taking the derivative with respect to $K_1$, and simplifying it we obtain
\begin{align}&
K_1\zeta (-2^{-2C_0}N_1+2^{-2C_0}P+N_1+P)+2\zeta P(2^{-2C_0} N_1-N_1-P)\nonumber \\
&=(2^{-2C_0} N_1\zeta -N_1\zeta -P\zeta -P)\sqrt{(P-K_1)((P+N_1)\zeta -(K_1+N_1)\zeta 2^{-2C_0} )}.\label{eqnDiff}
\end{align}

If we raise both sides to power two, we get a quadratic equation. {This quadratic equation has the following two roots:
\begin{align*}
    K_{1a}&=-N_1\frac{\zeta-2^{-2C_0}\zeta}{1-2^{-2C_0}\zeta}+P\frac{1-\zeta}{1-2^{-2C_0}\zeta},\\
    K_{1b}&=P\left(1-\frac{P(N_1+P)(2^{-2C_0}-1)}{
\zeta(2^{-2C_0}-1)((N_1+P)^2-N_1^22^{-2C_0})+(1-\zeta)2^{-2C_0}P^2}\right).
\end{align*}
One can verify that 
in the case of
$$S_{13}+S_{23}+S_{13}S_{23}\leq S_{12}
$$
only the root $K_{1a}$ is in $[0,P]$ while when
$$S_{13}+S_{23}+S_{13}S_{23}> S_{12}
$$ 
only the root $K_{1b}$ is in $[0,P]$. Let $K_1^*$ denote the unique root that lies in $[0,P]$. The function
\begin{align*}
   & \frac{-P^2 - P(P+N_1)\zeta + K_1(P+N_1)\zeta + P(K_1+N_1)\zeta2^{-2C_0} + 2P\sqrt{(P-K_1)((P+N_1)\zeta-(K_1+N_1)\zeta2^{-2C_0})  }}{K_1 \zeta2^{-2C_0}}.
\end{align*}
is increasing for $K_1\in (0,K_1^*]$ and decreasing for $K_1\in [K_1^*, P]$. Therefore, it reaches its maximum at $K_1=K_1^*$.
}
\end{proof}


\subsection{Evaluation of Corollary \ref{prop2} and Proposition \ref{prop:mainprirel} for the Gaussian  relay  channel  with  orthogonal  receiver  components  and  i.i.d.output}
\label{appndxEval}
{

\begin{lemma}
\label{le:Gausuffiorth}
It suffices to consider jointly Gaussian distributions when evaluating the outer bounds in Corollary \ref{prop2} or Proposition \ref{prop:mainprirel} for the channel described by \eqref{GaussianExampleEq1}, \eqref{GaussianExampleEq2}.
\end{lemma}

\begin{proof}
As with the other arguments, we will first define a perturbed $Y_1$ defined according to
\begin{align} Y_{1,\eps}=\begin{bmatrix} X+Y_\rmr + Z_1 \\ \eps Y_\rmr + Z_2 \end{bmatrix}.\label{defY1eps}\end{align}
Here $Z_2 \sim \Nc(0,1)$ is mutually independent of all other random variables.

Let us consider the argument for Proposition \ref{prop:mainprirel} first. Consider the maximum of
\begin{align}
I(X;Y_{1,\eps}|V)-I(X;Y_\rmr|V)-\delta I(Y_{1,\eps};\delta Y_{1,\eps} + Z_3|V)
 \end{align}
 subject to   $p(x)p(y_{1,\eps},y_{\mathrm{r}}|x)p(v|x,y_{\mathrm{r}})p(z_3)$ and \(
 I( V;Y_{\mathrm{r}})-I( V;Y_{1,\eps})\leq C_0.\)
 Here $Z_3 \sim \Nc(0,I_2)$ is mutually independent of all other random variables.
 The idea is to show that the maximum is attained by a jointly Gaussian distribution, and then take the limit and $\eps, \delta \to 0$. 
 
 Take two independent copies of the maximizer
 $(X_1,Y_{\rmr,1},Y_{1,\eps,1},V_1,Z_{3,1})$ and $(X_2,Y_{\rmr,2},Y_{1,\eps,2},V_2,Z_{3,2})$, 
 and perform two orthogonal rotations to obtain the $+$ and $-$ random variables $(X_+,Y_{\rmr,+},Y_{1,\eps,+},Z_{3,+})$ and $(X_-,Y_{\rmr,-},Y_{1,\eps,-},Z_{3,-})$. Let $V=(V_1,V_2)$. Now observe that
  \begin{align*}
& I(X_+,X_-;Y_{1,\eps,+},Y_{1,\eps,-}|V)-I(X_+,X_-;Y_{\rmr,+},Y_{\rmr,-}|V)
\\
&\quad = I(X_+;Y_{1,\eps,+}|V,Y_{1,\eps,-})-I(X_+;Y_{\rmr,+}|V,Y_{1,\eps,-}) \\
& \qquad + I(X_-;Y_{1,\eps,-}|V,Y_{\rmr,+})-I(X_-;Y_{\rmr,-}|V,Y_{\rmr,+}) \\
& \qquad +I(X_-;Y_{1,\eps,+}|V,Y_{1,\eps,-},X_+)-I(X_-;Y_{\rmr,+}|V,Y_{1,\eps,-},X_+) \\
&\qquad
+I(X_+;Y_{1,\eps,-}|V,Y_{\rmr,+},X_-)-I(X_+;Y_{\rmr,-}|V,Y_{\rmr,+},X_-) 
\\
&\quad = I(X_+;Y_{1,\eps,+}|V,Y_{1,\eps,-})-I(X_+;Y_{\rmr,+}|V,Y_{1,\eps,-}) \\
& \qquad + I(X_-;Y_{1,\eps,-}|V,Y_{\rmr,+})-I(X_-;Y_{\rmr,-}|V,Y_{\rmr,+}) \\
& \qquad - I(X_-;Y_{\rmr,+}|V,Y_{1,\eps,-},Y_{1,\eps,+},X_+) - I(X_+;Y_{\rmr,-}|V,Y_{\rmr,\eps,+},Y_{1,\eps,-},X_-).
 \end{align*}
The last step follows from
\[
I(X_-;Y_{1,\eps,+}|Y_{\rmr,+},V,Y_{1,\eps,-},X_+)=I(X_+;Y_{1,\eps,-}|Y_{\rmr,-},V,Y_{\rmr,+},X_-)=0
\]
which follows from the fact that  $(Z_{1,+}, Z_{1,-})$ and $(Z_{2,+}, Z_{2,-})$ are pairs of independent random variables. 
Next, we also have
\begin{align}
& I(Y_{1,\eps,+},Y_{1,\eps,-};\delta Y_{1,\eps,+} + W_{3+},\delta Y_{1,\eps,-} + Z_{3,-}|V) \nonumber
\\
&\quad =  I(Y_{1,\eps,+};\delta Y_{1,\eps,+} + Z_{3,+}|V,Y_{1,\eps,-}) + I(Y_{1,\eps,-};\delta Y_{1,\eps,-} + Z_{3,-}|V)\nonumber \\
& \qquad +I(Y_{1,\eps,-};\delta Y_{1,\eps,+} + Z_{3,+}|V,\delta Y_{1,\eps,-}+Z_{3,-}) \nonumber
\\
&\quad =  I(Y_{1,\eps,+};\delta Y_{1,\eps,+} + Z_{3,+}|V,Y_{1,\eps,-}) + I(Y_{1,\eps,-};\delta Y_{1,\eps,-} + Z_{3,-}|V,Y_{\rmr,+})\label{eq:vdec}\\
& \qquad +I(Y_{1,\eps,-};\delta Y_{1,\eps,+} + Z_{3,+}|V,\delta Y_{1,\eps,-}+Z_{3,-}) + I(Y_{\rmr,+};\delta Y_{1,\eps,-} + Z_{3,-}|V). \nonumber
 \end{align}
Thus, we obtain
 \begin{align}
& I(X_+,X_-;Y_{1,\eps,+},Y_{1,\eps,-}|V)-I(X_+,X_-;Y_{\rmr,+},Y_{\rmr,-}|V)-\delta I(Y_{1,\eps,+},Y_{1,\eps,-};\delta Y_{1,\eps,+} + Z_{3,+},\delta Y_{1,\eps,-} + Z_{3,-}|V)\nonumber\\
&\quad = I(X_+;Y_{1,\eps,+}|V,Y_{1,\eps,-})-I(X_+;Y_{\rmr,+}|V,Y_{1,\eps,-}) -\delta I(Y_{1,\eps,+};\delta Y_{1,\eps,+} + Z_{3,+}|V,Y_{1,\eps,-})\nonumber\\
& \qquad + I(X_-;Y_{1,\eps,-}|V,Y_{\rmr,+})-I(X_-;Y_{\rmr,-}|V,Y_{\rmr,+}) -\delta I(Y_{1,\eps,-};\delta Y_{1,\eps,-} + Z_{3,-}|V,Y_{\rmr,+})\nonumber\\
& \qquad - I(X_-;Y_{\rmr,+}|V,Y_{1,\eps,-},Y_{1,\eps,+},X_+) - I(X_+;Y_{\rmr,-}|V,Y_{\rmr,\eps,+},Y_{1,\eps,-},X_-)\nonumber\\
& \qquad - \delta \textcolor{blue}{I(Y_{1,\eps,-};\delta Y_{1,\eps,+} + Z_{3,+}|V,\delta Y_{1,\eps,-}+Z_{3,-})} - \delta I(Y_{\rmr,+};\delta Y_{1,\eps,-} + Z_{3,-}|V)\label{eqnbluede}
 \end{align}
 Further observe that
 \begin{align*}
     &I(V;Y_{\rmr,+},Y_{\rmr,-})- I(V;Y_{1,\eps,+},Y_{1,\eps,-})\\
     & \quad = I(V;Y_{\rmr,+})+ I(V,Y_{\rmr,+};Y_{\rmr,-}) - I(V;Y_{1,\eps,-}) - I(V,Y_{1,\eps,-};Y_{1,\eps,+}) + I(Y_{1,\eps,-};Y_{1,\eps,+})\\
     & \quad = I(V,Y_{1,\eps,-};Y_{\rmr,+})+ I(V,Y_{\rmr,+};Y_{\rmr,-}) - I(V,Y_{\rmr,+};Y_{1,\eps,-}) - I(V,Y_{1,\eps,-};Y_{1,\eps,+}) + I(Y_{1,\eps,-};Y_{1,\eps,+})
 \end{align*}
 implying that
 $$ I(V,Y_{1,\eps,-};Y_{\rmr,+})+ I(V,Y_{\rmr,+};Y_{\rmr,-}) - I(V,Y_{\rmr,+};Y_{1,\eps,-}) - I(V,Y_{1,\eps,-};Y_{1,\eps,+}) \leq 2C_0. $$
 We identify $V_+ = (V,Y_{1,\eps,-}) $ and $V_-=(V,Y_{\rmr,+})$.
 Now by taking a uniform binary $Q$ independent, and taking each of $+$ distribution and $-$ distribution with probability $\frac{1}{2}$, we get a new maximizer, which satisfies the constraint and yields a value at least as large as the original maximizer. For the maximum value to be equal we must have certain terms to be zero, in particular the blue term in \eqref{eqnbluede} must be zero, \textit{i.e.,}
 \begin{align}
     I(Y_{1,\eps,-};\delta Y_{1,\eps,+} + Z_{3,+}|V,\delta Y_{1,\eps,-}+Z_{3,-})=0.\label{bterm1}
 \end{align}

From \eqref{bterm1} and $I(\delta Y_{1,\eps,-}+Z_{3,-};\delta Y_{1,\eps,+} + Z_{3,+}|V,Y_{1,\eps,-})=0$, we obtain a double-Markovity condition as defined in Lemma \ref{le:dbma}. We obtain that conditioned on $V$, we have $Y_{1,\eps,-}, Y_{1,\eps,+}$ are independent. This in particular implies that $X_+,Y_{\rmr_+}$ and $X_-,Y_{\rmr_-}$ are independent conditioned on $V$, yielding the desired Gaussianity of $(X,Y_\rmr)$ given $V$. Next we claim that taking $X$ to be Gaussian minimizes $I(V;Y_\rmr) - I(V;Y_1)$. To see this, observe the constraint $h(Y_\rmr)-h(Y_{1,\eps}) - h(Y_\rmr|V)+h(Y_{1,\eps}|V) = I(V;Y_\rmr) - I(V;Y_{1,\eps}) \leq C_0$. Note that, we have established that conditioned on $V$, $(X,Y_\rmr)$ are Gaussians. We are given that $Y_\rmr$ is Gaussian. Therefore taking $X$ to be Gaussian maximizes the term $ h(Y_{1,\eps})$ and thus keeps the constraint satisfied.

Now let us consider the argument for Corollary \ref{prop2}. Using the definition of 
$Y_{1,\eps}$ and $Z_3$ as before (see \eqref{defY1eps}), consider a maximizer of
\begin{align}
    R &\leq I(X;Y_{1,\eps}|V,W,T)-I(X;Y_\rmr|V,W,T)-\delta I(Y_{1,\eps};\delta Y_{1,\eps} + Z_3|V,W,T) -  { \delta I(Y_{\rmr};\delta Y_{\rmr} + Z_3|W,T)}\label{eqntg1}
\end{align}
over $p(t,x)p(y_{1,\eps}|x,y_{\mathrm{r}})p(y_{\mathrm{r}})p(w|t,y_{\mathrm{r}})p(v|t,x,y_{\mathrm{r}},w)p(z_3)$ subject to
\begin{align*}
  I(V,W;Y_\rmr|T)-I(V,W;Y_{1,\eps}|T)&\leq I(W;Y_{\mathrm{r}}|T) \leq  C_0.
\end{align*}

As before, 
take two copies of the maximizer and perform two orthogonal rotations to obtain the $+$ and $-$ random variables. Set $V=(V_1,V_2), W=(W_1,W_2), T=(T_1,T_2)$. We have
\begin{align}
    &I(Y_{\rmr,+},Y_{\rmr,-};\delta Y_{\rmr,+} + Z_{3,+},\delta Y_{\rmr,-} + Z_{3,-}|W,T) \nonumber \\
    &\quad = I(Y_{\rmr,+};\delta Y_{\rmr,+} + Z_{3,+}|W,T) + I(Y_{\rmr,-};\delta Y_{\rmr,-} + Z_{3,-}|W,T,\delta Y_{\rmr,+} + Z_{3,+}) \nonumber \\
    & \quad = I(Y_{\rmr,+};\delta Y_{\rmr,+} + Z_{3,+}|W,T) + I(Y_{\rmr,-};\delta Y_{\rmr,-} + Z_{3,-}|W,T, Y_{\rmr,+}) \nonumber \\
    & \qquad + I(Y_{\rmr,+};\delta Y_{\rmr,-} + Z_{3,-}|\delta Y_{\rmr,+} + Z_{3,+},W,T). \label{eq:wdec}
\end{align}
The exact manipulation for the first constraint above that we did in \eqref{eqnbluede} goes through with  $(V,W,T)$ replacing $V$, and further using \eqref{eq:wdec} we obtain the following
\begin{align}
 & I(X_+,X_-;Y_{1,\eps,+},Y_{1,\eps,-}|V,W,T)-I(X_+,X_-;Y_{\rmr,+},Y_{\rmr,-}|V,W,T) \nonumber\\
 & \quad -\delta I(Y_{1,\eps,+},Y_{1,\eps,-};\delta Y_{1,\eps,+} + Z_{3,+},\delta Y_{1,\eps,-} + Z_{3,-}|V,W,T) - {\delta I(Y_{\rmr,+},Y_{\rmr,-};\delta Y_{\rmr,+} + Z_{3,+},\delta Y_{\rmr,-} + Z_{3,-}|W,T)}\nonumber\\
&\quad = I(X_+;Y_{1,\eps,+}|V,W,T,Y_{1,\eps,-})-I(X_+;Y_{\rmr,+}|V,W,T,Y_{1,\eps,-}) -\delta I(Y_{1,\eps,+};\delta Y_{1,\eps,+} + Z_{3,+}|V,W,T,Y_{1,\eps,-})\nonumber\\
& \qquad - {\delta I(Y_{\rmr,+};\delta Y_{\rmr,+} + Z_{3,+}|W,T)}\nonumber \\
& \qquad + I(X_-;Y_{1,\eps,-}|V,W,T,Y_{\rmr,+})-I(X_-;Y_{\rmr,-}|V,W,T,Y_{\rmr,+}) -\delta I(Y_{1,\eps,-};\delta Y_{1,\eps,-} + Z_{3,-}|V,W,T,Y_{\rmr,+})\nonumber\\
& \qquad - {\delta I(Y_{\rmr,-};\delta Y_{\rmr,-} + Z_{3,-}|W,T, Y_{\rmr,+})}\nonumber \\
& \qquad - I(X_-;Y_{\rmr,+}|V,W,T,Y_{1,\eps,-},Y_{1,\eps,+},X_+) - I(X_+;Y_{\rmr,-}|V,W,T,Y_{\rmr,\eps,+},Y_{1,\eps,-},X_-)\nonumber\\
& \qquad - \delta \textcolor{blue}{I(Y_{1,\eps,-};\delta Y_{1,\eps,+} + Z_{3,+}|V,W,T,\delta Y_{1,\eps,-}+Z_{3,-})} - \delta I(Y_{\rmr,+};\delta Y_{1,\eps,-} + Z_{3,-}|V,W,T)\label{eqnbluede2}  \\
&\qquad - {\delta  I(Y_{\rmr,+};\delta Y_{\rmr,-} + Z_{3,-}|\delta Y_{\rmr,+} + Z_{3,+},W,T).  } \nonumber
\end{align}

Let us identify $V_+ = (V,Y_{1,\eps,-}) $, $V_-=(V,Y_{\rmr,+})$, $W_+=W$, $W_-=(W,Y_{\rmr,+})$, and $T_+=T_-=T$. 
To verify the constraint, we need to show that
\begin{align*}
&I(V_+,W_+;Y_{\rmr,+}|T_+) - I(V_+,W_+;Y_{1,\eps,+}|T_+) + I(V_-,W_-;Y_{\rmr,-}|T_-) - I(V_-,W_-;Y_{1,\eps,-}|T_-) \\
&\leq I(W_+;Y_{\rmr,+}|T_+) + I(W_-;Y_{\rmr,-}|T_-).  
\end{align*}
Observe that
\begin{align*}
&   I(V_+,W_+;Y_{\rmr,+}|T_+) - I(V_+,W_+;Y_{1,\eps,+}|T_+) + I(V_-,W_-;Y_{\rmr,-}|T_-) - I(V_-,W_-;Y_{1,\eps,-}|T_-)\\
& = I(V,W,Y_{1,\eps,-};Y_{\rmr,+}|T) - I(V,W,Y_{1,\eps,-};Y_{1,\eps,+}|T) + I(V,W,Y_{\rmr,+};Y_{\rmr,-}|T) - I(V,W,Y_{\rmr,+};Y_{1,\eps,-}|T) \\
&= I(V,W;Y_{\rmr,+},Y_{\rmr,-}|T) - I(V,W;Y_{1,\eps,+},Y_{1,\eps,-}|T) - I(Y_{1,\eps,-};Y_{1,\eps,+}|T)\\
&\stackrel{(a)}{\leq} I(W;Y_{\rmr,+},Y_{\rmr,-}|T)- I(Y_{1,\eps,-};Y_{1,\eps,+}|T)\\
& \leq I(W;Y_{\rmr,+}|T) + I(W,Y_{\rmr,+};Y_{\rmr,-}|T)\\
& = I(W_+;Y_{\rmr,+}|T_+) + I(W_-;Y_{\rmr,-}|T_-),
\end{align*}
where (a) follows from the fact that random variables $(W_1, T_1)$ and $(W_2,T_2)$ of the two copies of the maximizer satisfy the constraints
\begin{align*}
  I(V_1,W_1;Y_{\rmr,1}|T_1)-I(V_1,W_1;Y_{1,\eps,1}|T_1)&\leq I(W_1;Y_{\mathrm{r},1}|T_1),\\
  I(V_2,W_2;Y_{\rmr,2}|T_2)-I(V_2,W_2;Y_{1,\eps,2}|T_2)&\leq I(W_2;Y_{\mathrm{r},2}|T_2).
\end{align*}
Thus the first part of the constraint is satisfied for $+$ and $-$ variables.
Finally observe that
\begin{align*}
   &I(W_+;Y_{\rmr,+}|T_+) + I(W_-;Y_{\rmr,-}|T_-)\\
   & \quad = I(W;Y_{\rmr,+}|T) + I(W,Y_{\rmr,+};Y_{\rmr,-}|T)\\
   & \quad = I(W;Y_{\rmr,+},Y_{\rmr,-}|T)
   \\&\quad\leq 2 C_0.
\end{align*}
This shows that the constraint is satisfied.
Now we denote $T=(T,Q)$ where $Q$ is uniform binary and conditioned on $Q=0$ we use the $+$ distribution and conditioned on $Q=-$ we use the - distribution. Thus, we get a new maximizer, which satisfies the constraint and yields a value at least as large as the original maximizer.

For the maximizing distribution \eqref{eqntg1} is tight and we must have the following
\begin{align*}
  I(Y_{\rmr,+};\delta Y_{\rmr,-} + Z_{3,-}|\delta Y_{\rmr,+} + Z_{3,+},W,T)&=0, \\
     I(Y_{1,\eps,-};\delta Y_{1,\eps,+} + Z_{3,+}|V,W,T,\delta Y_{1,\eps,-}+Z_{3,-})&=0.  
\end{align*}
Now using double-Markovity, we conclude that conditioned on $V,W,T$, the random variables  $X,Y_\rmr$ are jointly Gaussian for the maximizer and the covariance matrix of $X,Y_\rmr$ does not depend on the values of $V,W,T$. Moreover, and conditioned on $W,T$, the random variable $Y_\rmr$ is Gaussian for the maximizer. As observed earlier taking $X$ to be Gaussian conditioned on $T$ maximizes $h(Y_{1,\eps}|T)$ and hence keeps the $C_0$ constraint satisfied. This implies that we have a Gaussian maximizer.
\end{proof}}

\end{document}